\documentclass[journal,10pt,twoside,twocolumn]{IEEEtran}
\usepackage{cite}
\usepackage[T1]{fontenc}
\usepackage{graphicx}
\usepackage{amssymb}
\usepackage{amsmath}
\usepackage{paralist}
\usepackage{subfigure}
\usepackage{booktabs} 
\usepackage{algorithm}
\usepackage{algorithmic}
\usepackage{amsthm}
\usepackage{multirow}
\usepackage{microtype}
\usepackage{tablefootnote}
\usepackage{fancyref}
\usepackage{balance}

\usepackage{framed}
\usepackage{adjustbox}

\usepackage{amssymb}
\usepackage{amsfonts}
\usepackage{mathrsfs}
\usepackage{xspace}
\usepackage{bm}
\usepackage{upgreek}

\newcommand{\safemath}[2]{\newcommand{#1}{\ensuremath{#2}\xspace}}

\safemath{\bma}{\mathbf{a}}
\safemath{\bmb}{\mathbf{b}}
\safemath{\bmc}{\mathbf{c}}
\safemath{\bmd}{\mathbf{d}}
\safemath{\bme}{\mathbf{e}}
\safemath{\bmf}{\mathbf{f}}
\safemath{\bmg}{\mathbf{g}}
\safemath{\bmh}{\mathbf{h}}
\safemath{\bmi}{\mathbf{i}}
\safemath{\bmj}{\mathbf{j}}
\safemath{\bmk}{\mathbf{k}}
\safemath{\bml}{\mathbf{l}}
\safemath{\bmm}{\mathbf{m}}
\safemath{\bmn}{\mathbf{n}}
\safemath{\bmo}{\mathbf{o}}
\safemath{\bmp}{\mathbf{p}}
\safemath{\bmq}{\mathbf{q}}
\safemath{\bmr}{\mathbf{r}}
\safemath{\bms}{\mathbf{s}}
\safemath{\bmt}{\mathbf{t}}
\safemath{\bmu}{\mathbf{u}}
\safemath{\bmv}{\mathbf{v}}
\safemath{\bmw}{\mathbf{w}}
\safemath{\bmx}{\mathbf{x}}
\safemath{\bmy}{\mathbf{y}}
\safemath{\bmz}{\mathbf{z}}
\safemath{\bmzero}{\mathbf{0}}
\safemath{\bmone}{\mathbf{1}}

\bmdefine{\biad}{a}
\bmdefine{\bibd}{b}
\bmdefine{\bicd}{c}
\bmdefine{\bidd}{d}
\bmdefine{\bied}{e}
\bmdefine{\bifd}{f}
\bmdefine{\bigd}{g}
\bmdefine{\bihd}{h}
\bmdefine{\biid}{i}
\bmdefine{\bijd}{j}
\bmdefine{\bikd}{k}
\bmdefine{\bild}{l}
\bmdefine{\bimd}{m}
\bmdefine{\bind}{n}
\bmdefine{\biod}{o}
\bmdefine{\bipd}{p}
\bmdefine{\biqd}{q}
\bmdefine{\bird}{r}
\bmdefine{\bisd}{s}
\bmdefine{\bitd}{t}
\bmdefine{\biud}{u}
\bmdefine{\bivd}{v}
\bmdefine{\biwd}{w}
\bmdefine{\bixd}{x}
\bmdefine{\biyd}{y}
\bmdefine{\bizd}{z}

\bmdefine{\bixid}{\xi}
\bmdefine{\bilambdad}{\lambda}
\bmdefine{\bimud}{\mu}
\bmdefine{\bithetad}{\theta}
\bmdefine{\biphid}{\phi}
\bmdefine{\bideltad}{\delta}

\safemath{\bmia}{\biad}
\safemath{\bmib}{\bibd}
\safemath{\bmic}{\bicd}
\safemath{\bmid}{\bidd}
\safemath{\bmie}{\bied}
\safemath{\bmif}{\bifd}
\safemath{\bmig}{\bigd}
\safemath{\bmih}{\bihd}
\safemath{\bmii}{\biid}
\safemath{\bmij}{\bijd}
\safemath{\bmik}{\bikd}
\safemath{\bmil}{\bild}
\safemath{\bmim}{\bimd}
\safemath{\bmin}{\bind}
\safemath{\bmio}{\biod}
\safemath{\bmip}{\bipd}
\safemath{\bmiq}{\biqd}
\safemath{\bmir}{\bird}
\safemath{\bmis}{\bisd}
\safemath{\bmit}{\bitd}
\safemath{\bmiu}{\biud}
\safemath{\bmiv}{\bivd}
\safemath{\bmiw}{\biwd}
\safemath{\bmix}{\bixd}
\safemath{\bmiy}{\biyd}
\safemath{\bmiz}{\bizd}

\safemath{\bmxi}{\bixid}
\safemath{\bmlambda}{\bilambdad}
\safemath{\bmmu}{\bimud}
\safemath{\bmtheta}{\bithetad}
\safemath{\bmphi}{\biphid}
\safemath{\bmdelta}{\bideltad}

\safemath{\bA}{\mathbf{A}}
\safemath{\bB}{\mathbf{B}}
\safemath{\bC}{\mathbf{C}}
\safemath{\bD}{\mathbf{D}}
\safemath{\bE}{\mathbf{E}}
\safemath{\bF}{\mathbf{F}}
\safemath{\bG}{\mathbf{G}}
\safemath{\bH}{\mathbf{H}}
\safemath{\bI}{\mathbf{I}}
\safemath{\bJ}{\mathbf{J}}
\safemath{\bK}{\mathbf{K}}
\safemath{\bL}{\mathbf{L}}
\safemath{\bM}{\mathbf{M}}
\safemath{\bN}{\mathbf{N}}
\safemath{\bO}{\mathbf{O}}
\safemath{\bP}{\mathbf{P}}
\safemath{\bQ}{\mathbf{Q}}
\safemath{\bR}{\mathbf{R}}
\safemath{\bS}{\mathbf{S}}
\safemath{\bT}{\mathbf{T}}
\safemath{\bU}{\mathbf{U}}
\safemath{\bV}{\mathbf{V}}
\safemath{\bW}{\mathbf{W}}
\safemath{\bX}{\mathbf{X}}
\safemath{\bY}{\mathbf{Y}}
\safemath{\bZ}{\mathbf{Z}}

\safemath{\bZero}{\mathbf{0}}
\safemath{\bOne}{\mathbf{1}}
\safemath{\bDelta}{\mathbf{\Delta}}
\safemath{\bLambda}{\mathbf{\UpLambda}}
\safemath{\bPhi}{\mathbf{\Upphi}}
\safemath{\bSigma}{\mathbf{\Upsigma}}
\safemath{\bOmega}{\mathbf{\Upomega}}
\safemath{\bTheta}{\mathbf{\Uptheta}}

\bmdefine{\biAd}{A}
\bmdefine{\biBd}{B}
\bmdefine{\biCd}{C}
\bmdefine{\biDd}{D}
\bmdefine{\biEd}{E}
\bmdefine{\biFd}{F}
\bmdefine{\biGd}{G}
\bmdefine{\biHd}{H}
\bmdefine{\biId}{I}
\bmdefine{\biJd}{J}
\bmdefine{\biKd}{K}
\bmdefine{\biLd}{L}
\bmdefine{\biMd}{M}
\bmdefine{\biOd}{N}
\bmdefine{\biPd}{O}
\bmdefine{\biQd}{P}
\bmdefine{\biRd}{R}
\bmdefine{\biSd}{S}
\bmdefine{\biTd}{T}
\bmdefine{\biUd}{U}
\bmdefine{\biVd}{V}
\bmdefine{\biWd}{W}
\bmdefine{\biXd}{X}
\bmdefine{\biYd}{Y}
\bmdefine{\biZd}{Z}

\bmdefine{\biDelta}{\Delta}
\bmdefine{\biLambda}{\Lambda}
\bmdefine{\biPhi}{\Phi}
\bmdefine{\biSigma}{\Sigma}
\bmdefine{\biOmega}{\Omega}
\bmdefine{\biTheta}{\Theta}

\safemath{\bimA}{\biAd}
\safemath{\bimB}{\biBd}
\safemath{\bimC}{\biCd}
\safemath{\bimD}{\biDd}
\safemath{\bimE}{\biEd}
\safemath{\bimF}{\biFd}
\safemath{\bimG}{\biGd}
\safemath{\bimH}{\biHd}
\safemath{\bimI}{\biId}
\safemath{\bimJ}{\biJd}
\safemath{\bimK}{\biKd}
\safemath{\bimL}{\biLd}
\safemath{\bimM}{\biMd}
\safemath{\bimN}{\biNd}
\safemath{\bimO}{\biOd}
\safemath{\bimP}{\biPd}
\safemath{\bimQ}{\biQd}
\safemath{\bimR}{\biRd}
\safemath{\bimS}{\biSd}
\safemath{\bimT}{\biTd}
\safemath{\bimU}{\biUd}
\safemath{\bimV}{\biVd}
\safemath{\bimW}{\biWd}
\safemath{\bimX}{\biXd}
\safemath{\bimY}{\biYd}
\safemath{\bimZ}{\biZd}

\safemath{\bimDelta}{\biDelta}
\safemath{\bimLambda}{\biLambda}
\safemath{\bimPhi}{\biPhi}
\safemath{\bimSigma}{\biSigma}
\safemath{\bimOmega}{\biOmega}
\safemath{\bimTheta}{\biTheta}

\safemath{\setA}{\mathcal{A}}
\safemath{\setB}{\mathcal{B}}
\safemath{\setC}{\mathcal{C}}
\safemath{\setD}{\mathcal{D}}
\safemath{\setE}{\mathcal{E}}
\safemath{\setF}{\mathcal{F}}
\safemath{\setG}{\mathcal{G}}
\safemath{\setH}{\mathcal{H}}
\safemath{\setI}{\mathcal{I}}
\safemath{\setJ}{\mathcal{J}}
\safemath{\setK}{\mathcal{K}}
\safemath{\setL}{\mathcal{L}}
\safemath{\setM}{\mathcal{M}}
\safemath{\setN}{\mathcal{N}}
\safemath{\setO}{\mathcal{O}}
\safemath{\setP}{\mathcal{P}}
\safemath{\setQ}{\mathcal{Q}}
\safemath{\setR}{\mathcal{R}}
\safemath{\setS}{\mathcal{S}}
\safemath{\setT}{\mathcal{T}}
\safemath{\setU}{\mathcal{U}}
\safemath{\setV}{\mathcal{V}}
\safemath{\setW}{\mathcal{W}}
\safemath{\setX}{\mathcal{X}}
\safemath{\setY}{\mathcal{Y}}
\safemath{\setZ}{\mathcal{Z}}
\safemath{\emptySet}{\varnothing}

\safemath{\colA}{\mathscr{A}}
\safemath{\colB}{\mathscr{B}}
\safemath{\colC}{\mathscr{C}}
\safemath{\colD}{\mathscr{D}}
\safemath{\colE}{\mathscr{E}}
\safemath{\colF}{\mathscr{F}}
\safemath{\colG}{\mathscr{G}}
\safemath{\colH}{\mathscr{H}}
\safemath{\colI}{\mathscr{I}}
\safemath{\colJ}{\mathscr{J}}
\safemath{\colK}{\mathscr{K}}
\safemath{\colL}{\mathscr{L}}
\safemath{\colM}{\mathscr{M}}
\safemath{\colN}{\mathscr{N}}
\safemath{\colO}{\mathscr{O}}
\safemath{\colP}{\mathscr{P}}
\safemath{\colQ}{\mathscr{Q}}
\safemath{\colR}{\mathscr{R}}
\safemath{\colS}{\mathscr{S}}
\safemath{\colT}{\mathscr{T}}
\safemath{\colU}{\mathscr{U}}
\safemath{\colV}{\mathscr{V}}
\safemath{\colW}{\mathscr{W}}
\safemath{\colX}{\mathscr{X}}
\safemath{\colY}{\mathscr{Y}}
\safemath{\colZ}{\mathscr{Z}}

\safemath{\opA}{\mathbb{A}}
\safemath{\opB}{\mathbb{B}}
\safemath{\opC}{\mathbb{C}}
\safemath{\opD}{\mathbb{D}}
\safemath{\opE}{\mathbb{E}}
\safemath{\opF}{\mathbb{F}}
\safemath{\opG}{\mathbb{G}}
\safemath{\opH}{\mathbb{H}}
\safemath{\opI}{\mathbb{I}}
\safemath{\opJ}{\mathbb{J}}
\safemath{\opK}{\mathbb{K}}
\safemath{\opL}{\mathbb{L}}
\safemath{\opM}{\mathbb{M}}
\safemath{\opN}{\mathbb{N}}
\safemath{\opO}{\mathbb{O}}
\safemath{\opP}{\mathbb{P}}
\safemath{\opQ}{\mathbb{Q}}
\safemath{\opR}{\mathbb{R}}
\safemath{\opS}{\mathbb{S}}
\safemath{\opT}{\mathbb{T}}
\safemath{\opU}{\mathbb{U}}
\safemath{\opV}{\mathbb{V}}
\safemath{\opW}{\mathbb{W}}
\safemath{\opX}{\mathbb{X}}
\safemath{\opY}{\mathbb{Y}}
\safemath{\opZ}{\mathbb{Z}}
\safemath{\opZero}{\mathbb{O}}
\safemath{\identityop}{\opI}

\safemath{\veca}{\bma}
\safemath{\vecb}{\bmb}
\safemath{\vecc}{\bmc}
\safemath{\vecd}{\bmd}
\safemath{\vece}{\bme}
\safemath{\vecf}{\bmf}
\safemath{\vecg}{\bmg}
\safemath{\vech}{\bmh}
\safemath{\veci}{\bmi}
\safemath{\vecj}{\bmj}
\safemath{\veck}{\bmk}
\safemath{\vecl}{\bml}
\safemath{\vecm}{\bmm}
\safemath{\vecn}{\bmn}
\safemath{\veco}{\bmo}
\safemath{\vecp}{\bmp}
\safemath{\vecq}{\bmq}
\safemath{\vecr}{\bmr}
\safemath{\vecs}{\bms}
\safemath{\vect}{\bmt}
\safemath{\vecu}{\bmu}
\safemath{\vecv}{\bmv}
\safemath{\vecw}{\bmw}
\safemath{\vecx}{\bmx}
\safemath{\vecy}{\bmy}
\safemath{\vecz}{\bmz}

\safemath{\veczero}{\bmzero}
\safemath{\vecone}{\bmone}
\safemath{\vecxi}{\bmxi}
\safemath{\veclambda}{\bmlambda}
\safemath{\vecmu}{\bmmu}
\safemath{\vectheta}{\bmtheta}
\safemath{\vecphi}{\bmphi}
\safemath{\vecdelta}{\bmdelta}

\safemath{\matA}{\bA}
\safemath{\matB}{\bB}
\safemath{\matC}{\bC}
\safemath{\matD}{\bD}
\safemath{\matE}{\bE}
\safemath{\matF}{\bF}
\safemath{\matG}{\bG}
\safemath{\matH}{\bH}
\safemath{\matI}{\bI}
\safemath{\matJ}{\bJ}
\safemath{\matK}{\bK}
\safemath{\matL}{\bL}
\safemath{\matM}{\bM}
\safemath{\matN}{\bN}
\safemath{\matO}{\bO}
\safemath{\matP}{\bP}
\safemath{\matQ}{\bQ}
\safemath{\matR}{\bR}
\safemath{\matS}{\bS}
\safemath{\matT}{\bT}
\safemath{\matU}{\bU}
\safemath{\matV}{\bV}
\safemath{\matW}{\bW}
\safemath{\matX}{\bX}
\safemath{\matY}{\bY}
\safemath{\matZ}{\bZ}
\safemath{\matzero}{\bmzero}

\safemath{\matDelta}{\bDelta}
\safemath{\matLambda}{\bLambda}
\safemath{\matPhi}{\bPhi}
\safemath{\matSigma}{\bSigma}
\safemath{\matOmega}{\bOmega}
\safemath{\matTheta}{\bTheta}

\safemath{\matidentity}{\matI}
\safemath{\matone}{\matO}

\safemath{\rnda}{A}
\safemath{\rndb}{B}
\safemath{\rndc}{C}
\safemath{\rndd}{D}
\safemath{\rnde}{E}
\safemath{\rndf}{F}
\safemath{\rndg}{G}
\safemath{\rndh}{H}
\safemath{\rndi}{I}
\safemath{\rndj}{J}
\safemath{\rndk}{K}
\safemath{\rndl}{L}
\safemath{\rndm}{M}
\safemath{\rndn}{N}
\safemath{\rndo}{O}
\safemath{\rndp}{P}
\safemath{\rndq}{Q}
\safemath{\rndr}{R}
\safemath{\rnds}{S}
\safemath{\rndt}{T}
\safemath{\rndu}{U}
\safemath{\rndv}{V}
\safemath{\rndw}{W}
\safemath{\rndx}{X}
\safemath{\rndy}{Y}
\safemath{\rndz}{Z}

\safemath{\rveca}{\bimA}
\safemath{\rvecb}{\bimB}
\safemath{\rvecc}{\bimC}
\safemath{\rvecd}{\bimD}
\safemath{\rvece}{\bimE}
\safemath{\rvecf}{\bimF}
\safemath{\rvecg}{\bimG}
\safemath{\rvech}{\bimH}
\safemath{\rveci}{\bimI}
\safemath{\rvecj}{\bimJ}
\safemath{\rveck}{\bimK}
\safemath{\rvecl}{\bimL}
\safemath{\rvecm}{\bimM}
\safemath{\rvecn}{\bimN}
\safemath{\rveco}{\bomO}
\safemath{\rvecp}{\bimP}
\safemath{\rvecq}{\bimQ}
\safemath{\rvecr}{\bimR}
\safemath{\rvecs}{\bimS}
\safemath{\rvect}{\bimT}
\safemath{\rvecu}{\bimU}
\safemath{\rvecv}{\bimV}
\safemath{\rvecw}{\bimW}
\safemath{\rvecx}{\bimX}
\safemath{\rvecy}{\bimY}
\safemath{\rvecz}{\bimZ}

\safemath{\rvecxi}{\bmxi}
\safemath{\rveclambda}{\bmlambda}
\safemath{\rvecmu}{\bmmu}
\safemath{\rvectheta}{\bmtheta}
\safemath{\rvecphi}{\bmphi}

\safemath{\rmatA}{\bimA}
\safemath{\rmatB}{\bimB}
\safemath{\rmatC}{\bimC}
\safemath{\rmatD}{\bimD}
\safemath{\rmatE}{\bimE}
\safemath{\rmatF}{\bimF}
\safemath{\rmatG}{\bimG}
\safemath{\rmatH}{\bimH}
\safemath{\rmatI}{\bimI}
\safemath{\rmatJ}{\bimJ}
\safemath{\rmatK}{\bimK}
\safemath{\rmatL}{\bimL}
\safemath{\rmatM}{\bimM}
\safemath{\rmatN}{\bimN}
\safemath{\rmatO}{\bimO}
\safemath{\rmatP}{\bimP}
\safemath{\rmatQ}{\bimQ}
\safemath{\rmatR}{\bimR}
\safemath{\rmatS}{\bimS}
\safemath{\rmatT}{\bimT}
\safemath{\rmatU}{\bimU}
\safemath{\rmatV}{\bimV}
\safemath{\rmatW}{\bimW}
\safemath{\rmatX}{\bimX}
\safemath{\rmatY}{\bimY}
\safemath{\rmatZ}{\bimZ}

\safemath{\rmatDelta}{\bimDelta}
\safemath{\rmatLambda}{\bimLambda}
\safemath{\rmatPhi}{\bimPhi}
\safemath{\rmatSigma}{\bimSigma}
\safemath{\rmatOmega}{\bimOmega}
\safemath{\rmatTheta}{\bimTheta}
 
\usepackage{amssymb}
\usepackage{amsfonts}
\usepackage{mathrsfs}
\usepackage{xspace}
\usepackage{bm}
\usepackage{fancyref}
\usepackage{textcomp}

\usepackage{multirow}
\usepackage{stmaryrd}

\newenvironment{textbmatrix}{	\setlength{\arraycolsep}{2.5pt}%
								\big[\begin{matrix}}{\end{matrix}\big]%
								\raisebox{0.08ex}{\vphantom{M}}}

\def\be{\begin{equation}}
\def\ee{\end{equation}}
\def\een{\nonumber \end{equation}}
\def\mat{\begin{bmatrix}}
\def\emat{\end{bmatrix}}
\def\btm{\begin{textbmatrix}}
\def\etm{\end{textbmatrix}}

\def\ba#1\ea{\begin{align}#1\end{align}}
\def\bas#1\eas{\begin{align*}#1\end{align*}}
\def\bs#1\es{\begin{split}#1\end{split}}
\def\bg#1\eg{\begin{gather}#1\end{gather}}
\def\bml#1\eml{\begin{multline}#1\end{multline}}
\def\bi#1\ei{\begin{itemize}#1\end{itemize}}

\newcommand{\lefto}{\mathopen{}\left}

\DeclareMathOperator{\sign}{sign}			
\DeclareMathOperator*{\argmin}{arg\;min}		
\DeclareMathOperator*{\argmax}{arg\;max}		
\DeclareMathOperator{\Exop}{\opE}			
\DeclareMathOperator{\Varop}{\opV\!\mathrm{ar}} 

\newcommand{\abs}[1]{\lefto\lvert#1\right\rvert}		

\newcommand{\est}[1]{\ensuremath{\hat{#1}}} 	

\safemath{\dirac}{\delta}					
\safemath{\krond}{\dirac}					

\safemath{\upto}{\uparrow}
\safemath{\downto}{\downarrow}
\safemath{\iu}{j}							
\safemath{\ev}{\lambda}						
\safemath{\hilseqspace}{l^{2}}				
\newcommand{\banachfunspace}[1]{\setL^{#1}}	
\safemath{\hilfunspace}{\banachfunspace{2}}	

\safemath{\SNR}{\textit{SNR}} 				
\safemath{\PAR}{\textit{PAR}} 				
\safemath{\No}{N_0}							
\safemath{\Es}{E_s}							
\safemath{\Eb}{E_b}							
\safemath{\EbNo}{\frac{\Eb}{\No}}
\safemath{\EsNo}{\frac{\Es}{\No}}

\DeclareMathOperator{\CHop}{\ensuremath{\opH}} 
\safemath{\tvir}{\rndh_{\CHop}}				
\safemath{\tvtf}{\rndl_{\CHop}}				
\safemath{\spf}{\rnds_{\CHop}}				
\safemath{\bff}{H_{\CHop}}					

\safemath{\ircf}{r_{h}}						
\safemath{\tftvcf}{r_{s}}					
\safemath{\tfcf}{r_{l}}						
\safemath{\bfcf}{r_{H}}						

\safemath{\tcorr}{c_h}						
\safemath{\scf}{c_{s}}						
\safemath{\tfcorr}{c_{l}}					
\safemath{\fcorr}{c_{H}}						

\safemath{\mi}{I}							
\safemath{\capacity}{C}						

\safemath{\normal}{\mathcal{N}}			
\safemath{\jpg}{\mathcal{CN}}			
\safemath{\mchain}{\leftrightarrow}		

\safemath{\dB}{\,\mathrm{dB}}
\safemath{\dBm}{\,\mathrm{dBm}}
\safemath{\Hz}{\,\mathrm{Hz}}
\safemath{\kHz}{\,\mathrm{kHz}}
\safemath{\MHz}{\,\mathrm{MHz}}
\safemath{\GHz}{\,\mathrm{GHz}}
\safemath{\s}{\,\mathrm{s}}
\safemath{\ms}{\,\mathrm{ms}}
\safemath{\mus}{\,\mathrm{\text{\textmu}s}}
\safemath{\ns}{\,\mathrm{ns}}
\safemath{\ps}{\,\mathrm{ps}}
\safemath{\meter}{\,\mathrm{m}}
\safemath{\mm}{\,\mathrm{mm}}
\safemath{\cm}{\,\mathrm{cm}}
\safemath{\m}{\,\mathrm{m}}
\safemath{\W}{\,\mathrm{W}}
\safemath{\mW}{\, \mathrm{mW}}
\safemath{\J}{\,\mathrm{J}}
\safemath{\K}{\,\mathrm{K}}
\safemath{\bit}{\,\mathrm{bit}}
\safemath{\nat}{\,\mathrm{nat}}

\safemath{\define}{\triangleq}			

\safemath{\equivalent}{\sim}
\safemath{\distas}{\sim}					
\safemath{\sdiff}{\Delta}				

\safemath{\reals}{\mathbb{R}}
\safemath{\positivereals}{\reals_{+}}
\safemath{\integers}{\mathbb{Z}}
\safemath{\posint}{\integers_{+}}
\safemath{\naturals}{\mathbb{N}}
\safemath{\posnaturals}{\naturals_{+}}
\safemath{\complexset}{\mathbb{C}}
\safemath{\rationals}{\mathbb{Q}}

\newcommand*{\fancyrefapplabelprefix}{app}		
\newcommand*{\fancyrefthmlabelprefix}{thm}		
\newcommand*{\fancyreflemlabelprefix}{lem}		
\newcommand*{\fancyrefcorlabelprefix}{cor}		
\newcommand*{\fancyrefdeflabelprefix}{def}		
\newcommand*{\fancyrefproplabelprefix}{prop}		
\newcommand*{\fancyrefexmpllabelprefix}{exmpl}
\newcommand*{\fancyrefalglabelprefix}{alg}		
\newcommand*{\fancyreftbllabelprefix}{tbl}		
\newcommand*{\fancyrefremlabelprefix}{rem}		

\frefformat{vario}{\fancyrefseclabelprefix}{Section~#1}
\frefformat{vario}{\fancyrefthmlabelprefix}{Theorem~#1}
\frefformat{vario}{\fancyreftbllabelprefix}{Table~#1}
\frefformat{vario}{\fancyreflemlabelprefix}{Lemma~#1}
\frefformat{vario}{\fancyrefcorlabelprefix}{Corollary~#1}
\frefformat{vario}{\fancyrefdeflabelprefix}{Definition~#1}
\frefformat{vario}{\fancyreffiglabelprefix}{Fig.~#1}
\frefformat{vario}{\fancyrefapplabelprefix}{Appendix~#1}
\frefformat{vario}{\fancyrefeqlabelprefix}{(#1)}
\frefformat{vario}{\fancyrefproplabelprefix}{Proposition~#1}
\frefformat{vario}{\fancyrefexmpllabelprefix}{Example~#1}
\frefformat{vario}{\fancyrefalglabelprefix}{Algorithm~#1}
\frefformat{vario}{\fancyrefremlabelprefix}{Remark~#1}

\newtheorem{thm}{Theorem}

\newtheorem{defi}{Definition}
\newtheorem{lem}[thm]{Lemma}

 \newtheorem{rem}{Remark}

\safemath{\dictab}{[\,\dicta\,\,\dictb\,]}

\safemath{\ysig}{\bmy}
\safemath{\ysighat}{\hat{\ysig}}
\safemath{\ysigdim}{M}
\safemath{\xsig}{\bmx}
\safemath{\xsigdim}{N}
\safemath{\nx}{n_x}
\safemath{\zsig}{\bmz}
\safemath{\zsigdim}{\ysigdim}
\safemath{\rsig}{\bmr}
\safemath{\Adict}{\bA}
\safemath{\Adicttilde}{\widetilde{\Adict}}
\safemath{\Adictdim}{\outputdim\times\xsigdim}
\safemath{\avec}{\bma}
\safemath{\avectilde}{\tilde{\avec}}
\safemath{\Bdict}{\bB}
\safemath{\Bdicttilde}{\widetilde{\Bdict}}
\safemath{\Cdict}{\bC}
\safemath{\cvec}{\bmc}
\safemath{\Ddict}{\bD}
\safemath{\Ddictdim}{\ysigdim\times\xsigdim}
\safemath{\dvec}{\bmd}
\safemath{\Ddicttilde}{\widetilde{\bD}}
\safemath{\Bonb}{\bB}
\safemath{\bvec}{\bmb}
\safemath{\Bonbdim}{\ysigdim\times\ysigdim}
\safemath{\noise}{\bmn}
\safemath{\noisedim}{\ysigim}
\safemath{\err}{\bme}
\safemath{\errdim}{\ysigdim}
\safemath{\errset}{\setE}
\safemath{\nerr}{n_e}
\safemath{\delop}{\bP_\errset}
\safemath{\delopc}{\bP_{{\errset}^c}}

\safemath{\cplxi}{\imath}
\safemath{\cplxj}{\jmath}

\safemath{\dict}{\matD}
\safemath{\inputdim}{N}		
\safemath{\outputdim}{M}		
\safemath{\sparsity}{S}	
\safemath{\inputdimA}{{N_a}}	
\safemath{\inputdimB}{{N_b}}	
\safemath{\elemA}{{n_a}}	
\safemath{\elemB}{{n_b}}	
\safemath{\resA}{\matR_a}	
\safemath{\resB}{\matR_b}	
\safemath{\subD}{\matS} 
\safemath{\subA}{\matS_a} 
\safemath{\subB}{\matS_b} 
\safemath{\dicta}{\matA} 	
\safemath{\dictb}{\matB} 	
\safemath{\hollowS}{H}
\safemath{\hollowA}{H_a}
\safemath{\hollowB}{H_b}
\safemath{\cross}{Z}
\safemath{\coh}{\mu_d}			
\safemath{\coha}{\mu_a}			
\safemath{\cohb}{\mu_b}			
\safemath{\mubs}{\nu}	
\safemath{\cohm}{\mu_m} 
\safemath{\dictset}{\setD}	
\safemath{\dictsetp}{\dictset(\coh,\coha,\cohb)}	
\safemath{\dictsetgen}{\dictset_\text{gen}}
\safemath{\dictsetgenp}{\dictsetgen(\coh)}
\safemath{\dictsetonb}{\dictset_\text{onb}}
\safemath{\dictsetonbp}{\dictsetonb(\coh)}

\safemath{\leftside}{U}
\safemath{\rightsideA}{R_a}
\safemath{\rightsideB}{R_b}

\safemath{\indexS}{\setI_S} 

\safemath{\na}{n_a}			
\safemath{\nb}{n_b}			
\safemath{\coeffa}{p_i}	
\safemath{\coeffb}{q_j}	
\safemath{\seta}{\setP}		
\safemath{\setb}{\setQ}     
\safemath{\setw}{\setW}	
\safemath{\setz}{\setZ}	
\safemath{\cola}{\veca}		
\safemath{\colb}{\vecb}		
\safemath{\cold}{\vecd}		
\safemath{\inputvec}{\vecx} 	
\safemath{\error}{\vece}	
\safemath{\noiseout}{\vecz} 	
\safemath{\inputvecel}{x}
\safemath{\inputveca}{\vecx_a}
\safemath{\inputvecb}{\vecx_b}
\safemath{\outputvec}{\vecy}	
\safemath{\lambdamin}{\lambda_{\mathrm{min}}}

\safemath{\elltwo}{\ell_2}
\safemath{\ellone}{\ell_1}
\safemath{\ellzero}{\ell_0}
\safemath{\ellinf}{\ell_\infty}
\safemath{\ellinftilde}{\ell_{\widetilde\infty}}
\safemath{\licard}{Z(\coh,\coha,\cohb)}
\safemath{\xsol}{\hat{x}}
\safemath{\xbord}{x_b}		
\safemath{\xstat}{x_s}		
\safemath{\xstatLone}{\tilde{x}_s}
\safemath{\order}{\mathcal{O}} 
\safemath{\scales}{\Theta} 
\safemath{\ones}{\mathbf{1}} 
\safemath{\zeroes}{\mathbf{0}} 
\safemath{\thlone}{\kappa(\coh,\cohb)} 
\safemath{\constoneA}{\delta} 
\safemath{\constoneB}{\epsilon} 
\safemath{\nlarge}{L}				   
\safemath{\sumlarge}{S_\nlarge}
\safemath{\maxlarger}{P_\nlarge}	   
\safemath{\Pzero}{\textrm{P0}}	
\safemath{\Pone}{\textrm{P1}}
\safemath{\vecfir}{\vecw}			 
\safemath{\vecsec}{\vecz}
\safemath{\elvecfir}{w}              
\safemath{\elvecsec}{z}				 
\safemath{\nlargefir}{n}
\safemath{\normout}{\gamma}
\safemath{\auxfun}{h}
\safemath{\supp}{\textrm{supp}}

\safemath{\indexa}{\ell}
\safemath{\indexb}{r}
\safemath{\indexc}{i}
\safemath{\indexd}{j}

\safemath{\project}{P}

\usepackage[usenames,dvipsnames]{color}

\allowdisplaybreaks 

\newtheorem{alg}{Algorithm} 
	
\newcommand{\xrv}[1]{X#1}

\newcommand{\resid}[1]{\bmr#1}

\safemath{\NT}{N_\textnormal{T}}
\safemath{\tmax}{{t_\textnormal{max}}}
\safemath{\LAMA}{\textrm{LAMA}}
\safemath{\MRT}{\textrm{MRT}}
\safemath{\betamax}{\beta^\textnormal{max}}
\safemath{\betamin}{\beta^\textnormal{min}}

\safemath{\Nomin}{\No^\textnormal{min}(\beta)}
\safemath{\Nomax}{\No^\textnormal{max}(\beta)}

\safemath{\MAP}{\textnormal{MAP}}
\safemath{\IO}{\textnormal{IO}}
\safemath{\Opt}{\textnormal{O}}
\safemath{\JO}{\textnormal{JO}}
\safemath{\Nopost}{N_{0}^\textnormal{post}}
\safemath{\MT}{N}
\safemath{\MR}{M}
\safemath{\Tran}{\textnormal{T}}
\safemath{\Herm}{\textnormal{H}}
\safemath{\row}{\textnormal{r}}
\safemath{\col}{\textnormal{c}}

\safemath{\dd}{\textnormal{d}}

\begin{document}

\title{Optimal Data Detection and Signal Estimation \\  in Systems with Input Noise}

\author{Ramina Ghods, Charles Jeon, Arian Maleki, and Christoph Studer
\thanks{R.~Ghods, C. Jeon and C.~Studer were with the School of ECE, Cornell University, Ithaca, NY; e-mail:  RG is now with Carnegie Mellon University, Pittsburgh, PA; {rghods@cs.cmu.edu}; CJ is with Apple Inc., San Diego, CA; CS is with ETH Zurich, Zurich, Switzerland; {studer@ethz.ch}.}
\thanks{A. Maleki is with Department of Statistics at Columbia University, New York City, NY; {arian@stat.columbia.edu}.}
\thanks{Part of this paper on massive MIMO detection has been presented at the 53rd Annual Allerton Conference on Communication, Control, and Computing~\cite{GJMS_2015_conf}. The present paper provides algorithm details and theoretical results for AMPI, and extends the method to sparse signal recovery in compressive sensing. }
}
\maketitle

\begin{abstract}
Practical systems often suffer from hardware impairments that already appear during signal generation.
Despite the limiting effect of such input-noise impairments on signal processing systems, they are routinely ignored in the literature. In this paper, we propose an algorithm for data detection and signal estimation, referred to as Approximate Message Passing  with Input noise (AMPI), which takes into account  input-noise impairments. 
To demonstrate the efficacy of AMPI, we investigate two applications: Data detection in large multiple-input multiple-output (MIMO) wireless systems and sparse signal recovery in compressive sensing. 
For both applications, we provide precise conditions in the large-system limit for which AMPI achieves optimal performance.
We furthermore use simulations to demonstrate that  AMPI achieves near-optimal performance at low complexity in realistic, finite-dimensional systems.
\end{abstract}




\section{Introduction}
\label{sec:intro}
We consider a general class of data detection and signal estimation problems in a noisy linear channel affected by input noise. As illustrated in \fref{fig:sysmodel}, we are interested in recovering the $N$-dimensional \emph{input signal} $\vecs\in\complexset^N$ observed from the following model. The input signal $\vecs$ with prior distribution $p(\vecs)=\prod_{\ell=1}^{N}p(s_\ell)$ is affected by input-noise characterized by the statistical relation $p(\bmx|\bms)=\prod_{\ell=1}^{N}p(x_\ell|s_\ell)$. 
The generality of this  input-noise model captures a wide range of hardware and system impairments, including hardware non-idealities that exhibit statistical dependence between impairments and the input signal (e.g., phase noise) as well as deterministic effects (e.g., non-linearities).
The  impaired signal $\bmx\in\complexset^N$, which we refer to as the \emph{effective input signal}, is then passed to a noisy linear transform that is modeled as 
\begin{align}\label{eq:sysmodel}
\bmy = \bH\bmx + \bmn.
\end{align}
Here, the vector $\vecy\in\complexset^M$ is the \emph{measured signal} and $M$ denotes the number of measurements, the \emph{system matrix} $\bH\in\complexset^{M\times N}$ represents the measurement process, and the vector $\bmn\in\complexset^M$ models measurement noise.
We assume that the entries of the noise vector $\bmn$ are i.i.d.\ circularly-symmetric complex Gaussian with variance~$\No$.
In what follows, we make use of the \emph{system ratio} defined as $\beta={N}/{M}$ and the following definitions:
\begin{defi} 
	We define the \emph{large system limit} by fixing the system ratio $\beta=\MT/\MR$ and by letting $\MT\to\infty$.
\end{defi}
\begin{defi} 
	A matrix $\bH$ describes  \emph{uniform linear measurements} if the entries of $\bH$ are i.i.d.\ circularly-symmetric complex Gaussian with variance $1/\MR$, i.e., $H_{i,j}\sim\setC\setN(0,1/M)$. 
\end{defi}

\begin{figure}[tp]
\centering
\includegraphics[width=0.8\columnwidth]{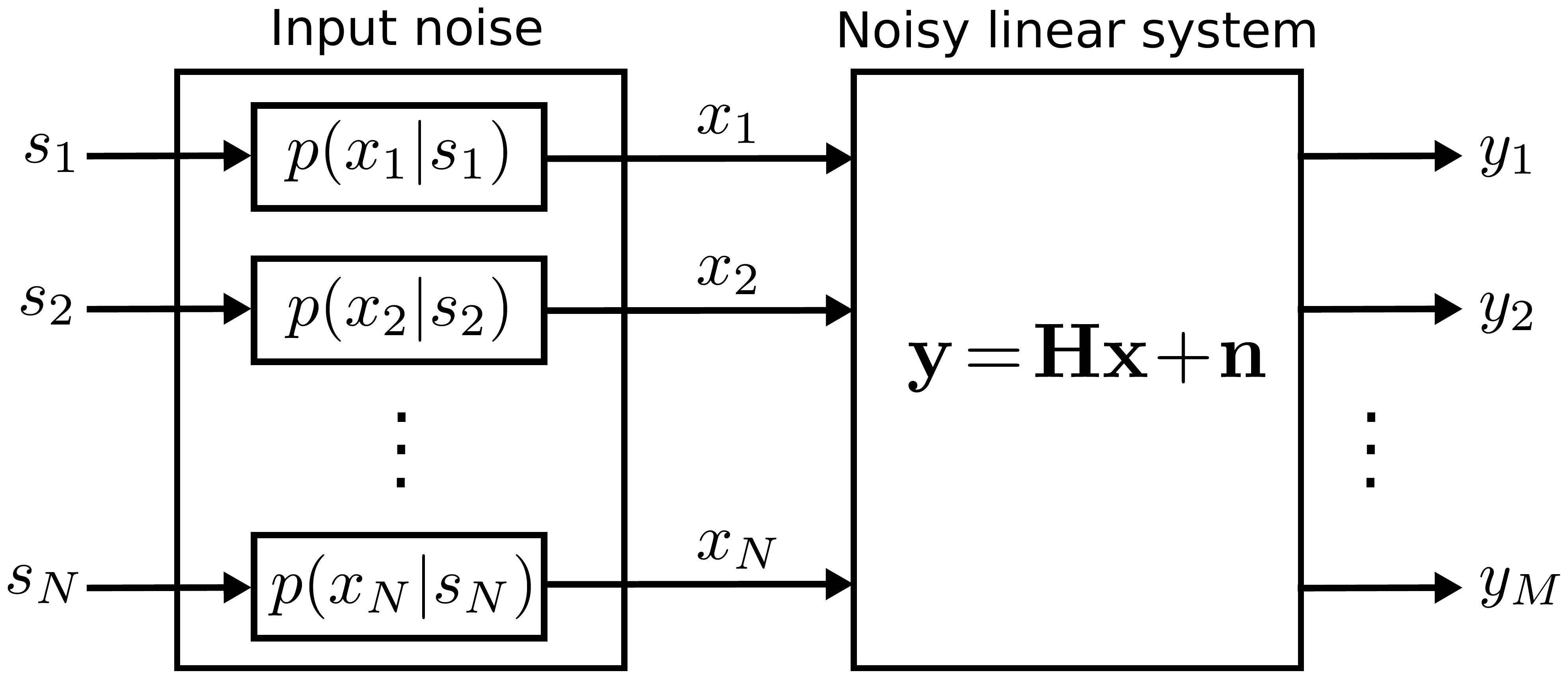}
\vspace{-0.1cm}
\caption{Illustration of a noisy linear system affected by input noise. The input signal $\bms$ is corrupted by input noise, resulting in the effective input signal~$\bmx$ that is observed through a noisy linear system. The goal is to recover the input signal~$\bms$ from the noisy observations $\bmy$.}
\label{fig:sysmodel}
\end{figure}
\subsection{Two Application Examples}
While numerous real-world systems suffer from input noise, we focus on two prominent scenarios.

\subsubsection{Massive MIMO Data Detection}
Massive multiple-input multiple-output (MIMO) is one of the core technologies in fifth-generation (5G) wireless systems~\cite{ABCHLAZ2014}. The idea is to equip the infrastructure base-stations with hundreds of antenna elements while simultaneously serving a smaller number of users.
One critical challenge in the realization of this technology is the computational complexity of data detection at the base-station~\cite{WBVSCJD2013}. While recent results have shown that the large dimensionality of massive MIMO can be exploited to  design near-optimal data-detection algorithms~\cite{JGMS2015,JGMS2015conf,wu2014low} using approximate message passing (AMP)~\cite{donoho2009message}, these methods ignore the fact that realistic communication systems are affected by impairments that already arise at the transmitter~\cite{Studer_Tx_OFDM,Tx_Replica}. 
In this paper, we introduce AMPI (short for AMP with input noise), which mitigates transmit-side RF impairments during data detection.
AMPI outperforms existing data-detection methods, e.g.,~\cite{JGMS2015conf}, that ignore the presence of input-noise impairments at virtually no additional computational complexity.

\fref{fig:MIMO} illustrates the symbol error-rate (SER) performance of AMPI in a symmetric massive MIMO system with $128$ user equipments (UEs) transmitting QPSK and $128$ base-station (BS) antennas. As in~\cite{Studer_Tx_OFDM,schenk2008rf,Tx_Replica},  the input noise is modeled as complex Gaussian noise.
We see that AMPI outperforms the LAMA algorithm~\cite{JGMS2015,JGMS2015conf}, which achieves---under certain conditions on the MIMO system---the error-rate performance of the individually-optimal (IO) data detector in absence of input noise. 
AMPI entails virtually no complexity increase over LAMA and  achieves comparable SER to whitening-based methods~\cite{Studer_Tx_OFDM}, which are optimal for Gaussian input noise but result in prohibitively high complexity in massive MIMO.

\begin{figure}[tp]
\centering
\subfigure[Symbol-error rate (SER) versus average SNR in a $128$ user equipment (UE), $128$ base-station antenna massive MIMO system with QPSK.]{\includegraphics[width=0.47\columnwidth]{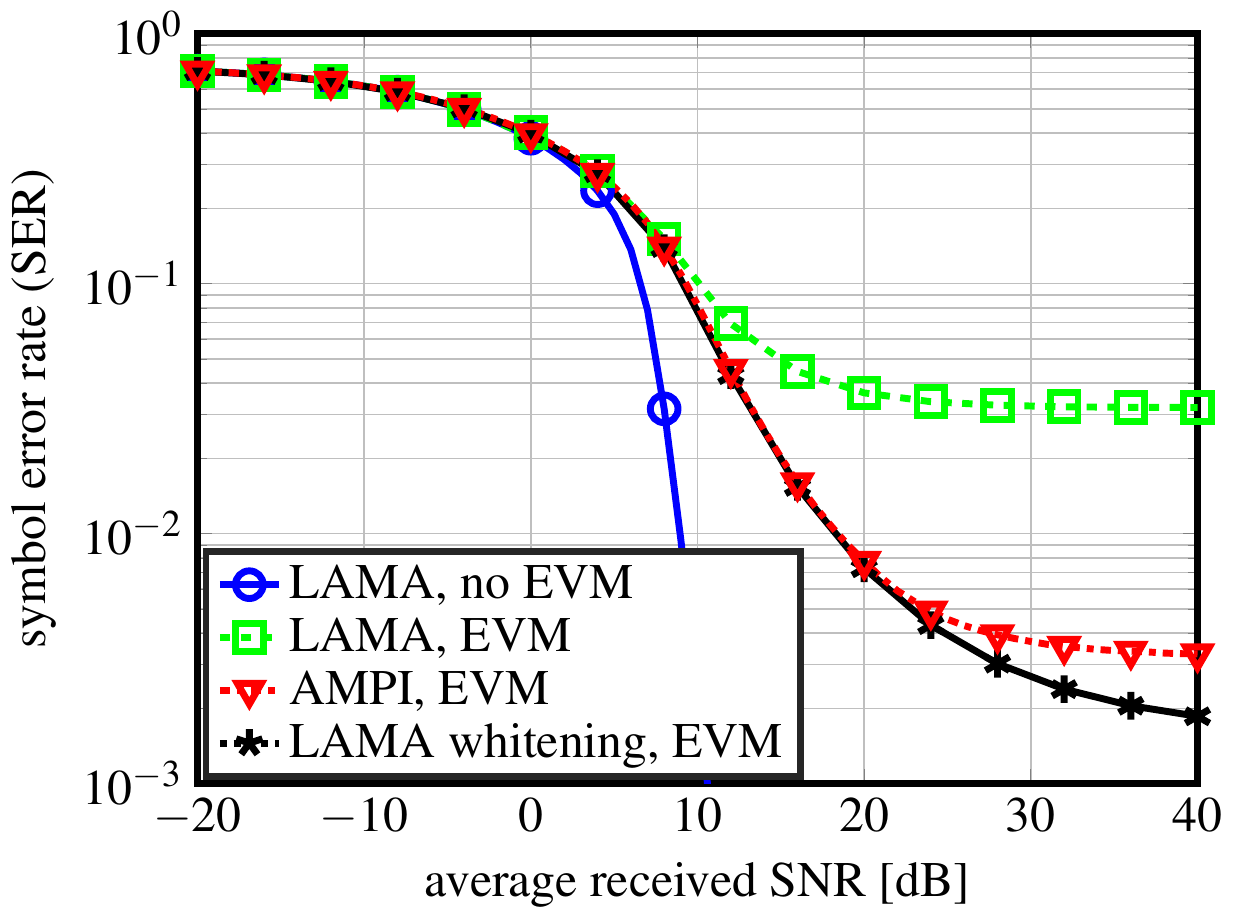}\label{fig:MIMO}}
\hspace{0.2cm}
\subfigure[Reconstruction SNR of a sparse signal recovery task for a  $5\%$ sparse signal of dimension $N=1000$ and 20\,dB SNR.]{\includegraphics[width=0.47\columnwidth]{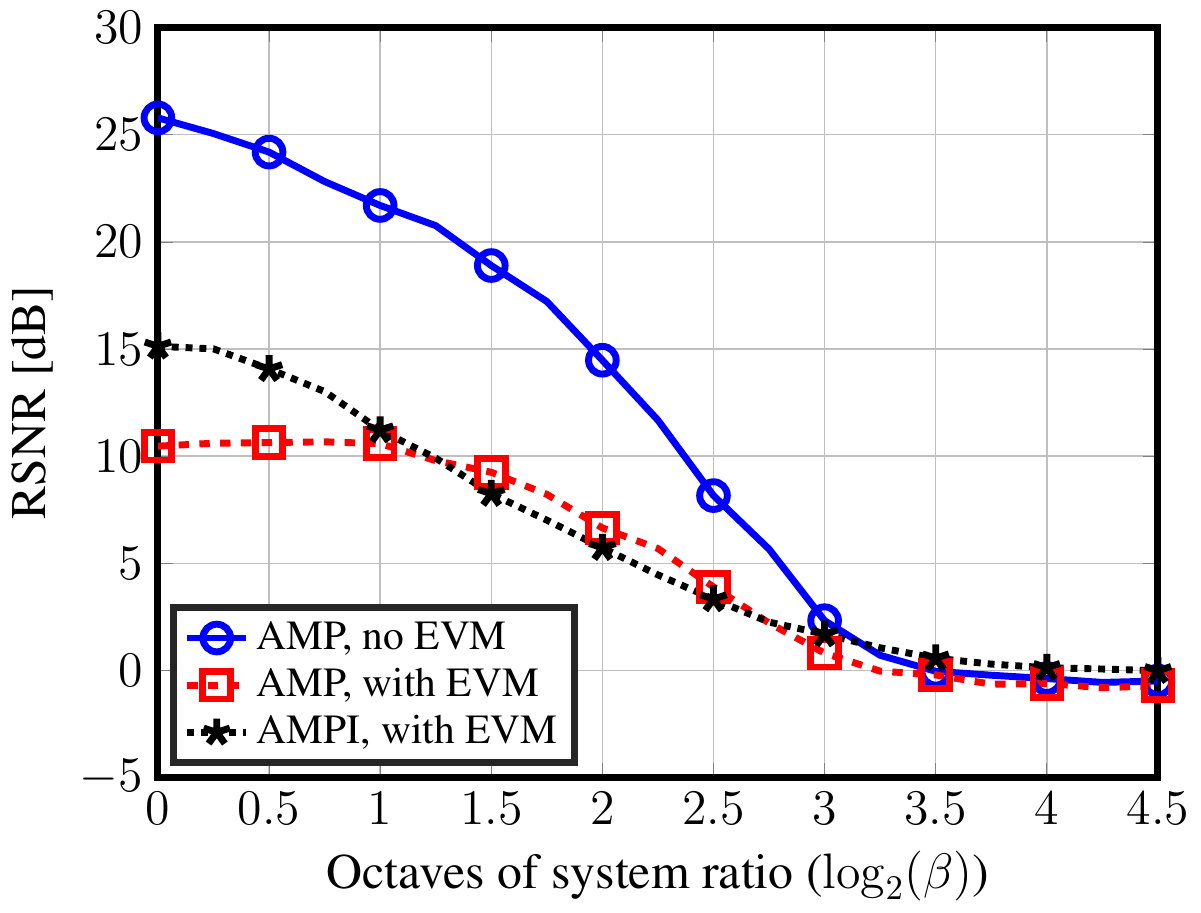}\label{fig:NF}}
\caption{Simulation results of two applications of the proposed AMPI algorithm in massive MIMO and compressive sensing with $\textit{EVM}=-10$\,dB input noise. AMPI yields significant improvements compared to methods that ignore input noise and achieves comparable performance to whitening-based methods that entail prohibitive complexity for the considered system dimensions.}
\end{figure}

\subsubsection{Compressive Sensing Signal Recovery} 

Compressive sensing (CS) enables  sampling and recovery of sparse signals at sub-Nyquist rates~\cite{donoho2006,candes2006c}. 
While the CS  literature extensively focuses on systems with measurement noise, numerous practical applications already contain noise on the sparse signal to be recovered; see~\cite{davenport2012proscons,treichler2009application,arias2011noise} and the references therein.
We will use AMPI to take input-noise into account directly during sparse signal recovery and show substantial improvements compared to that of existing sparse recovery methods for systems with input noise~\cite{arias2011noise} at no additional expense in complexity.

\fref{fig:NF}  shows the recovery signal-to-noise-ratio (RSNR) of compressive sensing scenario in which we recover a $N=1000$ dimensional sparse signal in the presence of input noise with an error-vector magnitude (EVM) of $-10$\,dB. The RSNR is plotted for different system ratios~$\beta$. The signal is assumed to have $5\%$ sparsity and an SNR of $20$\,dB; the non-zero entries are  i.i.d.\ standard normal.  AMPI significantly improves the reconstruction SNR over AMP which ignores input noise.

\subsection{Contributions}\label{sec:contributions}

We propose AMPI and provide precise conditions for which it yields optimal performance in data detection and signal estimation applications. 
Depending on the application, we use the following optimality criteria. For massive MIMO data detection, optimality is  achieving the same error-rate performance as the individually-optimal (IO) data detector~\cite{JGMS2015conf,V1998}, which solves the following minimization problems:
\begin{align}\label{eq:IOproblem}
\hat{s}_\ell^\IO=\argmin_{\tilde s_\ell \in \setO} \, \mathbb{P} \!\left( \tilde s_\ell \neq s_\ell \right)\!, \quad \ell=1,\ldots,N.
\end{align}
Here, $\setO$ is a finite set containing possible transmit symbols---in wireless systems this set corresponds to the transmit constellation, e.g., quadrature amplitude modulation (QAM).
For signal estimation, optimality is defined by minimizing the mean-squared error (MSE) 
\begin{align}\label{eq:MMSEproblem}
\hat{\bms}^\Opt=\argmin_{\tilde \bms \in \complexset^N} \, \textstyle  \frac{1}{N} \left\| \tilde \bms - \bms \right\|^2\!.
\end{align}
Here, the superscript $\Opt$ in $\hat{s}_\ell^\Opt$, $\ell=1,\ldots,N$, stands for optimal. To solve~\fref{eq:IOproblem} or~\fref{eq:MMSEproblem}, we need to compute the MAP or minimum MSE (MMSE) estimate of the marginal posterior distribution $p(s_\ell|\bmy,\bH)$ for all $\ell=1,\dots,N$. Computing the marginal distribution for large-dimensional systems is one of the key challenges in data detection and signal estimation problems as its requires prohibitive complexity \cite{cooper1990computational}.
We propose AMPI, which achieves optimal performance in the large-system limit and for uniform linear measurements.
Our optimality conditions are derived via the state-evolution (SE) framework \cite{donoho2009,andreaGMCS} of approximate message passing (AMP)~\cite{Maleki2010phd,DMM10a,DMM10b}.
 For both applications, we demonstrate  the efficacy  and low-complexity of AMPI in more realistic, finite-dimensional systems.




\subsection{Related Results}
The effect of input-noise (often called transmit-side impairments) on the performance of communication systems has been studied in \cite{schenk2008rf,Studer_Tx_OFDM,Tx_Replica,studer2011system,schenk2008rf,schenk2005performance,goransson2008effect,suzuki2008transmitter,suzuki2009practical,gonzalez2011impact,gonzalez2011transmit,bjornson2013capacity,zhang2014mimo}.
%
%
%
Most of these papers use a Gaussian input-noise model, which assumes that the input noise is i.i.d.\ additive Gaussian noise and independent of the transmit signal $\vecs$.
While the accuracy of this model has been confirmed via real-world measurements \cite{Studer_Tx_OFDM} for MIMO systems that use orthogonal frequency-division multiplexing (OFDM), it may not be accurate in other scenarios. 
AMPI is a practical data detection method that allows us to study the fundamental performance of more general input-noise models, which may exhibit statistical dependence with the transmit signal and even include deterministic nonlinearities.
For the well-established Gaussian transmit-noise model, we will show in \fref{sec:Opt} that the SE equations of AMPI coincide to the ``coupled fixed point equations'' provided in \cite{Tx_Replica}, which demonstrates that AMPI is a practical algorithm that achieves the performance predicted by replica-based channel  capacity expressions.

In the compressive sensing literature, input noise causes an effect known as ``noise folding'' \cite{davenport2012proscons,arias2011noise,NFalg2015}.
Reference~\cite{arias2011noise} shows that in the large-system limit, the received SNR is increased by a factor of ${N}/{M}$ due to input noise. Reference~\cite{davenport2012proscons} shows that an oracle-based recovery procedure that knows the signal support exhibits a $3$\,dB loss of reconstructed SNR per octave of sub-sampling. 
More recently, reference \cite{NFalg2015} has introduced an $\ell_1$-norm based algorithm that reduces the effect of input noise.
In contrast to these results, AMPI is a practical signal estimation method that enables one to study the fundamental performance limits of sparse signal recovery in the presence of input noise. 
We also note that reference~\cite{GAMP2011} investigates signal recovery for a similar model as in \fref{eq:sysmodel}. The key difference is that AMPI computes an estimate of the original input signal~$\bms$, whereas the generalized AMP (GAMP) algorithm in~\cite{GAMP2011} recovers the effective input signal $\bmx$.


\subsection{Notation}

Lowercase and uppercase boldface letters represent column vectors and matrices, respectively. For a matrix $\bH$, the conjugate transpose is $\bH^\Herm$. 
The entry on the $i$th row and $j$th column of the matrix $\bH$ is $H_{i,j}$; the $k$th entry of a vector $\vecx$ is $x_k$. 
The $M\times M$ identity matrix is denoted by $\bI_M$ and the $M\times N$ all-zeros matrix by $\mathbf{0}_{M\times N}$. 
%
We define $\left\langle \bmx \right\rangle = \frac{1}{N}\sum_{k=1}^N x_k$.  The quantities $\|\bmx\|_1$ and $\|\bmx\|_2$ represent the $\ell_1$ and $\ell_2$ norms of the vector $\bmx$, respectively. Multivariate real-valued and complex-valued Gaussian probability density functions (pdfs) are denoted by $\setN(\bmm,\bK)$ and $\setC\setN(\bmm,\bK)$, respectively, where~$\bmm$ is the mean vector and $\bK$ the covariance matrix. The operator $\Exop_X\!\left[\cdot\right]$ denotes expectation with respect to the probability density function (PDF) of the random variable (RV)~$X$; $p(x)$ represents the probability distribution of RV $x$. The notation $a\overset{d}{\rightarrow} b$ represents convergence in distribution of $p(a)$ to $p(b)$. The function $\mathsf{1}(\cdot)$ returns $1$ if its argument is true and $0$ otherwise.

\subsection{Paper Outline}
The rest of the paper is organized as follows. 
\fref{sec:AMPI} introduces the AMPI algorithm along with its state evolution (SE) framework. 
\fref{sec:Opt} analyzes optimality conditions for AMPI. 
\fref{sec:app1MIMO} and \fref{sec:app2CS} investigate AMPI for data detection in massive MIMO systems and for sparse signal recovery, respectively. 
We conclude in \fref{sec:conclusion}. 


\section{AMPI: Approximate Message Passing with~Input~Noise}
\label{sec:AMPI}

We now introduce the message passing algorithm used to derive AMPI. We then develop the complex-valued state-evolution (cSE) framework for AMPI, which will be used in \fref{sec:Opt} to establish optimality conditions.

%
%

\subsection{Sum-Product Message Passing}

As discussed in \fref{sec:contributions}, the most challenging step in data detection and signal estimation is calculating the marginal posterior distribution. While the problem of marginalizing a distribution is in general NP-hard \cite{cooper1990computational}, there exist  heuristics that have been successful in certain applications---one of the most prominent marginalizing schemes is message passing.

\begin{figure}[tp]
\centering
\includegraphics[width=0.65\columnwidth]{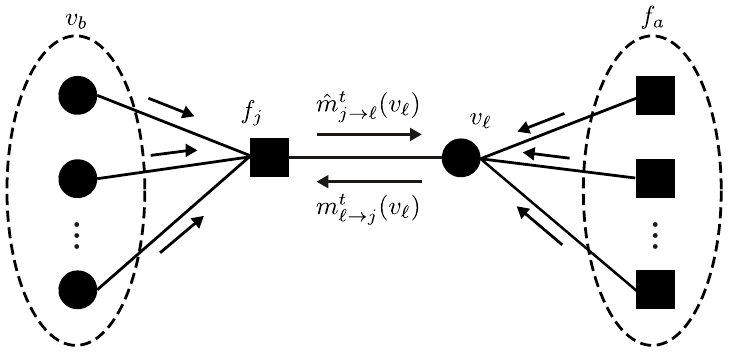}
\vspace{-0.25cm}
\caption{A factor graph illustrating the sum-product message-passing algorithm.}
\label{fig:generalFG}
\end{figure}

Sum-product message passing is a well-established method to compute the marginal distributions~\cite{FG1,donoho2011design}. Consider a joint probability distribution $p(v_1,\dots,v_I)$ with random variables taken from the set $\{v_1,\dots,v_I\}$. Suppose that $p(v_1,\dots,v_I)$ factors into a product of $J$ functions with subsets of the variable set $\{v_1,\dots,v_I\}$ as their argument: $p(v_1,\dots,v_I)=\prod_{j=1}^{J}f_j(V_j)$, where each $V_j$ is a subset of the variable set $\{v_1,\dots,v_I\}$ and $\cup_{j=1}^J V_j=\{v_1,\dots,v_n\}$. 
Such distributions  can be expressed as a factor graphs, which are bipartite graphs consisting of variable nodes to represent each random variable $v_\ell$, factor nodes for each factor $f_j$, and edges connecting them if the factor node is an argument of the variable node.  \fref{fig:generalFG} illustrates a factor graph with variable nodes (circles) and factor nodes (squares).

For the sum-product message passing algorithm, we consider messages $m_{\ell\rightarrow j}^t\!\left(v_i\right) $ (from every variable node $v_i$ to every factor node $f_j$) and $\est m_{j\rightarrow i}^t\!\left(v_i\right)$ (from every factor node $f_j$ to every variable node $v_i$) at iteration $t=1,\ldots,\tmax$ with the equations
\begin{align}
\label{eq:sumprod}
m_{i\rightarrow j}^t\!\left(v_i\right) 
&= \prod_{a\neq j} \est m_{a\rightarrow i}^{t-1}
\!\left(v_i\right)\,\\
\label{eq:sumprod2}
\est m_{j\rightarrow i}^t\!\left(v_i\right) &= \int_\complexset f_j\!\left(\partial f_j\right)
\prod_{b\neq i} m_{b\rightarrow j}^t \!\left(v_b\right)\text{d}(\partial f_j\neq v_i).
\end{align}
Here, $\tmax$ is the maximum number of iterations and $\partial f_j$ is the set of neighbors of node $f_j$ in the graph. After iteratively passing messages between variable nodes and factor nodes, the marginal function of the random variable $v_i$ is approximated by the product of all messages directed toward that variable node. 
If a factor graph is cycle-free, then message-passing converges to the exact marginals. 
If the factor graph has cycles, then general convergence conditions are unknown \cite{weiss2000correctness}.

\subsection{AMPI Derivation}\label{sec:path-AMPI}
%
%
 Consider the system model in \fref{sec:intro}. We are interested in recovering the input signal $\bms$ by computing the MAP or MMSE estimate of the marginal posterior distributions $p(s_\ell|\bmy,\bH)$. 
%
%
The marginal distribution $p(s_\ell | \bmy,\bH)$ can be derived from the joint probability distribution $p(\bms,\bmx,\bmy|\bH)$ as follows:
\begin{align}
& p(s_\ell | \mathbf{y,H})=\int_{\complexset^{N-1}} p(\bms|\bmy,\bH) \,\, \text{d}(s_1, \dots ,s_N \neq s_\ell) \\
&\quad \propto \int_{\complexset^{N-1}} \left(\int_{\complexset^N} p(\bms,\bmx,\bmy|\bH) \dd \bmx \right) \text{d}(s_1, \dots ,s_N \neq s_\ell).
\end{align}
Here, the notation $\text{d}(s_1, \dots ,s_N \neq s_\ell)$ indicates integration over all entries $s_1, \dots ,s_N$ except for~$s_\ell$. Instead of an exhaustive integration for each entry $s_\ell$, $\ell=1,\ldots,N$, we perform sum-product message passing on the factor graph given by the joint PDF $p(\bms,\bmx,\bmy|\bH)=p(\bmy|\bmx,\bH)p(\bmx|\bms)p(\bms).$
The factor graph for this distribution is illustrated in \fref{fig:factorgraph} and consists of the factors $p(\bmy|\bmx,\bH)$, $p(\bmx|\bms)$, and $p(\bms)$. 
The following observations allow us to simplify message passing: (i) The message from variable node $s_\ell$ to the factor node $p(x_\ell | s_\ell)$ is equal to $p(s_\ell)$ and remains constant over all iterations. (ii)  The message from factor node $p(x_\ell | s_\ell)$ to variable node $x_\ell$ is $\int_{\complexset} p(x_\ell | s_\ell) p(s_\ell) \dd s_\ell$. 
From these observations, we see that the factor graph can be divided into two parts. Furthermore, as shown in \fref{fig:factorgraph}, let $v^t_{a \rightarrow \ell} (x_\ell)$ denote the message from the factor node $p(y_a|\bmx)$ to variable node $x_\ell$, and  $\hat{v}^t_{\ell \rightarrow a} (x_\ell)$ the message from variable node $x_\ell$ to factor node $p(y_a|\bmx)$. 
To calculate the messages $v^t_{a \rightarrow \ell} (x_\ell)$ and $\hat{v}^t_{\ell \rightarrow a} (x_\ell)$ on the right side of the graph (marked with a dashed box in \fref{fig:factorgraph}), we can ignore the left part of the graph (on the left of the variable nodes~$x_\ell$) and identify them as messages $p(x_\ell)$ connected to $x_\ell$ that are computed by
\begin{align}\label{eq:effectiveprior}
p(x_\ell)=\int_{\complexset} p(x_\ell | s_\ell) p(s_\ell) \dd s_\ell.
\end{align}
This implies that we can perform message passing on the effective system model $\bmy=\bH \bmx+\bmn$ with effective input signal prior $p(\bmx)$ given in \fref{eq:effectiveprior}. Since we are interested in the estimate of the marginal distribution $p(s_\ell | \mathbf{y,H})$, we can return to the left side of the factor graph and calculate the corresponding messages once $v^t_{a \rightarrow \ell} (x_\ell)$ and $\hat{v}^t_{\ell \rightarrow a} (x_\ell)$ have been computed.
Then, the estimated marginal distribution of $p(s_\ell | \mathbf{y,H})$ is 
\begin{align}\label{eq:s_hatofv}
\hat{p}(s_\ell | \bmy,\bH) = \int_{x_\ell} \prod_{b =1}^{N} \hat{v}_{b \rightarrow \ell} (x_\ell) p(x_\ell |s_\ell) p(s_\ell)dx_\ell, 
\end{align}
where we use the notation $\hat{p}(s_\ell | \bmy,\bH)$ to emphasize that this marginalization is an estimate. In the next section, we will provide conditions for which this estimate is exact.

\begin{figure}[tp]
\centering
		\includegraphics[width=0.8\columnwidth]{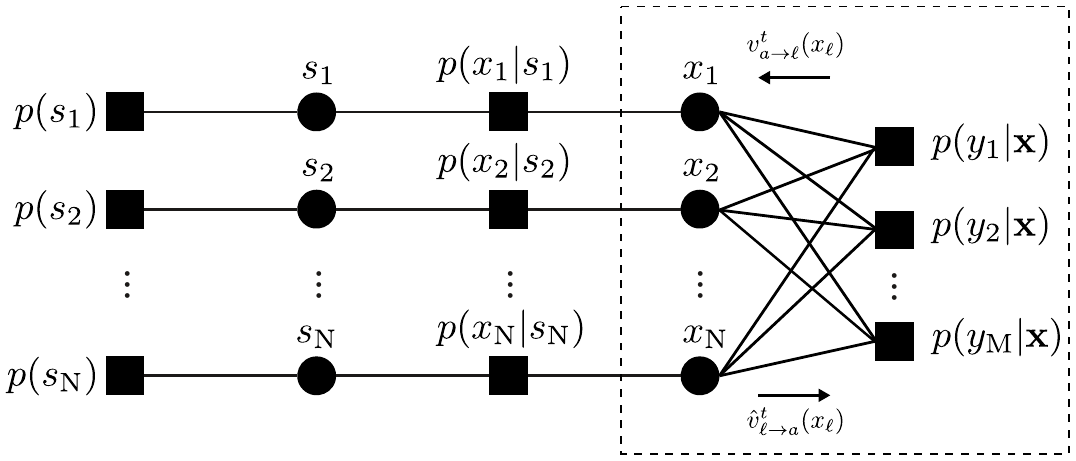}
		\vspace{-0.2cm}
		\caption{The factor graph of the joint distribution $p(\bms,\bmx,\bmy|\bH$). Performing sum-product message passing  on this factor graph yields the marginal posterior distributions $p(s_\ell | \bmy,\bH)$, $\ell=1,\ldots,N$.}
		\label{fig:factorgraph}
\end{figure}

Note that \fref{eq:s_hatofv} is a one-dimensional integral that can either be evaluated in closed form or via  numerical integration as long as the distribution of $\prod_{b =1}^{N} \hat{v}_{b \rightarrow \ell} (x_\ell)$ is known. 
%
To compute the messages $\hat{v}_{b \rightarrow \ell}(x_\ell)$, we perform message passing on  the right side of the factor graph in \fref{fig:factorgraph}.
However, an exact message passing algorithm can quickly result in high complexity; this is mainly because the right side of factor graph in \fref{fig:factorgraph} is fully connected and we need to compute $2 M N$ messages in each iteration.
%
To reduce complexity, we use AMP introduced in~\cite{donoho2009,Maleki2010phd,donoho2011design}. AMP uses the bipartite structure of the graph and the high dimensionality of the problem to approximate the messages with Gaussian distributions. Passing Gaussian messages only requires the message mean and variance instead of PDFs. Furthermore, the structure and dimensionality also allows AMP to approximate all the messages emerging from or going toward one factor node. Specifically, \cite{Maleki2010phd,donoho2011design} show that all the messages emerging from one factor node have approximately the same value---similarly, all messages going toward the same factor node share approximately the same value. Hence. $v^t_{a \rightarrow \ell} \approx v^t_{a}, \forall \ell=1,\dots,N$ and $\hat{v}^t_{i \rightarrow a} \approx \hat{v}^t_{i}, \forall i=1,\dots,M$. This key observation    reduces the number of messages that need to be computed in each iteration from $2MN$ to $M+N$.
Reference \cite{JGMS2015} performs approximate message passing on the system model $\bmy=\bH \bmx+\bmn$ with complex entries called cB-AMP (short for complex Bayesian AMP), which calculates an estimate for the effective input signal~$\hat{x}_\ell$,~$\forall \ell$ using the following algorithm.

\begin{oframed}	
\vspace{-0.3cm}
		\begin{alg}[\bf{cB-AMP}]\label{alg:AMP} 
				Initialize $\hat{x}^1_\ell=\Exop_X[X]$, \mbox{$\resid^1= \bmy$}, and $\gamma_1^2=\No+\beta\Varop_X[X]$ with $\xrv \sim p(x_\ell)$ as defined in~\fref{eq:effectiveprior}. For $t=1,\ldots,\tmax$, compute
		\begin{align}
		\label{eq:GaussianZ}
		\bmz^{t}&=\hat\bmx^t+\bH^\Herm\bmr^t\\
		\nonumber
		\hat{\bmx}^{t+1} &=  \mathsf{F}(\bmz^{t},\gamma_t^2)\\
		\label{eq:postulated}
		\gamma_{t+1}^2 &=  \No+\beta\!\left\langle\mathsf{G}(\bmz^{t},\gamma_t^2)\right\rangle\\
		\nonumber
		\resid^{t+1}  &=   \textstyle \bmy-\bH\hat{\bmx}^{t+1}+\beta \frac{\resid^{t}}{\gamma_t^2} \left\langle\mathsf{G}(\bmz^{t},\gamma_t^2)\right\rangle.
		\end{align}
		The scalar functions $\mathsf{F}(z^t_{\ell},\sigma_t^2)$ and $\mathsf{G}(z^t_{\ell},\sigma_t^2)$ operate element-wise on vectors, correspond to the posterior mean and variance, respectively, and are defined as follows:
		\begin{align}\label{eq:Ffunc}
		\!\!\!\! \mathsf{F}(z^t_{\ell},\sigma_t^2) &= 
		\int_\complexset x_{\ell}p(x_{\ell}\vert z^t_{\ell},\sigma_t^2)\mathrm{d}x_{\ell},\\\label{eq:Gfunc}
		\!\!\!\!  \mathsf{G}(z^t_\ell,\sigma_t^2)&=  \int_\complexset \abs{x_{\ell}}^2\!p(x_{\ell}\vert z^t_{\ell},\sigma_t^2)\mathrm{d}x_{\ell} \!-\! \abs{\mathsf{F}(z^t_{\ell},\sigma_t^2)}^2\!.
		\end{align}	
		The message posterior distribution is $p(x_{\ell}\vert z^t_{\ell},\sigma_t^2)=\frac{1}{Z}p(z^t_{\ell}\vert x_{\ell},\sigma_t^2)p(x_\ell)$, where $p(z^t_{\ell}\vert x_{\ell},\sigma_t^2)\sim\setC\setN(x_{\ell},\sigma_t^2)$ and $Z$ is a normalization constant.
		 \end{alg}
\vspace{-0.3cm}
\end{oframed}

As detailed in \cite[Sec.~\uppercase\expandafter{\romannumeral3}.C]{JGMS2015}, by performing \fref{alg:AMP}, the so-called Gaussian output $\bmz^t$ of cB-AMP at iteration $t$ in~\fref{eq:GaussianZ} can be modelled in the large system limit as
\begin{align}\label{eq:decoupling}
z_\ell^t=x_\ell+w_\ell^t,
\end{align} 
with $w_\ell^t\sim\setC\setN(0,\sigma^2_t)$,
being independent from $x_\ell$ (see \cite[Sec.~6.4]{andreaGMCS} for details in the real domain). This property is known as the decoupling property as cB-AMP  effectively decouples the system into a set of~$N$ parallel and independent additive white Gaussian noise (AWGN) channels. Here, $\sigma^2_t$ is the effective noise variance that can be computed using state evolution, which we  introduce in \fref{sec:SE}. While the quantity $\sigma^2_t$ cannot be tracked directly within cB-AMP,  it can be estimated using the  threshold parameter $\gamma^2_t$ in \fref{eq:postulated} as shown in~\cite{donoho2011design}. 
\fref{alg:AMP} and its Gaussian output $\bmz^t$ are the results of performing AMP on the right side of the factor graph in \fref{fig:factorgraph}. Next, we will use $\bmz^t$ to perform sum-product message passing on the left side of this factor graph.

\subsection{AMP with Input Noise (AMPI)}
\begin{figure}
	\centering
	\subfigure[Impaired linear system with AMPI as the estimator.]{\includegraphics[width=0.95\columnwidth]{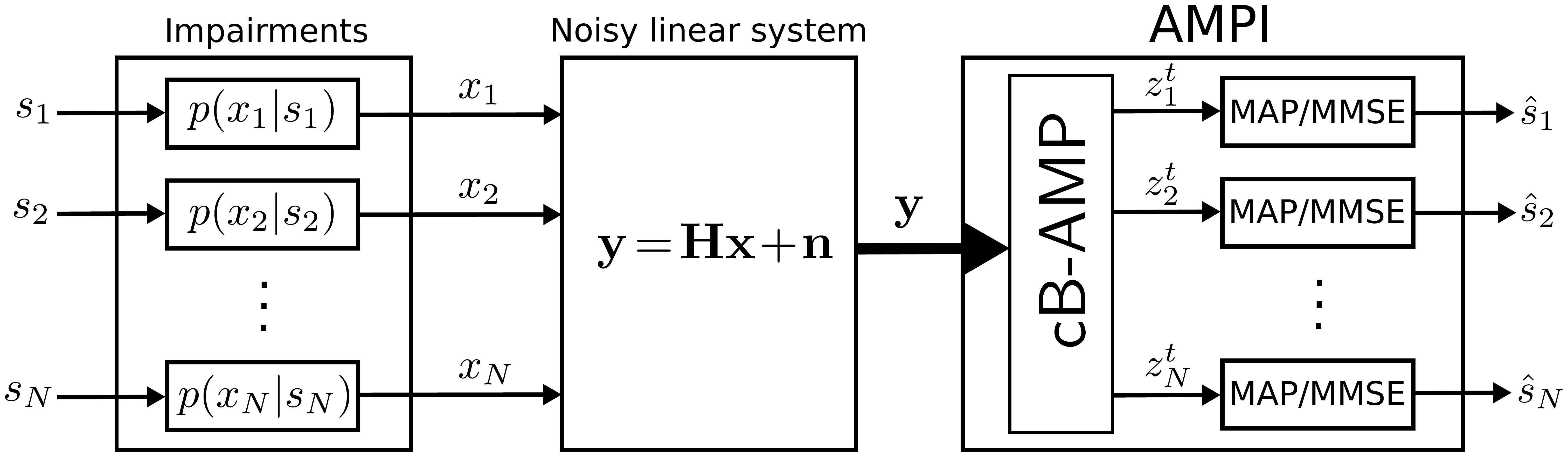}\label{fig:decouple_a}}
	\subfigure[Equivalent decoupled system.]{\includegraphics[width=0.95\columnwidth]{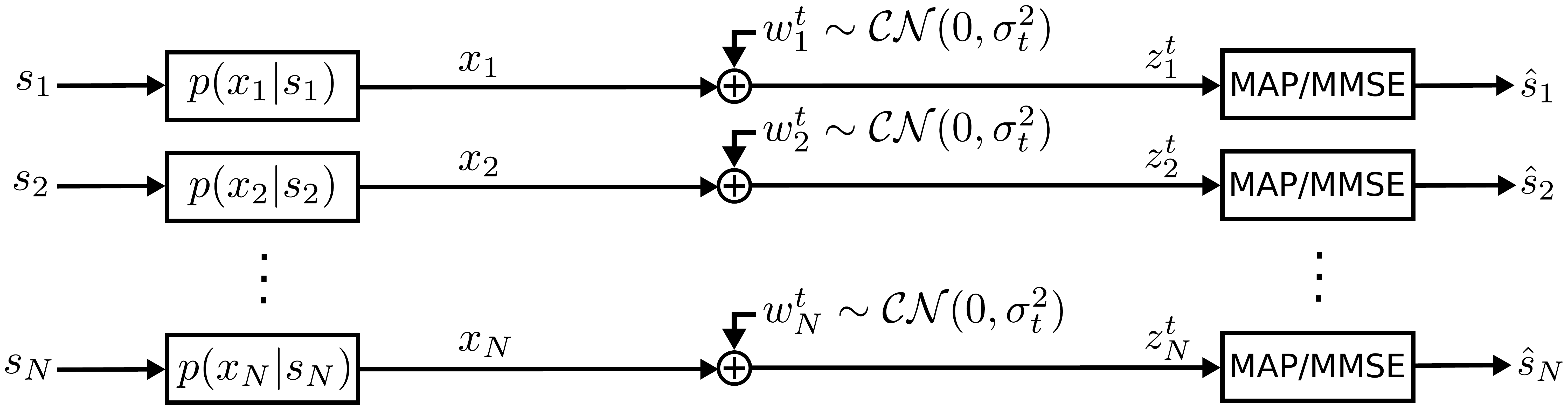}\label{fig:decouple_b}}
	\vspace{-0.2cm}
	\caption{In the large-system limit, AMPI decouples the impaired system into~$N$ parallel and independent AWGN channels, which allows us to perform impairment-aware MAP data detection or MMSE estimation.} \label{fig:decouple}
\end{figure}
%
%
%
We now introduce AMPI, the 2-step procedure to recover the input signal $\bms$ from the input-output relation 
illustrated in \fref{fig:sysmodel}.
First, as illustrated in \fref{fig:decouple_a}, we  use cB-AMP in \fref{alg:AMP} on the effective system model \fref{eq:sysmodel} to compute the Gaussian output~$\vecz^t$ and the effective noise variance $\sigma^2_t$ at iteration $t$, where the Gaussian output $\bmz^t$ is modelled as in~\fref{eq:decoupling}. \fref{fig:decouple_b} shows the equivalent input-output relation of the decoupled system. As detailed in the previous section, this is the result of running AMP on the right side of the factor graph in \fref{fig:factorgraph}.
Second, we use sum-product message passing on the left side of factor graph in \fref{fig:factorgraph} to compute the estimated marginal distribution of $p(s_\ell | \mathbf{y,H})$.
To compute the marginal, we use the Gaussian output $\bmz^t$ in~\fref{eq:decoupling}, i.e., $p(z_\ell^t | x_\ell^t) \sim \setC\setN( x_\ell^t,\sigma^2_t)$; this allows us to compute the marginal posterior distribution for each input signal element $s_\ell$ as follows:
\begin{align}\nonumber
 p(s_\ell | \bmy,\bH) & = p(s_\ell | z_\ell^t) \propto p(s_\ell) p(z_\ell^t | s_\ell) \\
&  = p(s_\ell) \!\int_\complexset p(z_\ell^t | x_\ell^t)p(x_\ell^t | s_\ell)  \dd x_\ell^t. \label{eq:MAP_posterior}
\end{align}
%
Finally,  we can compute individually optimal MAP data detection or MMSE signal estimation for each entry $s_\ell$, $\ell=1,\ldots,N$, independently using the marginal distribution.
Note that \fref{eq:MAP_posterior} can be  obtained from \fref{eq:s_hatofv} by computing $\prod_{b =1}^{N} \hat{v}_{b \rightarrow \ell}$. As shown in~\cite{Maleki2010phd}, we have $\prod_{b =1}^{N} \hat{v}_{b \rightarrow \ell} \sim \setC \setN(z_\ell^t,\sigma^2_t)$. Even though this approach appears to be more straightforward, it lacks the 2-step intuition behind AMPI.
The resulting AMPI algorithm is summarized as follows.
%


\begin{oframed}	
\vspace{-0.3cm}
\begin{alg}[\bf{AMPI}]\label{alg:AMPI} 
	Initialize $\hat{x}^1_\ell=\Exop_X[X]$, \mbox{$\resid^1= \bmy$}, and $\gamma_1^2=\No+\beta\Varop_X[X]$ with $\xrv \sim p(x_\ell)$ as in \fref{eq:effectiveprior}.
	\begin{enumerate}
		\item Run cB-AMP as in \fref{alg:AMP} for $\tmax$ iterations.
		%
		\item Compute $\hat{s}_\ell^{\tmax}=D(z_\ell^\tmax,\sigma^2_\tmax)$, where the function $D$  either computes the  MAP or MMSE estimate of $s_\ell$ using the posterior PDF $p(s_\ell | z_\ell^\tmax)$ defined in \fref{eq:MAP_posterior}. The effective noise variance $\sigma^2_\tmax$ is estimated using the threshold parameter $\gamma_{\tmax}^2$  from cB-AMP.
	\end{enumerate}
\end{alg}
\vspace{-0.3cm}
\end{oframed}	

As shown in Sections~\ref{sec:app1MIMO} and \ref{sec:app2CS}, the function $D$ is chosen to satisfy the optimality conditions in \fref{eq:IOproblem} or \fref{eq:MMSEproblem}.

\subsection{Theoretical Analysis of AMPI via State Evolution}\label{sec:SE}
Analyzing message passing methods operating on dense graphs is generally difficult. However,  the normality of the messages in our application enables us to study theoretical properties in the large-system limit and for uniform linear measurements.
As detailed in \cite{donoho2011design}, the effective noise variance~$\sigma_t^2$ of AMP can be calculated analytically for every iteration $t=1,2,\ldots,\tmax$, using the state evolution recursion.The following theorem repeats the complex state evolution (cSE)  for complex AMP (cB-AMP)~\cite{JGMS2015}.
In \fref{sec:Opt}, we will use the cSE framework to derive optimality conditions for AMPI.
%
\begin{thm}\label{thm:CSE} Assume the model in \fref{eq:sysmodel} with uniform linear measurements. 
	Run cB-AMP using the function $\mathsf{F}$, where $\mathsf{F}$ is a pseudo-Lipschitz function \cite[Sec.~1.1, Eq.~1.5]{bayatimontanari}. Then, in the large-system limit the effective noise variance $\sigma^2_{t+1}$ 
	of cB-AMP at iteration $t$  is given by the following cSE recursion:
	\begin{align}
	\sigma_{t+1}^2 = \No +\beta\Psi(\sigma_t^2,\sigma_t^2), \label{eq:SErecursion} 
	\end{align}
	Here, the MSE function $\Psi$ 
	is defined by
	\begin{align}\label{eq:Psi}
	\Psi(\sigma_t^2,\gamma_t^2) = \Exop_{\xrv,Z}\!\left[\abs{ \mathsf{F}\!\left(\xrv + \sigma_t Z,\gamma_t^2\right) - \xrv}^2 \right]
	\end{align}
	with $\xrv\sim p(x_\ell)$, $Z \sim\setC\setN(0,1)$, and $\mathsf{F}$ and $\mathsf{G}$ are the posterior mean and variance functions from \fref{alg:AMP}. The cSE recursion is initialized at $t=1$ by $\sigma_1^2 = \No + \beta \Varop_X[X]$.
\end{thm}
\begin{rem}\label{rem:dependence}
	The posterior mean function $\mathsf{F}$ and the MSE function $\Psi(\sigma_t^2)$  in \fref{eq:Psi} depend on the effective input signal prior $p(\bmx)$, which, as shown in \fref{eq:effectiveprior}, is a function of the input signal prior $p(\bms)$ and the conditional probability $p(\bmx|\bms)$ that models the transmit-side impairments. 
Furthermore, \fref{thm:CSE} assumes perfect knowledge of the noise variance $\No$; \cite[Thm.~1]{JGMS2015} analyzes the case of  a mismatch in the noise variance. 
\end{rem}

\section{Optimality of AMPI } \label{sec:Opt}
We now analyze the optimality of AMPI for the model introduced in \fref{sec:intro} using the cSE framework.

%
\subsection{Optimality of AMPI Within the AMP Framework}
%
In \fref{sec:AMPI}, we have derived AMPI using message-passing. However, there exists a broader class of algorithms for the same task. Specifically, our version of AMPI performs sum-product message passing using  the posterior mean function $\mathsf{F}$ as defined in \eqref{eq:Ffunc}. One can potentially modify $\mathsf{F}$ (or even use different functions at different iterations) to obtain estimates $\hat{x}_\ell$, $\ell=1,\ldots,N$, and perform MAP data detection or MMSE estimation on these estimates. Such algorithms can still be analyzed through the state evolution framework. The optimality question we ask here is whether it is possible to improve  AMPI by choosing functions different to the ones introduced in \eqref{alg:AMPI}. As we will show in \fref{thm:F_MAPMMSE}, the functions we used in first and second step of AMPI algorithm are indeed optimal.  

Suppose we run AMPI for $\tmax$ iterations. Consider a generalization of \mbox{AMPI}, where, in the first step, the posterior mean function $\mathsf{F}$ in \fref{eq:Ffunc} is replaced with a general pseudo-Lipschitz function $\mathsf{F}_t$ that may depend on the iteration index~$t$; the MAP or MMSE function in the second is replaced with another function $\mathsf{F}_{\tmax+1}$. More specifically, we consider 
\begin{align}\label{eq:generalFs}
&\!\!\!\hat{\bmx}^{t+1} = \textstyle \mathsf{F}_t(\bmz^{t},\gamma_t^2), \,\, t=1,\dots,\tmax\\
&\!\!\!\hat{\bms} \!=\! \textstyle \mathsf{F}_{\tmax+1}\!(\bmz^{\tmax+1}\!,\!\gamma_{\tmax+1}^2), \,
\bmz^{\tmax+1} \!=\! \hat\bmx^{\tmax+1}\!\!+\!\bH^\Herm\bmr^{\tmax+1}\!.
\end{align}
We require the sequence of functions $\{\mathsf{F}_1, \mathsf{F}_2, \ldots,\mathsf{F}_{\tmax+1}\}$ so that \fref{thm:CSE} holds. Now, the question is whether there exists a sequence of functions $\{\mathsf{F}_1, \mathsf{F}_2, \ldots,\mathsf{F}_{\tmax+1}\}$, such that given the application, the resulting algorithm achieves lower probability of error or lower MSE than AMPI. 
%
\fref{thm:F_MAPMMSE} shows that if the solution to the fixed-point equation of \fref{eq:SErecursion} is unique, then AMPI is optimal within AMP framework.

The fixed-point equation of \fref{eq:SErecursion} is computed by letting the number of iterations $\tmax \to \infty$ in \fref{eq:SErecursion}, which yields
%
%
\begin{align}\label{eq:FPsimplified}
\sigma^2 &  
= \No +\beta\Psi(\sigma^2,\sigma^2).
\end{align}
%

\begin{thm}\label{thm:F_MAPMMSE}
	Assume the system model in \fref{sec:intro} with uniform linear measurements and the large-system limit. Suppose that we use \mbox{AMPI} with an arbitrary set of pseudo-Lipschitz functions $\mathsf{F}_1,\ldots,\mathsf{F}_{\tmax+1}$ as described in \fref{eq:generalFs}.
	If the solution to the fixed-point equation in \fref{eq:FPsimplified} is unique, then the choice of $\mathsf{F}_1,\ldots,\mathsf{F}_{\tmax+1}$ that achieves optimal performance according to \fref{eq:IOproblem} and \fref{eq:MMSEproblem} are as introduced in \fref{alg:AMPI}.
	%
	%
\end{thm}
\fref{thm:F_MAPMMSE} shows that it is impossible to improve upon the original choice of \mbox{AMPI}. The proof is given in \fref{app:F_MAPMMSE}. 
The fixed-point equation \fref{eq:FPsimplified} can in general have one or more fixed points. If it has more than one fixed point, then \mbox{AMPI} may converge to different solutions, depending on the initialization~\cite{ZMWL2015}. 
As it is clear from the proof of \fref{thm:F_MAPMMSE} in \fref{app:F_MAPMMSE}, even in cases where \mbox{AMPI} does not have a unique fixed point, one of its fixed points corresponds to the optimal solution in AMP framework.  
Hence,  to provide precise conditions for optimality of AMPI, will analyze the fixed point equation \fref{eq:FPsimplified} for a unique solution only.
To establish conditions under which the fixed point equation \eqref{eq:FPsimplified} has a unique solution, we use the following quantities from~\cite[Defs.~1--4]{JGMS2015conf}.
\begin{defi}\label{def:betaN0}
Fix the input signal prior $p(\vecs)$ and input noise distribution $p(\bmx|\bms)$. Then, the exact recovery threshold (ERT) $\betamax$ and the minimum recovery threshold (MRT) $\betamin$ are
\begin{align}
\!\betamax \!& =\! \min_{\!\sigma^2>0}{\textstyle \!\left\{\!\!\left(\!\frac{\Psi(\sigma^2,\sigma^2)}{\sigma^2}\!\right)^{\!\!-1}\!\right\}},\, 
\betamin \!=\!\min_{\!\sigma^2>0} \textstyle\!\left\{\!\!\left(\!\frac{\textnormal{d}\Psi(\sigma^2,\sigma^2)}{\textnormal{d}\sigma^2}\!\right)^{\!\!-1}\!\right\}\!.
\end{align}
The minimum critical noise $\Nomin$ is defined as 
\begin{align}
\Nomin &=  \min_{\sigma^2>0}\textstyle \!\left\{\sigma^2-\beta\Psi(\sigma^2,\sigma^2):\beta\frac{\textnormal{d}\Psi(\sigma^2,\sigma^2)}{\textnormal{d}\sigma^2}=1\right\}\!,
\end{align}
and the maximum guaranteed noise $\Nomax$ is defined as  
\begin{align}
\Nomax&= \max_{\sigma^2>0}\textstyle\!\left\{\sigma^2 - \beta\Psi(\sigma^2,\sigma^2):\beta\frac{\textnormal{d}\Psi(\sigma^2,\sigma^2)}{\textnormal{d}\sigma^2} = 1\right\}\!.
\end{align}
\end{defi}	

Using \fref{def:betaN0}, the following theorem establishes three regimes for which fixed-point equation \fref{eq:FPsimplified} has an unique solution. The proof follows from \cite[Sec.~\uppercase\expandafter{\romannumeral 4}-D, \uppercase\expandafter{\romannumeral 4}-E]{JGMS2015}.
\begin{lem} \label{lem:betaN0}
Let the assumptions of \fref{thm:CSE} hold and $\tmax\to\infty$. Fix $p(\vecs)$ and $p(\vecx|\vecs)$.  If the variance of the receive noise $\No$ and system ratio $\beta$ are in one of the following three regimes:
\begin{enumerate}
\item $\beta \in \left(0,\betamin\right]$ and $\No \in \reals^+$
\item $\beta \in \left(\betamin,\betamax\right)$ and $\No \in \left[0,\Nomin\right) \cup \left(\Nomax,\infty\right)$
\item $\beta \in \left[\betamax,\infty\right)$ and $\No \in \left(\Nomax,\infty\right)$,
\end{enumerate}
then the fixed point equation \fref{eq:FPsimplified} has a unique solution.
\end{lem}

For AMPI, the quantities in \fref{def:betaN0} depend on  $p(\bmx)$, which is a function of $p(\vecs)$ and  $p(\vecx|\vecs)$ (cf.~\fref{rem:dependence}). These quantities can be computed either numerically or in closed-form (see \cite[Sec.~\uppercase\expandafter{\romannumeral 3}]{JGMS2015conf}).
In many applications, the effective input signal prior $p(\bmx)$ is continuous and bounded which results in certain properties for ERT and MRT as discussed in the following lemma. The proof is given in \fref{app:betalimits}.
\begin{lem}\label{lem:betalimits}
Suppose the probability density $p(\bmx)$ of the effective input signal $\bmx$ is continuous and bounded. Furthermore, let the assumptions made in \fref{thm:CSE} hold. Then, the ERT and MRT satisfy $\betamax=1$ and $\betamin \leq 1$.
\end{lem}

From \fref{lem:betalimits}, we conclude that for a system with a continuous effective input signal prior $p(\bmx)$ we have $\betamin \leq 1$. As noted in \fref{lem:betaN0}, $\betamin$ determines the values of system ratio $\beta$ under which AMPI can be optimal for any noise variance. In other words, $\beta \leq \betamin \leq 1$ implies that the system should not be under-determined. As an example, consider a massive MIMO system that uses QPSK constellations. In the absence of input noise, $\betamin\approx2.9505$ (see \cite[Tbl.~\uppercase\expandafter{\romannumeral 1}]{JGMS2015conf}). However, by adding the slightest amount of input noise with a continuous probability distribution, such as Gaussian input noise, $\betamin$ abruptly decreases to values no larger than $1$.

\section{Application 1: Massive MIMO}
\label{sec:app1MIMO}

As discussed in \fref{sec:intro}, impairment-aware data detection is an important  part of practical massive MIMO systems.
In reference \cite{GJMS_2015_conf}, we provided an impairment-aware data detection algorithm called LAMA-I (short for \underline{la}rge-\underline{M}IMO \underline{a}pproximate message passing with transmit \underline{i}mpairments). LAMA-I which is a low-complexity data detection algorithm, is AMPI derived for massive MIMO systems. In this section, we briefly revisit the signal and system model for massive MIMO and LAMA-I from \cite{GJMS_2015_conf}. We then provide an optimality analysis of LAMA-I which was not shown in \cite{GJMS_2015_conf}. In particular, we prove that besides optimality in the AMP framework, LAMA-I is able to achieve the same error-rate performance as the IO data detector.
\subsection{LAMA-I: AMPI for Massive MIMO Data Detection}\label{sec:MIMO_model}
Consider an input signal $\vecs \in \complexset^{N}$ sent through an impaired MIMO channel with input-output relation \fref{eq:sysmodel} introduced in \fref{sec:intro} with the following assumptions. The entries of~$\vecs$ are chosen from a constellation set $\setO$, e.g., QAM, and $\vecs$ is assumed to have i.i.d. prior distribution $p(\vecs)=\prod_{\ell=1}^{N} p(s_{\ell})$ with the following distribution for each transmit symbol:
\begin{align} 
p\!\left(s_{\ell}\right) = \sum_{a\in \setO}p_a\delta\!\left(s_{\ell} - a\right)\!.
\end{align}
The received vector is $\vecy \in \complexset ^{M}$ is the received vector, and $N$ and $M$ denote the number of user equipments and base-station antennas, respectively.
The MIMO channel matrix $\bH \in \complexset^{M \times N}$ is assumed to be perfectly known at the receiver. 

As shown in \fref{alg:AMPI}, we first use cB-AMP to compute the Gaussian output~$\vecz^{\tmax}$ and the effective noise variance $\sigma^2_{\tmax}=\gamma^2_{\tmax}$ at iteration $t$. The MIMO system is decoupled into a set of~$\MT$ parallel and independent additive white Gaussian noise (AWGN) channels. \fref{fig:decouple_b} shows the equivalent decoupled system. Using \fref{eq:MAP_posterior}, we can compute the marginal posterior distribution $p(s_\ell | \bmy,\bH) = p(s_\ell | z_\ell^{\tmax})$ from the Gaussian output. 
The marginal posterior distribution allows us to compute the MAP estimate for each data symbol independently as 
\begin{align}
\hat{s}_\ell^{\tmax}=D(z_\ell^\tmax,\sigma^2_\tmax)=\argmax_{s_\ell\in\setO}p(s_\ell | z_\ell^{\tmax}).
\end{align}
We call this procedure the LAMA-I algorithm in \cite{GJMS_2015_conf}.
Note that \cite[Sec.~\uppercase\expandafter{\romannumeral4}]{GJMS_2015_conf} details the derivation of LAMA-I for a Gaussian input noise model $p(x_\ell|s_\ell) = \setC \setN (s_\ell,\NT), \forall \ell={1,\ldots,N}$.

\subsection{Individually Optimal (IO) Data Detection}

We now show that for uniform linear measurements and the large-system limit, \mbox{LAMA-I} is able to achieve the error-rate performance of the IO data detector~\fref{eq:IOproblem}, if the fixed-point equation~\fref{eq:FPsimplified} has a unique solution. There has been some work in the past focusing on optimality of AMP in the absence of input noise such as \cite{barbier2020mutual} (see \cite[Sec.~\uppercase\expandafter{\romannumeral4}]{JGMS2015} for a survey). The core of our optimality analysis is the performance of IO data detection in the presence of input noise based on the replica analysis presented in \cite{GV2005}. 
To prove individual optimality, we first introduce the following definition. The replica analysis for IO data detection makes the following assumption about~$\hat{s}^\IO_\ell$.
\begin{defi}\label{def:hardsoft}
The IO solution $\hat{s}^\IO_\ell$ is said to satisfy hard-soft assumption, if and only if there exist a function $D: \mathbb{R} \rightarrow \mathcal{O}$, with the following properties: (i) $\hat{s}^\IO_\ell = D(\mathbb{E} (s_\ell | \mathbf{y, H}) )$ and (ii) for every $s\in \setO$ the $D^{-1}(s)$ is Borel measurable and its boundary has Lebesgue measure zero. 
\end{defi}


We can prove that the hard-soft assumption is in fact true for equiprobable BPSK constellation points, i.e., we have
\[
\mathbb{E} ({s}_\ell | \mathbf{y, H}) = \mathbb{P} ({s}_\ell = +1 | \mathbf{y, H}) - \mathbb{P} ({s}_\ell = -1 | \mathbf{y, H}),
\] 
and hence, $\hat{s}^\IO_\ell  = \sign(\mathbb{E} ({s}_\ell | \mathbf{y, H})) $. However, it is an interesting open problem whether the assumption in \fref{def:hardsoft} holds for other, more general, constellations as well.  

To simplify our proofs, we make an extra assumption:
\begin{defi}\label{def:fixedD}
The IO solution $\hat{s}^\IO_\ell$ is said to satisfy the fixed-D assumption, if in addition to satisfying the hard-soft assumption, the function $D$ from \fref{def:hardsoft} is only a function of $\beta$, $p(s_\ell)$, $p(x_\ell|s_\ell)$, and $p(n_\ell)$. In particular, $D$ does not depend on the dimension $N$ of the input signal. 
 \end{defi}
Note that for equiprobable BPSK symbols, we have $\hat{s}^\IO_\ell  = \sign(\mathbb{E} ({s}_\ell | \mathbf{y, H}))$ and the fixed-D assumption clearly holds. 
 Intuitively speaking, when the dimensions are large, we do not expect the function~$D$ to change with the dimension $N$.
 
 Before we can establish individual optimality of LAMA-I, we introduce \fref{thm:IOfixedD}, which analyzes the error probability of the IO solution and provides an equivalent relation, which enables us to compute the error probability of an IO data detector and consequently compare it with other detectors.

\begin{thm}\label{thm:IOfixedD}
Suppose that the IO solution satisfies both the hard-soft and fixed-D assumptions  in Definitions \ref{def:hardsoft} and \ref{def:fixedD}. Furthermore, assume that the assumptions underlying  the replica symmetry in \cite{GV2005} are correct. Then, $\mathbb{P} ( \hat{s}^\IO_\ell \neq s_\ell)$ converges to $\mathbb{P} (D(Q) \neq S)$ in probability. Here $Q= X+ \tilde{\sigma}Z$ with
$p(S,X)=p(s_\ell) p(x_\ell|s_\ell)$, $Z \sim \setC \setN(0,1)$ being independent of $(S,X)$ and $\tilde{\sigma}$ satisfying the following equation:
\begin{eqnarray}\label{eq:fpIOAMP}
\tilde{\sigma}^2 = \No+ \beta \Psi(\tilde{\sigma}^2).
\end{eqnarray}
\end{thm}
We now provide conditions for which LAMA-I algorithm achieves the error-rate performance of IO data detector. The proof is given in \fref{app:IOptimality}.
\begin{thm} \label{thm:IOptimality}
	Assume the system model in \fref{eq:sysmodel} with uniform linear measurements. Suppose that the assumptions made in \fref{thm:IOfixedD} hold.
	 Furthermore, assume that the fixed-point equation \eqref{eq:fpIOAMP} has a unique solution. Let us call the estimate of LAMA-I after $t$ iterations $\hat{s}_\ell^t$. Then in large-system limit, for any $\epsilon\geq0$ there exists an iteration number $t_0$ such that
	\[
	 \lim_{N \rightarrow \infty} \frac{1}{N} \sum_{\ell=1}^N \mathbb{P} (\hat{s}_\ell^{t_0} \neq s_\ell) \leq \lim_{N \rightarrow \infty} \frac{1}{N} \sum_{\ell=1}^N \mathbb{P} (\hat{s}_\ell^\IO \neq s_\ell)+ \epsilon,
	\]
	where the limits are taken in probability.\footnote{If the limit in $\lim_{n \to \infty} X_n =X$ is taken in probability, it means the probability of  $X_n$ being far from $X$ should go to zero when $n$ increases.}
\end{thm}

\fref{thm:IOptimality} proves individual optimality of LAMA-I algorithm given certain conditions are met on system size and ratio. The inequality in this theorem shows how LAMA-I with an infinite number of iterations achieves the same error-rate as that of IO data detector. While LAMA-I requires the large-system limit and an infinite number of iterations to achieve the performance of IO data detector, \fref{fig:MIMO} and simulation results in~\cite{GJMS_2015_conf} demonstrate that LAMA-I achieves near-IO performance in realistic, finite dimensional large-MIMO systems.


 \section{Application 2: Compressive Sensing}
 \label{sec:app2CS}
 
We now apply  AMPI to compressive sensing signal recovery in the presence of input noise. We first introduce the system model and then derive AMPI for this system. We conclude with simulation results that compare AMPI to existing methods.

 \subsection{System Model}\label{sec:CS_model}
 The noiseless version of CS signal recovery can be solved perfectly under certain conditions on the system dimension and the sparsity level~\cite{donoho2006,donoho2009}. Recovery under measurement noise has been analyzed extensively; see, e.g.,  \cite{andreaGMCS,candes2006stable}. 
However, practical CS systems may be affected by input noise~\cite{arias2011noise,treichler2009application}.
For example, the input signals may not be perfectly sparse or might be affected by noise that appears prior to the measurement process. 
In what follows, we introduce AMPI to recover the sparse input signal from measurements contaminated with both input and measurement noise. 
Let $\vecs \in \reals^{N}$ be the signal of interest we want to reconstruct from the noisy measurements $\bmy \in \reals^{M}$ with the system model in \fref{eq:sysmodel} introduced in \fref{sec:intro}. Since the system model for compressive sensing is assumed to be under-determined, we have $M \leq N$ (or equivalently $\beta \geq 1$). Moreover, the input signal $\bms$ is a sparse vector with at most $K$ non-zero entries. 

\subsection{AMPI for Compressive Sensing}
%
%
%
To apply \fref{alg:AMPI} for CS, we first need to derive the effective input signal prior $p(\bmx)$ in~\fref{eq:effectiveprior}, which requires   the input signal prior $p(\bms)$. 
As explained in~\cite{andreaGMCS}, a practical solution to capture sparsity in $\bms$ is to assume an i.i.d.\ Laplace prior. We assume $p(\bms)=\left( \frac{\lambda}{2}\right)^N \exp \left( -\lambda \|\bms\|_1\right)$ with a regularization parameter $\lambda>0$ that can be tuned to best model the sparsity of the input signal $\bms$. For optimal performance, one should tune~$\lambda$ to minimize the MSE of AMPI. 
%
Besides the parameter $\lambda$, AMPI requires a threshold parameter $\gamma_t^2$ that must be tuned in each iteration. To attain optimal performance, both of these parameters should be tuned in each algorithm iteration. 
To this end, we will use an iteration index subscript for $\gamma_t^2$ and for $\lambda_t$. 
 
As noted in \fref{sec:path-AMPI}, there exist different methods to tune the threshold parameter $\gamma_t^2$. In the paper \cite{mousavi2015consistent}, the authors propose an asymptotically optimal tuning approach using Stein's unbiased risk estimate (SURE)~\cite{stein1981estimation} for the threshold parameter $\gamma_t^2$. Here, we follow a similar approach that optimally tunes both  parameters $\lambda_t$ and $\gamma_t^2$.

First, run Step~1 of \fref{alg:AMPI} for $\tmax$ iterations. Optimal tuning for AMPI can be achieved if $\lambda_1,\dots,\lambda_{\tmax}$ and $\gamma_1^2,\dots,\gamma_{\tmax}^2$ are tuned in such a way that the value of the asymptotic MSE or $ \!\! \lim\limits_{N \to \infty}\frac{1}{N}\|\hat{\bmx}^{\tmax+1}\!-\bmx\|^2$ is minimized. This requires a joint optimization of the asymptotic MSE over all variables $\{\lambda_1,\dots,\lambda_{\tmax},\gamma_1^2,\dots,\gamma_{\tmax}^2\}$. However, such a  joint optimization is not practical as the iterative nature of AMPI does not allow one to write an explicit expression for MSE. The following theorem shows that one can simplify the joint parameter optimization by tuning each pair $(\lambda_t,\gamma_t^2)$ at iteration~$t$ starting from $t=1$ to $\tmax$. The proof for this theorem follows from \cite[Thm.~3.7]{mousavi2015consistent} with minor modifications.


\begin{thm}\label{thm:opt-tuning}
Suppose that the parameters $\lambda_1^,\dots,\lambda_{\tmax}^,,$ $\gamma_1^{2},\dots,\gamma_{\tmax}^{2}$ are optimally tuned for iteration $\tmax$ of AMPI. Then, $\forall t<\tmax$, the parameters $\lambda_1,\dots,\lambda_{t},\gamma_1^{2},\dots,\gamma_{t}^{2}$ are also optimally tuned for iteration $t$.
\end{thm}
 
\fref{thm:opt-tuning} implies that instead of jointly tuning all parameters $\{\lambda_1,\dots,\lambda_{\tmax},\gamma_1^{2},\dots,\gamma_{\tmax}^{2}\}$, we can tune $(\lambda_1,\gamma_1^2)$ at iteration 1. Given the parameters $(\lambda_1,\gamma_1^2)$, we can then tune $(\lambda_2,\gamma_2^2)$ at iteration 2, and repeat this  process for $\tmax$ iterations.
 
The missing piece is to minimize the MSE at iteration $t$ with appropriate parameters $(\lambda_t,\gamma_t^2)$. As cSE in \fref{thm:CSE} suggest, the asymptotic MSE at iteration $t$ is given by the function $\Psi(\sigma_t^2,\gamma_t^2,\lambda_t)=\Exop_{\xrv,Z}\!\left[\abs{ \mathsf{F}\!\left(\xrv + \sigma_t Z,\gamma_t^2,\lambda_t\right) - \xrv}^2 \right]$, with $\xrv\sim p(x_\ell)$, $Z \sim\setC\setN(0,1)$, and $\mathsf{F}$ as the posterior mean function introduced in \fref{eq:Ffunc}. Notice that all functions $\Psi$, $\mathsf{F}$, and $\mathsf{G}$ will also be a function of $\lambda_t$. We use a similar tuning approach as in \cite{mousavi2015consistent} and since the MSE function $\Psi(\sigma_t^2,\gamma_t^2,\lambda_t)$ depends on the unknown signal prior $p(\xrv)$, we estimate it using SURE in each iteration. 
For AMPI, SURE is given by
\begin{align}
\hat{\Psi}(\sigma_t^2,\lambda_t,\gamma_t^2) =\, & \textstyle \frac{1}{N}\|\mathsf{F}(\bmz^{t},\gamma_t^2,\lambda_t)-\bmz^{t}\|^2+\sigma_t^2 \notag \\
&\textstyle  +\frac{2\sigma_t^2}{\gamma_t^2}\left\langle\mathsf{G}(\bmz^{t},\gamma_t^2,\lambda_t)-1\right\rangle, \label{eq:Psihat}
\end{align}
where we estimate $\sigma_t^2$ by ${\|\resid^{t}\|^2}/{M}$ as in \cite{andreaGMCS}.
We now modify AMPI as in \fref{alg:AMPI} for CS applications as follows.

\begin{oframed}	
\vspace{-0.3cm}
\begin{alg}[\bf{AMPI-SURE}]\label{alg:AMPI-SURE} 
	Set $\hat{\bmx}^1=0$ and \mbox{$\resid^1= \bmy$}.
	\begin{enumerate}
		\item For $t=1,2,\ldots,\tmax$ compute
	\begin{align*}
	\bmz^{t}&=\hat\bmx^t+\bH^\Herm\bmr^t\\	
	(\lambda_t,\gamma_t^{2}) &= \argmin \limits_{\lambda \geq 0, \gamma^2 \geq 0} \hat{\Psi}(\sigma_t^2,\lambda,\gamma^{2})\\
	\hat{\bmx}^{t+1} &=  \mathsf{F}(\bmz^{t},\gamma_t^2,\lambda_t)\\
	\resid^{t+1}  &= \textstyle   \bmy-\bH\hat{\bmx}^{t+1}+\beta \frac{\resid^{t}}{\gamma_t^2} \left\langle\mathsf{G}(\bmz^{t},\gamma_t^2,\lambda_t)\right\rangle.
	\end{align*}
	Here, $\hat{\Psi}$ is given by \fref{eq:Psihat}, and $\mathsf{F}$ and $\mathsf{G}$ are the posterior mean and variance given by \fref{eq:Ffunc}, where $p(x_\ell)=\int_{\complexset} p(x_\ell | s_\ell) p(s_\ell) \dd s_\ell$ and $p(s_\ell)= \frac{\lambda_t}{2} \exp \left( -\lambda_t |s_\ell|\right)$.
		\item Compute the MMSE estimate for $t=\tmax$ with the posterior PDF $p(s_\ell | z_\ell^\tmax)$ as defined in \fref{eq:MAP_posterior} and $p(w_\ell^\tmax) \sim \setN(0,\sigma^2_\tmax)$. The effective noise variance $\sigma^2_\tmax$ and signal prior distribution $p(s_\ell)$ are estimated using $\gamma_{\tmax}^2$ and $ \frac{\lambda_{\tmax}^*}{2} \exp\! \left( -\lambda_{\tmax}^* \|\s_\ell\|_1\right)$.
	\end{enumerate}
\end{alg}
\vspace{-0.3cm}
\end{oframed}

\fref{alg:AMPI-SURE} summarizes AMPI for CS. The only difference to \fref{alg:AMPI} is the presence of the extra parameter $\lambda_t$  that is optimally tuned using SURE depending on the signal sparsity.

 \subsection{AMPI Sparse Signal Recovery Under Gaussian Input Noise}
AMPI for CS recovery as in \fref{alg:AMPI-SURE} is defined for a general class of input noise distributions $p(\bmx|\bms)$. We now provide a derivation of AMPI for the specific case of Gaussian input noise, i.e., where   $p(\bmx|\bms)=\setC \setN(\bms,\NT \bI_M)$ \cite{treichler2009application,arias2011noise}.
The following lemma details the prior $p(\bmx)$ and the functions $\mathsf{F}$ and~$\mathsf{G}$ needed in Steps 1 and 2 of \fref{alg:AMPI-SURE}. The derivations are  given in \fref{app:FG_CS} in supplementary.

\begin{lem}\label{lem:FG_CS}
	
	For a CS system as defined by \fref{sec:CS_model}, the prior $p(\bmx)$, the posterior mean $\mathsf{F}$ and variance $\mathsf{G}$ defined in \fref{alg:AMPI-SURE} are given by:
	\begin{align}
		p(\bmx)
		&= \prod\limits_{i=1}^{N}\frac{\lambda}{2}\exp\!\left(\frac{\lambda^2\NT}{2}\right) 
		\!\bigg(\!
		\exp(\lambda x_i)Q\!\left(\frac{x_i+\lambda\NT}{\sqrt{\NT}}\right)\! \notag
		 \\
		& \quad + \exp(-\lambda x_i)\!\left(1-Q\!\left(\frac{x_i-\lambda\NT}{\sqrt{\NT}}\right)\right)\!
		\bigg) \\
		\mathsf{F}(\hat{x},\tau,\lambda)&=\hat{x}+\lambda \tau \eta(\hat{x},\tau)\\
		\mathsf{G}(\hat{x},\tau,\lambda)&= \tau + \lambda^2 \tau^2 (1- (\eta(\hat{x},\tau))^2) \notag \\ 
		& \quad  - \frac{4}{\gamma(\hat{x},\tau)} \frac{\lambda \tau^2}{\sqrt{2\pi (\NT+\tau)}},
	\end{align}
	where we define
	\begin{align}
	 \eta(\hat{x},\tau) & = \frac{\mathrm{erfcx}(\alpha)-\mathrm{erfcx}(\beta)}{\mathrm{erfcx}(\alpha)+\mathrm{erfcx}(\beta)} \label{eq:eta} \\
	\gamma(\hat{x},\tau)&=\mathrm{erfcx}(\alpha)+\mathrm{erfcx}(\beta) \\
		\alpha& =\frac{\hat{x}+\lambda(\NT+\tau)}{\sqrt{2(\NT+\tau)}} \\ 		 
		\beta& =\frac{-\hat{x}+\lambda(\NT+\tau)}{\sqrt{2(\NT+\tau)}}.
	\end{align}
The Q-function is $Q(x)=\int_x^\infty\frac{1}{\sqrt{2\pi}}\exp\left(\!-\frac{t^2}{2}\right)\dd t$, the error function $\mathrm{erfc}(x)=2Q(\sqrt{2}x)$, and $\mathrm{erfcx}(x)=x^2 \mathrm{erfc}(x)$.
\end{lem}
Before providing simulation results for AMPI-SURE, we next summarize two baseline algorithms used as a comparison.

\subsection{Two Alternative Algorithms}\label{sec:2otherAlg}

\subsubsection{Noise Whitening} 
Noise whitening has been proposed for data detection and compressive sensing in \cite{Studer_Tx_OFDM} and \cite{arias2011noise}, respectively. This approach relies on the Gaussian input-noise model, which enables one to ``whiten'' the impaired system model~\fref{eq:sysmodel} by multiplying the vector $\bmy$ with the whitening matrix $\bW=\No\bQ^{-\frac{1}{2}}$, where $\bQ=\NT\bH \bH^\Herm+\No\bI_M$ is the covariance matrix of the effective input and receive noise $\bmn+\bH \bme$. Whitening resultsin  a statistically equivalent input-output relation
%
$\tilde{\bmy}=\tilde{\bH}\bms+\tilde{\bmn}$,
where $\tilde{\bmy}=\bW\bmy$, $\tilde{\bH}=\bW\bH$ and $\tilde{\bmn} \sim \setC \setN(0,\No\bI_M)$ is independent of $\bms$ \cite{Studer_Tx_OFDM}. Thus, signal recovery can be performed with conventional algorithms, such as AMP~\cite{DMM10a,DMM10b}. 
The drawback of noise whitening is in computing the whitening matrix $\bW$, whose dimensions may be extremely large (e.g., in imaging applications). AMPI avoids computation of $\bW$, which reduces complexity. Furthermore, AMPI supports more general  input noise models---in contrast, noise whitening requires a Gaussian input-noise model. 

\subsubsection{Convex Optimization}

Consider the system model in \fref{sec:CS_model}. If the input noise is zero, i.e. $\bmx=\bms$, then for a sparse signal $\bms$ or equivalently~$\bmx$, sparse signal recovery can be performed by solving~\cite{donoho2011design}
\begin{align*}
\hat{\bms}=\argmin_{\bms} \textstyle \frac{1}{2}\|\bmy-\bH\bms\|^2_2+\lambda\|\bms\|_1.
\end{align*}

If the input noise is non-zero, i.e., $\bmx \neq \bms$, then we can solve the following optimization problem:
\begin{align}\label{eq:CS_ConvexProblem}
\hat\bmx = \argmin_\bmx \frac{1}{2 \No} \|\bmy-\bH\bmx\|_2^2 - \log p(\bmx) 
\end{align}
If the term $- \log p(\bmx) $ is convex and differentiable, then we can use efficient algorithms that guarantee convergence to an optimal solution. 
The following result  establishes convexity for the Gaussian input-noise model and provides the gradient, which we use to solve  \fref{eq:CS_ConvexProblem}. The proof is given in \fref{app:CS_convexity} in supplementary derivations.
\begin{lem}\label{lem:CS_convexity}
The objective function $q(\bmx)=\frac{1}{2 \No} \|\bmy-\bH\bmx\|_2^2 - \log p(\bmx) $ is convex and its gradient is given by 
%
\begin{align}
\nabla_\bmx q(\bmx) = \frac{1}{\No}(\bH \bmx -\bmy)^\Tran \bH - \nabla_\bmx \left[ \log p(\bmx) \right]\!,
\end{align}
and,
%
$\!\nabla_\bmx \!\left[ \log p(\bmx) \right] \!\!=\!\! \lambda 
[\eta(x_1,\!0),...,\!\eta(x_N,\!0)]^\Tran
$
%
with $\eta(\hat{x},\!\tau\!)$ in \fref{eq:eta}.
\end{lem}

Hence, we propose to use the non-linear conjugate gradient method of Polak-Ribiere~\cite{polak1969note} to solve for \fref{eq:CS_ConvexProblem}; see \cite[alg. 4.4]{frandsen1999unconstrained} for the algorithm details.
The downside of such an approach is that there is no known fast approach to set the parameter $\lambda$. In contrast, AMPI in \fref{alg:AMPI-SURE} can be tuned optimally.

\subsection{Results and Comparison}
\begin{figure}
	\centering
	\includegraphics[width=0.7\columnwidth]{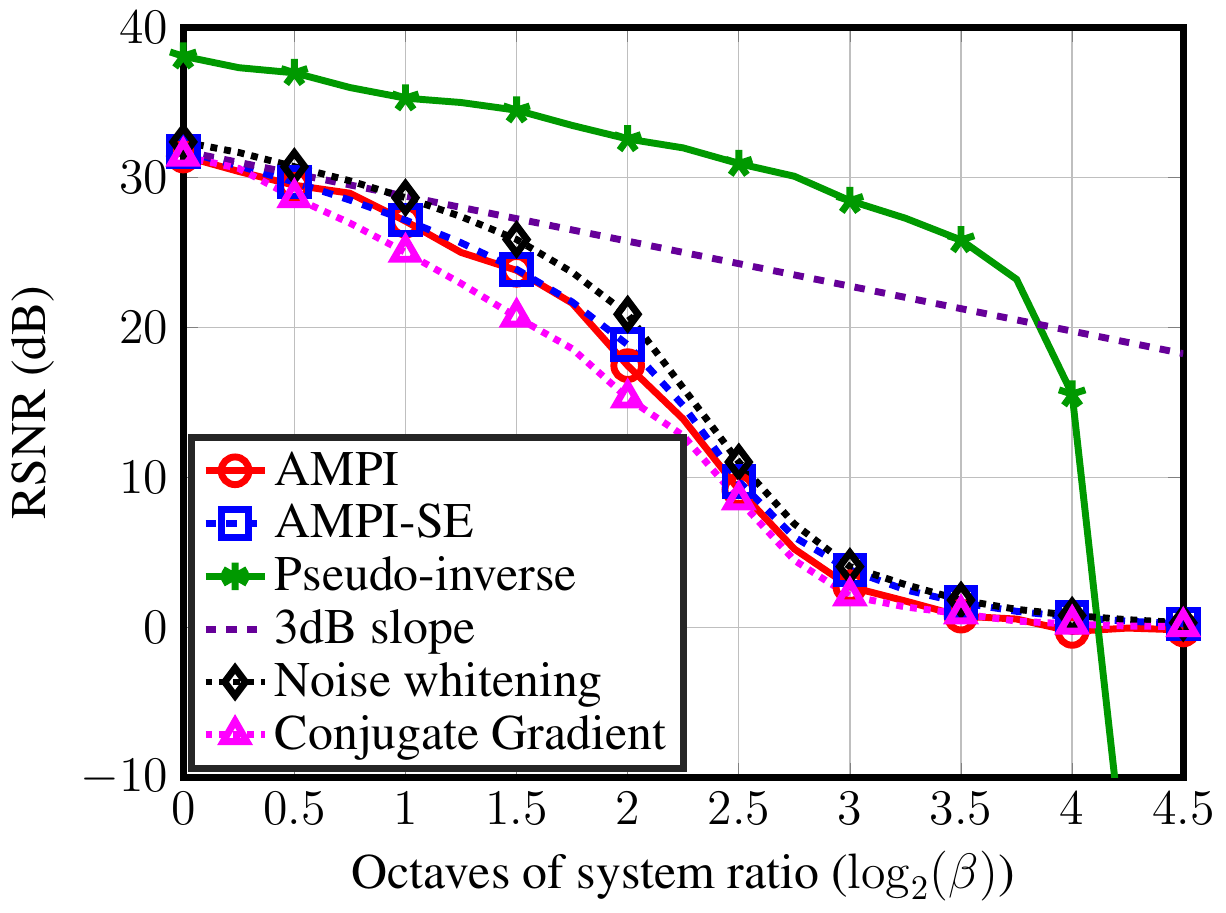}
	\vspace{-0.2cm}
	\caption{Reconstruction SNR of AMPI and other algorithms for sparse signal recovery with input noise. The signal has sparsity of $5\%$, SNR of $30$\,dB, and is affected by input noise with $\textit{EVM}$ of $-30$\,dB. AMPI achieves the same performance as noise whitening and nonlinear conjugate gradients but at much lower computational complexity.}
	\label{fig:RSNRcomp}
\end{figure}
\fref{fig:RSNRcomp} shows simulation results for sparse signal recovery with a sparsity rate of $\frac{K}{N}=5\%$ and signal dimension of $N=1000$. The input signal is generated with a Bernoulli-Gaussian distribution. The indices of the non-zeros entries are selected  from an equiprobable Bernoulli distribution and each entry is generated from a standard normal distribution. The reconstruction signal to noise ratio (RSNR) is plotted as a function of the octaves of system ratio $\beta$ (also referred to as sub-sampling ratio). Furthermore, we consider an average SNR of $\SNR={\Exop\left[\|\bH \vecs\|^2\right]}/{\Exop\left[ \|\vecn\|^2\right]}=\beta \frac{\Es}{\No}=30$\,dB and an error vector magnitude of $\textit{EVM}={\Exop\left[ \|\vece\|^2 \right]}/{\Exop\left[ \|\vecs\|^2 \right]}=\frac{\NT}{\Es}=-30$\,dB.
In \fref{fig:RSNRcomp}, the RSNR results of each algorithm is averaged over 20 different randomly-created input signals. The performance of AMPI with 100 iterations is depicted as a solid circle-marked red curve. AMPI's performance almost perfectly matches the cSE predictions in  \fref{eq:SErecursion} (the dashed square-marked blue curve). As a comparison, we show the performance of  noise whitening and convex optimization with non-linear conjugate gradients. The dotted diamond-marked black curve shows the performance of noise whitening followed by AMP. While whitening achieves the same performance as AMPI, it entails significantly higher complexity as one has to first compute the whitening matrix. 
The dotted triangle-marked magenta curve shows the performance of nonlinear conjugate gradients with $100$ iterations. This method requires a computationally expensive grid search per each iteration to tune the parameter $\lambda$. Since we set this tuning parameter  using the \emph{optimal} ones from AMPI  (solid circle-marked red curve), the conjugate gradients method performs very well.
%
The solid star-marked green curve corresponds to the oracle-based approach of taking the pseudo-inverse assuming the support is known. As \cite{davenport2012proscons} suggests, the RSNR of this method decays with a slope of $3$\,dB per octave. However, none of the proposed algorithms follows the $3$dB per octave slope.
%

 
 \section{Conclusion}
\label{sec:conclusion}

We have introduced AMPI (short for approximate message passing with input noise), a novel  data detection and estimation algorithm for  systems that are corrupted by input noise. AMPI is computationally efficient and can be used for a broad range of input-noise models. Furthermore, the complex state-evolution (cSE) framework enables a theoretical analysis of AMPI in the large system limit and for i.i.d.\ Gaussian measurement matrices. 
Under these conditions, we have investigated optimality conditions of AMPI for data detection and signal estimation.
We have shown that AMPI is optimal within the AMP framework and, under additional assumptions, achieves individually-optimal error-rate performance in massive MIMO applications. 
For the Gaussian input-noise model, we have used numerical results to show that AMPI outperforms methods that ignore input noise and performs on-par with whitening and optimization-based methods, but at (often significantly) lower complexity.

%

\appendices

%
%
%
%
%
%
%

\section{Proof of \fref{thm:F_MAPMMSE}}\label{app:F_MAPMMSE}

\subsection{Proof Outline}

We provide an optimality proof for an application where AMPI solves  the IO problem in \fref{eq:IOproblem}. This means that in Step 2 of AMPI, we use the MAP estimate. The optimality proof where AMPI is supposed to minimize the MSE follows analogously.
%
%
Suppose that we use \mbox{AMPI} with an arbitrary set of pseudo-Lipschitz functions $\mathsf{F}_1,...,\mathsf{F}_{\tmax+1}$ as described in \fref{eq:generalFs}.
In this proof, we will establish that AMPI in \fref{alg:AMPI} chooses these functions such that the outputs
 $\hat{s}_\ell$ for $\ell\!=\!1,...,N$, correspond to the solution of the IO problem~\fref{eq:IOproblem} in the large-system limit.
%
%
%
We start by writing down the optimality criterion in  \fref{eq:IOproblem} as
\begin{align} \label{eq:Opt_criterion}
\mathsf{F}_{\!\tmax\!+\!1}\!(\!z_\ell^{\tmax+1}\!\!,\!\sigma_{\tmax+1}^2\!) 
\!=\!\argmin\limits_{\mathsf{F}} \mathbb{P} \!\left(\!\mathsf{F}(z_\ell^{\tmax\!+\!1}\!\!,\!\sigma_{\tmax\!+\!1}^2) \!\neq\! s_\ell \!\right)\!\!.
\end{align}
The estimate generated by the function $\mathsf{F}_{\tmax+1}(z^{\tmax+1}_\ell,\sigma_{\tmax+1}^2)$ during Step 2 at iteration ${\tmax+1}$ minimizes the per-entry symbol-error probability. Note that the functions $\mathsf{F}_1,\ldots,\mathsf{F}_{\tmax+1}$ operate element-wise on vectors.
\begin{rem}
	The criterion \fref{eq:Opt_criterion} appears to only consider the $\ell$th entry. Since, however, the probability in~\fref{eq:Opt_criterion} is taken w.r.t.\ the randomness in the matrix $\bH$, the vector $\bms$, input and receive noise, the criterion is in fact affected by all other entries.
\end{rem}	

We now establish the optimality proof by the following two lemmas, with proofs in \fref{app:FMAP} and \ref{app:FMMSE}. 
In what follows, we assume the random variables $S \sim p(s_\ell), X|S \sim p(x_\ell|s_\ell)$, and $Z\sim\setC\setN(0,1)$ to be independent of $X$ and $S$.
\begin{lem}\label{lem:FMAP}
	Let the assumptions of Thm.~\ref{thm:F_MAPMMSE} hold.
	For the criterion \fref{eq:Opt_criterion} to hold for $\ell = 1,\ldots,N$, $\mathsf{F}_{\tmax+1}(z^{\tmax+1}_\ell,\sigma_{\tmax+1}^2)$ at iteration $\tmax+1$ must be the MAP estimator
	\begin{align}\label{eq:maptmax1}
	\mathsf{F}_{\!\tmax\!+\!1\!}(\!z^{\tmax\!+\!1}_\ell\!\!,\!\sigma_{\tmax\!+\!1}^2\!) \!=\! \argmax \limits_{s_\ell\in\setO} {p_{S | X+\sigma_{\tmax\!+\!1}\!Z}(\!s_\ell|z_\ell^{\tmax+1}\!)}
	\end{align}	
\end{lem}
\begin{lem}\label{lem:FMMSE}
	Let the assumptions of Thm.~\ref{thm:F_MAPMMSE} hold.
	For the criterion \fref{eq:Opt_criterion} to hold for $\ell = 1,\ldots,N$, the functions $\mathsf{F}_t(z^t_\ell,\sigma_t^2)$, $t=1,\ldots,{\tmax}$, must be the unique set of MMSE estimators, i.e., 
	$\mathsf{F}_t(z^t_\ell,\sigma_t^2)=\Exop_{X | X+\sigma_t Z}[x_\ell | z_\ell^t]$. 
\end{lem}
\fref{lem:FMAP} suggests that for optimality to hold, the function $\mathsf{F}_{\tmax+1}(z^{\tmax+1}_\ell,\sigma_{\tmax+1}^2)$ must be the MAP estimator as provided in Step 2 of \fref{alg:AMPI}. Furthermore, \fref{lem:FMMSE} suggest that if the solution to the fixed-point equation of \fref{eq:FPsimplified} is unique, then the set of functions $\mathsf{F}_1,\ldots,\mathsf{F}_{\tmax}$ that satisfy the optimality criterion are given by the unique set of MMSE functions $\mathsf{F}_t(z^t_\ell,\sigma_t^2)=\Exop_{X | X+\sigma_t Z}[x_\ell | z_\ell^t]$ for $t=1,\ldots,\tmax$. These functions match the posterior mean function $\mathsf{F}$ (defined by \fref{eq:Ffunc}) in Step 1 of \fref{alg:AMPI}.
Thus, from \fref{lem:FMAP} and \fref{lem:FMMSE}, we conclude that \fref{alg:AMPI} solves the IO problem \fref{eq:IOproblem} given that the fixed point equation \fref{eq:FPsimplified} is unique.

\subsection{Proof of \fref{lem:FMAP}}\label{app:FMAP}
We start with the following lemma.
\begin{lem}\label{lem:zeta_limit}
	Define $\zeta_N \!  = \!\frac{1}{N}\sum_{\ell=1}^N \!  \mathsf{1}(\mathsf{F}_{\tmax+1}(z^{\tmax+1}_\ell\!\!\!,\sigma_{\tmax+1}^2 ) \! \neq \! s_\ell)$. Fix the system ratio $\beta=N/M$ and let $N\rightarrow\infty$. Then, for a given $\bH$, we have
	\begin{align}\label{eq:prob-convergence}
	\zeta_{N}  \!& \stackrel{\text{a.s.}}{\to}  \Exop_{X,S,Z}
	\!\left[\mathsf{1}(\mathsf{F}_{\tmax+1}(X+\sigma_{\tmax+1} Z,\sigma_{\tmax+1}^2) \neq S)\right]\!=\!\zeta_{\infty},
	\end{align}
\end{lem}
\begin{proof}
The proof follows from \cite[Thm.~1]{bayatimontanari} and we briefly outline the main ideas. Reference \cite{bayatimontanari} established that for any Pseudo-Lipschitz function $\psi$ in the large-system limit we have 
\begin{align}\label{eq:Bayati}
\frac{1}{N} \sum_{\ell=1}^N \psi (z_\ell^t, x_\ell) \rightarrow \Exop\left[{ \psi (X+ \sigma_t Z, X)}\right].
\end{align}
Note that the expression on the left in \fref{eq:Bayati} is the expectation under the empirical distribution of the joint random variables $(z_\ell^t, x_\ell)$. Hence, we can say this equation corresponds to one of the forms of convergence in distribution as pointed out in \cite[Lem.~2.2]{van2000asymptotic}. Based on this lemma, \fref{eq:Bayati} suggests that the empirical distribution of $(z_\ell^t, x_\ell)$ also converges weakly to the distribution of $(X+ \sigma_t Z, X)$. 
Furthermore, since $s_\ell \rightarrow x_\ell \rightarrow \bmy$ forms a Markov chain, $z_\ell^t$ (which is a function of $\bmy$) is independent of $s_\ell$ given $x_\ell$; this implies that the empirical distribution of $(z_\ell^t, s_\ell)$ converges weakly to the distribution of $(X+ \sigma_t Z, S)$. Hence, based on the same Lemma 2.2 in \cite{van2000asymptotic} we can conclude that if $\mathcal{B}$ is a Borel measurable set, whose boundary has Lebesgure measure zero (and if $X+ \sigma_t Z$ is absolutely continuous with respect the Lebesgue measure) then
$\frac{1}{N} \sum_{\ell=1}^N  \mathsf{1} \left( (z_\ell^t,s_\ell) \in \mathcal{B} \right) \stackrel{d}{\rightarrow} \mathbb{P} \left( (X+ \sigma_t Z,S) \in \mathcal{B}\right).$
From this result, we conclude that \fref{eq:prob-convergence} holds. 
%
\end{proof}

Notice from \fref{lem:zeta_limit} that $\zeta_N$ is bounded. As a consequence, in the large-system limit, we have 
\begin{align}\label{eq:zeta_limit}
\Exop_{\bmy,\bH}\!\left[\zeta_{N}\right] \stackrel{\text{a.s}}{\to} \Exop_{\bmy,\bH}\!\left[\zeta_{\infty}\right].
\end{align}
Let us compute 
\begin{align}\nonumber
&\!\! \Exop_{\bmy,\bH}\left[ \zeta_{N} \right] 
 = \Exop_{\bmy,\bH}\!\left[ \frac{1}{N} \sum_{\ell=1}^N \mathsf{1}(\mathsf{F}_{\tmax+1}(z^{\tmax+1}_\ell,\sigma_{\tmax+1}^2) \!\neq\! s_\ell) \right] \\ 
& \quad= \mathbb{P}\left(\mathsf{F}_{\tmax+1}(z^{\tmax+1}_\ell,\sigma_{\tmax+1}^2) \neq s_\ell\right)\! ,\,  \ell=1,\ldots,N. \label{eq:zeta_Exop}
\end{align}
Here, the last equality holds 
because under the permutations of the entries in $\bms$, the distribution does not change and thus, $\mathbb{P} \left(\mathsf{F}(z_\ell^{\tmax+1},\sigma_{\tmax+1}^2) \neq s_\ell \right)$ does not depend on the index $\ell$. Hence, using \fref{eq:zeta_limit} and \fref{eq:zeta_Exop}, in the large system limit we have 
%
\begin{align}
&\mathbb{P} \left(\mathsf{F}_{\tmax+1}(z^{\tmax+1}_\ell,\sigma_{\tmax+1}^2) \neq s_\ell\right) \notag \\
& \!\!\stackrel{\text{a.s}}{\to} \Exop_{\bmy,\bH}\!\left[\zeta_{\infty}\right] \!= \!\mathbb{P} \left(\mathsf{F}_{\tmax+1}(X+\sigma_{\tmax+1} Z,\sigma_{\tmax+1}^2) \!\neq\! S\right)\!.
\end{align}
Hence, instead of minimizing $\mathbb{P} \left(\mathsf{F}(z^{\tmax+1}_\ell,\sigma_{\tmax+1}^2) \neq s_\ell\right)$ in  \fref{eq:Opt_criterion}, we can minimize $\mathbb{P} \left(\mathsf{F}(X+\sigma_{\tmax+1} Z,\sigma_{\tmax+1}^2) \neq S \right)$. 
Thus, the optimal choice of $\mathsf{F}$ at iteration ${\tmax+1}$ is the MAP estimator in \fref{eq:maptmax1}.


\subsection{Proof of \fref{lem:FMMSE}}\label{app:FMMSE}
We next show that satisfying \fref{eq:Opt_criterion}, requires the functions  $\mathsf{F}_t(z^t_\ell,\sigma_t^2)$, $t=1,\ldots,{\tmax}$, to be the MMSE estimators. 
Let us call the MAP estimator $\mathsf{F}_{\tmax+1}$ at iteration ${\tmax+1}$ from \fref{lem:FMAP} as  $\mathsf{F}_\sigma^{\text{MAP}}$.  Then, the following lemma holds.
\begin{lem}\label{lem:MAP_smallsigma}
	$\mathbb{P}\left(\mathsf{F}_\sigma^{\textnormal{MAP}}(x_\ell+\sigma Z,\sigma^2) \neq s_\ell\right)$ is a non-decreasing function in $\sigma$.
\end{lem}

The proof follows by contradiction. In particular, we try to show that exists two quantities $\sigma_1\!<\!\sigma_2$ such that
	\begin{align}
	&\mathbb{P}\left(\mathsf{F}_{\sigma_1}^{\text{MAP}}(x_\ell+\sigma_1 Z,\sigma_1^2) \!\neq\! s_\ell\right) \notag \\
	&\quad \qquad >\mathbb{P}\left(\mathsf{F}_{\sigma_2}^{\text{MAP}}(x_\ell+\sigma_2 Z,\sigma_2^2) \!\neq\! s_\ell\right)\!.
	\end{align}
	Based on $x_\ell+\sigma_1 Z$, we consider the randomized estimator $\mathsf{F}_{\sigma_2}^{\text{MAP}}(x_\ell+\sigma_1 Z + \sqrt{\smash[b]{\sigma_2^2- \sigma_1^2}} \tilde{Z},\sigma_2^2)$, where $\tilde{Z} \sim \setC \setN(0,1)$ independent of $Z$.
	It is easy to see that since $\sigma_1 Z + \sqrt{\sigma_2^2- \sigma_1^2} \tilde{Z}$ is distributed according to $\setC \setN(0, \sigma_2^2)$, we have 
	\begin{align}
	&  \mathbb{P}\!\left( \mathsf{F}_{\sigma_2}^{\text{MAP}}(x_\ell+\sigma_1 Z + \sqrt{\smash[b]{\sigma_2^2- \sigma_1^2}} \tilde{Z},\sigma_2^2) \neq s_\ell\right) \notag \\
	 & \qquad \qquad = \mathbb{P}\left(\mathsf{F}_{\sigma_2}^{\text{MAP}}(x_\ell+\sigma_2 Z,\sigma_2^2) \neq s_\ell\right)\!. \label{eq:cont_assum}
	\end{align}
	Hence,
	\begin{align}
	& \mathbb{E}_{\tilde{Z}}\!\left[ \mathbb{P}\!\left( \!\left. \mathsf{F}_{\sigma_2}^{\text{MAP}}(x_\ell+\sigma_1 Z + \sqrt{\smash[b]{\sigma_2^2- \sigma_1^2}} \tilde{Z},\sigma_2^2) \neq s_\ell  \right\vert  {\tilde{Z}} \right)\right] \notag \\
	& \qquad \qquad =\mathbb{P}\left(\mathsf{F}_{\sigma_2}^{\text{MAP}}(x_\ell+\sigma_2 Z,\sigma_2^2) \neq s_\ell\right)\!.
	\end{align}
	Hence, there exists a value of $\tilde{Z}$, call it $\bar{Z}$, for which 
	\begin{align}
	& \mathbb{P} \left( \mathsf{F}_{\sigma_2}^{\text{MAP}}(x_\ell+\sigma_1 Z + \sqrt{\smash[b]{\sigma_2^2- \sigma_1^2}} \bar{Z},\sigma_2^2) \neq s_\ell \right) \notag \\
	& \qquad \qquad <\mathbb{P}\left(\mathsf{F}_{\sigma_2}^{\text{MAP}}(x_\ell+\sigma_2 Z,\sigma_2^2) \neq s_\ell \right)\!.
	\end{align}
	Note that this estimator is the non-randomized estimator
	of $x_\ell+ \sigma_1 Z$. Consequently, we have  
	\begin{align}
& \mathbb{P} \left( \mathsf{F}_{\sigma_2}^{\text{MAP}}(x_\ell+\sigma_1 Z + \sqrt{\smash[b]{\sigma_2^2- \sigma_1^2}} \tilde{\theta},\sigma_2^2) \neq s_\ell\right) \notag \\
& \qquad \qquad \leq \mathbb{P}\left(\mathsf{F}_{\sigma_2}^{\text{MAP}}(x_\ell+\sigma_2 \theta,\sigma_2^2) \neq s_\ell\right)\!,
	\end{align}
	which is in contradiction with \fref{eq:cont_assum}.

%
\fref{lem:MAP_smallsigma} shows that in order for $\mathsf{F}_\sigma^{\text{MAP}}$ to provide the smallest probability of error in \fref{eq:Opt_criterion}, the function sequence $\{\mathsf{F}_1,\ldots,\mathsf{F}_{{\tmax}}\}$ should lead to the minimum possible $\sigma_{\tmax+1}^2$. In \fref{lem:MMSE}, we prove that $\sigma^2_{\tmax+1}$ is minimal only if $\{\mathsf{F}_1,\ldots,\mathsf{F}_{{\tmax}}\}$ are the MMSE estimators.
We first provide \fref{lem:MMSE_Auxiliary}, which is required in the proof for \fref{lem:MMSE}. 
\begin{lem}\label{lem:MMSE_Auxiliary}
	$\inf \limits_{\mathsf{F}} \Exop_{X,Z}\!\left[\left| \mathsf{F}(X+\sigma Z,\sigma^2)-X \right|^2\right] $ is a nondecreasing function in $\sigma$.
\end{lem}

The proof follows by contradiction. Suppose that the statement of \fref{lem:MMSE_Auxiliary} is not true. Then, there exists two quantities $\hat{\sigma}_1 < \hat{\sigma}_2$ such that 
	\begin{align}
	&\inf \limits_{\mathsf{F}} \Exop _{X,Z}\!\left[\left| \mathsf{F}(X+\hat{\sigma}_1 Z,\hat{\sigma}_1^2)-X \right|^2\right] \notag \\
	& \qquad \qquad > \inf \limits_{\mathsf{F}} \Exop _{X,Z}\! \left[ \left| \mathsf{F}(X+\hat{\sigma}_2 Z,\hat{\sigma}_2^2)-X \right|^2\right]\!. \label{eq:ContradictAssumption}
	\end{align}
	Now suppose that both infima in \fref{eq:ContradictAssumption} are achieved with $\mathsf{F}_{\hat{\sigma}_1}$ and $\mathsf{F}_{\hat{\sigma}_2}$, respectively. Then, we can construct a new estimator $\tilde{\mathsf{F}}_{\hat{\sigma}_1}$ for the variance $\hat{\sigma}_1$ as 
	\begin{align}
	& \tilde{\mathsf{F}}_{\hat{\sigma}_1}(X\!+\!\hat{\sigma}_1 Z,\hat{\sigma}_1^2) \notag \\
	& \qquad = \Exop_{\tilde{Z}} \left[ \!\left.\mathsf{F}_{\hat{\sigma}_2}(X\!+\!\hat{\sigma}_1 Z \!+\! \sqrt{\smash[b]{\hat{\sigma}_2^2\!-\!\hat{\sigma}_1^2}} \tilde{Z} ,\hat{\sigma}_2^2)  \right| Z\right]\!,
	\end{align}
	where $\tilde{Z} \sim \setC \setN(0,1)$. Hence, $\hat{\sigma}_1 Z + \sqrt{\hat{\sigma}_2^2-\hat{\sigma}_1^2} \tilde{Z} \sim \setC \setN(0,\hat{\sigma}_2^2)$. 
	We now prove that $\tilde{\mathsf{F}}_{\hat{\sigma}_1}$ has a lower risk than $\mathsf{F}_{\hat{\sigma}_1}$, which is in contradiction with $\mathsf{F}_{\hat{\sigma}_1}$ achieving infimum of the function $\Exop _{X,Z}\!\left[\left| \mathsf{F}(X+\hat{\sigma}_1 Z,\hat{\sigma}_1^2)-X \right|^2\right]$ for $\hat{\sigma}_1$, i.e.,
	\begin{align}
	&\Exop^{\star}\!\left[\left| \tilde{\mathsf{F}}_{\hat{\sigma}_1}(X+\hat{\sigma}_1 Z,\hat{\sigma}_1^2) - X \right|^2 \right]\\ 
	&=\Exop^\star\!\left[  \!\left. \Big\vert \Exop_{\tilde{Z}} \left[ \mathsf{F}_{\hat{\sigma}_2}(X+\hat{\sigma}_1 Z + \sqrt{\smash[b]{\hat{\sigma}_2^2-\hat{\sigma}_1^2}} \tilde{Z},\hat{\sigma}_2^2) \right\vert Z \right] - X \Big\vert^2\right] \notag  \\
	&\!\!\stackrel{(a)}{\leq}\! \Exop^\star \!\!\left[ \!\left. \Exop_{\tilde{Z}} \!\left[\left|  \mathsf{F}_{\hat{\sigma}_2}(X\!+\!\hat{\sigma}_1 Z \!+\! \sqrt{\smash[b]{\hat{\sigma}_2^2\!-\!\hat{\sigma}_1^2}} \tilde{Z},\hat{\sigma}_2^2 ) \! - \!X\! \right|^2\right|\! Z\right] \right] \\
	&=\!\Exop^\star\!\! \left[\!\left| \mathsf{F}_{\hat{\sigma}_2}\!(X\!+\!\hat{\sigma}_2 Z,\hat{\sigma}_2^2)\!-\!X \right|^2 \!\right] \\
	& \stackrel{(b)}{<} \!\!\Exop^\star\!\!\left[\!\left| \mathsf{F}_{\hat{\sigma}_1}\!(X\!+\!\hat{\sigma}_1 Z,\hat{\sigma}_1^2)\!-\!X \right|^2\!\right]\!\!,
	\end{align}
	where $\Exop^\star[\cdot]$ is the expectation over the random variables $X$ and~$Z$. Here, the two inequalities $(a)$ and $(b)$ come from Jensen's inequality and assumption \fref{eq:ContradictAssumption}, respectively. 

%
\begin{lem}\label{lem:MMSE}
	The sequence of functions $\{\mathsf{F}_1,\ldots,\mathsf{F}_{{\tmax}}\}$ must be the MMSE estimators to lead to the minimum $\sigma_{{\tmax+1}}^2$.
\end{lem}
	The proof follows by induction. Suppose that the functions $\mathsf{F}_1,\ldots,\mathsf{F}_{t-1}$ are MMSE estimators to minimize $\sigma_t^2$. Then, we prove by contradiction that to minimize $\sigma_{t+1}^2$, all functions $\mathsf{F}_1,\ldots,\mathsf{F}_{t}$ must be MMSE estimators. 
	Now, suppose that $\mathsf{F}^*_1,\ldots,\mathsf{F}^*_{t}$ are the optimal functions that lead to the minimum effective noise variance that we call $\sigma_{t+1}^{*2}$. And assume that at least one of these functions is not an MMSE estimator. Then, we prove that if $\bar{\mathsf{F}}_1,\ldots,\bar{\mathsf{F}}_{t}$ are all MMSE estimators, they generate a lower variance $\bar{\sigma}_{t+1}^2$. Let us compute $\bar{\sigma}_{t+1}^2$ from the c-SE framework in \fref{thm:CSE}:
	\begin{align}	
	& \bar{\sigma}_{t+1}^2(\bar{\mathsf{F}}_1,\ldots,\bar{\mathsf{F}}_{t}) \notag \\	
	&\quad \stackrel{(a)}{=}\! \No \!+\! \beta \Exop _{X,Z}\!\left[\left| \bar{\mathsf{F}}_t (X\!+\!\bar{\sigma}_t Z,\bar{\sigma}_t^2) \!-\!X\right|^2\right] \\
	&\quad\stackrel{(b)}{=}\! \No \!+\! \beta \inf\limits_{\mathsf{F}_t} \Exop _{X,Z}\!\left[\left| \mathsf{F}_t (X\!+\!\bar{\sigma}_t Z,\bar{\sigma}_t^2) \!-\!X\right|^2\right] \\
	&\quad\stackrel{(c)}{\leq}\! \No \!+\! \beta \inf\limits_{\mathsf{F}_t}  \Exop _{X,Z}\!\left[ \left| \mathsf{F}_t (X\!+\!\sigma^*_t Z,\sigma^{*2}_t) \!-\!X\right|^2\right]\\
	&\quad\leq \No \!+\! \beta  \Exop _{X,Z}\!\left[\left| \mathsf{F}^*_t (X\!+\!\sigma^*_t Z,\sigma^{*2}_t) \!-\!X\right|^2\right]\\
	&\quad\stackrel{(d)}{=}  \sigma_{t+1}^{*2}(\mathsf{F}^*_1,...,\mathsf{F}^*_{t}).\label{eq:inequality2}
	\end{align}
	Here, $(a)$ and $(d)$ follow from cSE, $(b)$ follows from $\bar{\mathsf{F}}_{t}$ being an MMSE estimator. Lastly, $(c)$ follows from \fref{lem:MMSE_Auxiliary} and the base case  of induction, i.e., $\bar{\sigma}_{t}^2(\bar{\mathsf{F}}_1,...,\bar{\mathsf{F}}_{t-1}) < \sigma_{t}^{*2}(\mathsf{F}^*_1,...,\mathsf{F}^*_{t-1})$.
	By inspecting inequality  \fref{eq:inequality2}, we see that it is in contradiction with the optimality assumption of $\mathsf{F}^*_1,...,\mathsf{F}^*_{t}$ unless $\mathsf{F}^*_{i}=\bar{\mathsf{F}}_{i}$ for $i=1,...,t$. Note that here we have assumed that at every stage, the MMSE estimator is unique. 
	Because otherwise, there can be another set of MMSE estimators $\tilde{\mathsf{F}}_1,...,\tilde{\mathsf{F}}_{t}$ which generates a lower $\tilde{\sigma}_{t+1}^2$. Thus, this lemma proves that if $\mathsf{F}_1,...,\mathsf{F}_{\tmax}$ are a unique set of MMSE estimators, then they generate the minimum $\sigma_{\tmax+1}^2$ by letting $\tmax \to \infty$.
	%

%
From Lemmas \ref{lem:MAP_smallsigma} and \ref{lem:MMSE}, we conclude that to satisfy \fref{eq:Opt_criterion}, the functions $\mathsf{F}_t(z^t_\ell,\sigma_t^2)$, $t=1,\ldots,{\tmax}$, must be the set of MMSE estimators, i.e.,  $\mathsf{F}_t(z^t_\ell,\sigma_t^2)=\Exop_{X | X+\sigma_t Z}[x_\ell | z_\ell^t]$, which is equivalent to the message mean  \fref{eq:Ffunc} for $t=1,\ldots,{\tmax}$ in Step~1 of AMPI. Note that for optimality to hold, we need the MMSE estimators to be unique. If the solution to the fixed-point equation of \fref{eq:FPsimplified}  is unique, then this guarantees uniqueness of the MMSE estimators.

\section{Proof of \fref{lem:betalimits}}\label{app:betalimits}

Starting with \fref{def:betaN0} of $\betamax$, assume that the minimum in this equation is achieved by $\sigma^2=\bar{\sigma}^2$. Thus,
\begin{align}\label{eq:betamaxgeq}
\betamax= \textstyle \!\left(\frac{\Psi(\bar{\sigma}^2,\bar{\sigma}^2)}{\bar{\sigma}^2}\right)^{-1} \geq 1.
\end{align}
Here, the inequality comes from \cite[Prop.~4]{GWSS2011} which provides an upper bound for the MSE function $\Psi(\sigma^2,\sigma^2)$ as follows:
	$\Psi(\sigma^2,\sigma^2) \leq \sigma^2, \quad \forall \sigma^2 \geq 0.$
Recall from \fref{sec:path-AMPI} that AMPI decouples the system into a set of $N$ parallel and independent AWGN channels $z_\ell=x_\ell+\sigma Z$ with $Z\sim\setN(0,1)$ and $\sigma^2_t$ being the effective noise variance computed using state evolution equations in \fref{sec:SE}. Hence, using  \cite[Thm.~12]{WV2010} for each AWGN channels, we conclude that if the prior signal distribution $p(x_\ell)$ is continuous and bounded, then the MMSE dimension $\bD$ as defined below will have the value of 1, i.e., 
$\bD(x_\ell,Z)=\lim\limits_{\sigma^2 \to 0}\frac{\Psi(\sigma^2,\sigma^2)}{\sigma^2}=1.$
Using this equation along with the definition of $\betamax$ in \fref{def:betaN0}, we obtain
\begin{align}\label{eq:betamaxleq}
\betamax\!=\!{\min \limits_{\sigma^2>0}\!\left\{\!\left(\!\frac{\Psi(\!\sigma^2,\sigma^2)}{\sigma^2}\!\right)^{\!\!\!-1}\!\right\}\!\!} \leq\! \left( \lim\limits_{\sigma^2 \to 0}\!\!\frac{\Psi(\sigma^2,\sigma^2)}{\sigma^2}\right)^{\!\!\!-1}\!\!\!\!\!=\!1.
\end{align}
From \fref{eq:betamaxgeq} and \fref{eq:betamaxleq}, we have $\betamax=1$.
Additionally by \cite[Lem.~4]{JGMS2015conf}, $\betamin \leq \betamax=1$ which completes the proof.

\section{Proof of \fref{thm:IOfixedD}}\label{app:IOfixedD}

To simplify notation, we denote the PDF of all distributions by $p$. Now to characterize the error probability of IO data detector, we start with the hard-soft assumption
\begin{align}\label{eq:hs}
\mathbb{P}(\hat{s}_\ell^\text{IO} \neq s_\ell) = \mathbb{P} (D(\mathbb{E} (s_\ell | \mathbf{y} , \mathbf{H})) \neq s_\ell).
\end{align}
Based on this assumption,  we have to characterize the joint distribution of $(s_\ell, \mathbb{E} (s_\ell | \mathbf{y, H}))$. Note that in \cite{GV2005} the limiting distribution of $(x_\ell, \mathbb{E} (x_\ell | \mathbf{y, H}))$ has been characterized.  We will use this limiting distribution  to study $(s_\ell, \mathbb{E} (s_\ell | \mathbf{y, H}))$. 

From \fref{eq:sysmodel}, we have that $s_\ell \rightarrow x_\ell \rightarrow \mathbf{y}$ is a Markov chain. This implies that the random variable $q_\ell = \mathbb{E} (s_\ell | \mathbf{y, H})$ which is a function of $\mathbf{y}$ and $\mathbf{H}$ is independent of $s_\ell$ given $x_\ell$. Hence, instead of $(s_\ell, \mathbb{E} (s_\ell | \mathbf{y, H}))$, we can characterize the limiting distribution of $(s_\ell, x_\ell, \mathbb{E} (s_\ell | \mathbf{y,H}))$ which can be written as:
\begin{align}
p(s_\ell, x_\ell, \mathbb{E} (s_\ell | \mathbf{y,H})) & =p(s_\ell, x_\ell)p(\mathbb{E} (s_\ell | \mathbf{y,H})|x_\ell,s_\ell) \\
&=p(s_\ell, x_\ell)p(\mathbb{E} (s_\ell | \mathbf{y,H})|x_\ell).
\end{align}
Since the joint distribution $p(s_\ell, x_\ell)$ is known, characterizing the distribution of $(s_\ell, x_\ell, \mathbb{E} (s_\ell | \bmy, \bH))$ reduces to characterizing the distribution of $(\mathbb{E} (s_\ell | \mathbf{y,H})|x_\ell)$ (or equivalently $(x_\ell, \mathbb{E} (s_\ell | \mathbf{y,H}))$). 
Let us compute $\mathbb{E} (s_\ell | \bmy,\bH)$, which is 
\begin{align}
\mathbb{E} (s_\ell | \bmy,\bH) & \!=\!\! \int\! s_\ell p(s_\ell | \mathbf{y}, \mathbf{H}) \dd s_\ell \\
& \!=\!\! \int\! \mathbb{E}(s_\ell | x_\ell)  p(x_\ell | \bmy, \bH) \dd x_\ell.
\end{align}
Define $L(x_\ell) \triangleq \mathbb{E} (s_\ell | x_\ell)$. Thus, our original problem of characterizing the limiting distribution of $(s_\ell, \mathbb{E} (s_\ell | \mathbf{y, H}))$ is simplified to characterizing the limiting distribution of $(x_\ell, \mathbb{E} (L(x_\ell)| \mathbf{y,H}))$. This latter problem can be solved by the replica method as explained in \cite{GV2005}. Assuming that the assumptions underlying the replica symmetry in \cite{GV2005} are correct, 
we can argue from claim 1 in this paper that
\begin{align}\label{eq:replica}
(x_\ell ,  \mathbb{E} (L(x_\ell)| \mathbf{y,H}) \overset{d}{\rightarrow} (X, \mathbb{E} (L(X) | X+ \tilde{\sigma}Z)),
\end{align}
where $S \sim p(s_\ell)$, $X | S\sim p(x_\ell | s_\ell )$, $Z \sim N(0,1)$ is independent of $S$ and $X$, and $\tilde{\sigma}$ satisfies the fixed point equation 
\begin{eqnarray}\label{eq:fpreplica}
\tilde{\sigma}^2 = N_0+ \beta \Psi(\tilde{\sigma}^2).
\end{eqnarray}
Note that 
$
\mathbb{E} (L(X) | X+ \tilde{\sigma}Z) = \mathbb{E} (\mathbb{E} (S | X) | X+ \tilde{\sigma}Z) = \mathbb{E} (S | X+ \tilde{\sigma}Z)
$.
In other words, \fref{eq:replica} can be written as
\begin{align}
(x_\ell ,  \mathbb{E} (L(x_\ell)| \mathbf{y,H}) \overset{d}{\rightarrow} (X, \mathbb{E} (S | X+ \tilde{\sigma}Z)).
\end{align}
Next, we use this result to characterize the joint limiting distribution of $(s_\ell, x_\ell, \mathbb{E} (s_\ell | \mathbf{y,H}))$. If we define $q_\ell = \mathbb{E} (s_\ell | \mathbf{y, H})$ and $Q= X+ \tilde{\sigma}Z$, then for every $s,x,q \in \mathbb{R}$ we have
$
p_{q_\ell, x_\ell}(q, x) \rightarrow p_{Q,X}(q, x),
$
and, furthermore,
\begin{align}
\textstyle f_{\!s_\ell, x_\ell, q_\ell}   \!(s, x, q) & \!=\!p_{s_\ell | x_\ell}  \!(s|x) p_{q_\ell, x_\ell} \!(q,\!x) \\
& \!=\!  p_{S | X} \!(s |x) p_{q_\ell, x_\ell} \!(q,\!x),
\end{align}
which will converge to $p_{S | X} (s|x) p_{Q,X} (q,x)$. 
Consequently, $(s_\ell, x_\ell, \mathbb{E} (s_\ell | \mathbf{y,H}))$ converges to $(S, X, \mathbb{E} (S | X+ \tilde{\sigma} Z))$ in distribution, which along with markov chain $s_\ell \rightarrow x_\ell \rightarrow \mathbb{E} (s_\ell | \mathbf{y,H})$ leads to the result
$(s_\ell,q_\ell) \!\overset{d}{\rightarrow} \!(S,Q),$
or equivalently, 
\begin{align}\label{eq:QgivenS}
p_{q_\ell | s_\ell} (q|s) \overset{d}{\rightarrow} p_{Q| S} (q |s).
\end{align}
Next, we will use this result to characterize the error probability of IO data detector $\mathbb{P}(\hat{s}_\ell^\IO \neq s_\ell)$.
To simplify the rest of the proof we make several assumptions that are correct for systems MIMO. Suppose $s_\ell \in \setO$ and that the cardinality of this set is finite.
From \fref{eq:hs},  the IO error probability can be written as $ \mathbb{P} (D(q_\ell) \neq s_\ell)$ which converges as follows for a given $s_\ell=s$:
\begin{align}
\mathbb{P}(D(q_\ell) \neq s_\ell  |  s_\ell = s) & = 1- \mathbb{P} (q_\ell \in D^{-1} (s)  |  s_\ell =s) \\
& \rightarrow 1- \mathbb{P} (Q \in D^{-1} (s)  |  S = s).   
\end{align}
The last claim is a consequence of \cite[Lem.~2.2]{van2000asymptotic} which connects convergence in distribution of~\fref{eq:QgivenS} to the convergence in probability above. This relation holds due to the fact that the boundary of $D^{-1}$ has Lebesgue measure zero. Averaging over all values of $s_\ell \in \setO$, we obtain
\begin{align}
& \mathbb{P}(D(q_\ell) \neq s_\ell) = \sum_{s\in \setO} \mathbb{P}(D(q_\ell) \neq s_\ell  |  s_\ell = s) p(s_\ell=s) \\
&\!\rightarrow\! \sum_{s\in \setO} \mathbb{P}(D(Q) \neq S  |  S = s) p(S=s) \! = \! \mathbb{P} (D(Q) \neq S). 
\end{align}
Hence, we have $\mathbb{P}(\hat{s}_\ell^\IO \neq s_\ell)\rightarrow \mathbb{P} (D(Q) \neq S)$.

\section{Proof of Theorem \ref{thm:IOptimality}}\label{app:IOptimality}
Throughout this section, we assume that the random variables $S \sim p(s_\ell), X|S \sim p(x_\ell|s_\ell)$ and $Z\sim\setC\setN(0,1)$ are independent of $X$ and $S$. We start with the following lemma.
\begin{lem}
$\mathbb{P} (S \neq D(X+ \sigma Z) )$ is continuous in $\sigma$.
\end{lem}
%
Note that $\mathbb{P} (S \neq D(X+ \sigma Z) ) = \sum_{s \in \setO}  \mathbb{P} (S \neq D(X+ \sigma Z) | S= s)p(S= s)$. Hence, if we prove that $\mathbb{P} (S \neq D(X+ \sigma Z) | S= s)$ is continuous then, so is $\mathbb{P} (S \neq D(X+ \sigma Z) ) $. Furthermore, $\mathbb{P} (S \neq D(X+ \sigma Z) | S=s) = \mathbb{P} (X+ \sigma Z \in D^{-1} (s) | S=s)$. It is straightforward to write the probability in its integral form and confirm that it is continuous in $\sigma$. 

Suppose that we run AMPI for $t$ iterations and then apply $D$ to $\mathbf{z}^t$ and $\sigma_{t}$ to obtain the signal estimate $\hat{s}_\ell^t$.
Then, according to Lemma \ref{lem:zeta_limit}, the asymptotic error probability of AMPI is
\begin{align}\label{eq:err_AMPI}
\mathbb{P} (s_\ell \neq \hat{s}^t_\ell ) = \mathbb{P} (S \neq D(X+ \sigma_t Z)).
\end{align}
Also, note that the effective noise variance $\sigma_{t}$ is given by the cSE recursion in \fref{eq:SErecursion}; i.e. for $t \to \infty$, $\sigma_{t}$ converges to the solution of AMPI's fixed-point equation as given in \fref{eq:FPsimplified}. Now, since the fixed-point equation of AMPI in~\fref{eq:FPsimplified} coincides with fixed-point equation of the IO data detector in \fref{eq:fpreplica}, we have $\sigma_t \rightarrow \tilde{\sigma}$. 
The rest of the proof is a simple continuity argument with two statements:
\begin{enumerate}
	\item Since $\mathbb{P} (S \neq D(X+ \sigma Z) )$ is a continuous function in $\sigma$, for every $\epsilon>0$ there exists $\Delta \sigma$ such that if $\bar{\sigma} \in (\tilde{\sigma}- \Delta \sigma, \tilde{\sigma}+ \Delta\sigma)$, then $\mathbb{P} (S \neq D(X+ \bar{\sigma} Z) ) < \mathbb{P} (S \neq D(X+ \tilde{\sigma} Z) ) + \epsilon$.
	\item Since $\sigma^t \rightarrow \tilde{\sigma}$ as $t \rightarrow \infty$, we know that there exists  a $t_0$ such that for $t > t_0$, $\sigma^t < \tilde{\sigma} + \Delta \sigma$. 
\end{enumerate}
By combining these two statements, we conclude that for every $\epsilon>0$, there exists  a $t_0$ such that for $t > t_0$
\begin{align}
& \mathbb{P} (s_\ell \neq \hat{s}^{t_0}_\ell) \stackrel{(a)}{=} \mathbb{P} (S \neq D(X+ \sigma_{t_0} Z)) \\
& \qquad < \mathbb{P} (S \neq D(X+ \tilde{\sigma} Z) ) + \epsilon \stackrel{(b)}{=} \mathbb{P} (s_\ell \neq \hat{s}^\IO_\ell) + \epsilon.
\end{align}
Here, $(a)$ and $(b)$ follow from \fref{eq:err_AMPI} and \fref{thm:IOptimality}, respectively. The proof is complete by averaging over all $\ell=1,\ldots,N$.

\balance

\bibliographystyle{IEEEtran}
\bibliography{VIPconfs-jrnls,publishers,VIP}

\begin{thebibliography}{10}
\providecommand{\url}[1]{#1}
\csname url@samestyle\endcsname
\providecommand{\newblock}{\relax}
\providecommand{\bibinfo}[2]{#2}
\providecommand{\BIBentrySTDinterwordspacing}{\spaceskip=0pt\relax}
\providecommand{\BIBentryALTinterwordstretchfactor}{4}
\providecommand{\BIBentryALTinterwordspacing}{\spaceskip=\fontdimen2\font plus
\BIBentryALTinterwordstretchfactor\fontdimen3\font minus
  \fontdimen4\font\relax}
\providecommand{\BIBforeignlanguage}[2]{{%
\expandafter\ifx\csname l@#1\endcsname\relax
\typeout{** WARNING: IEEEtran.bst: No hyphenation pattern has been}%
\typeout{** loaded for the language `#1'. Using the pattern for}%
\typeout{** the default language instead.}%
\else
\language=\csname l@#1\endcsname
\fi
#2}}
\providecommand{\BIBdecl}{\relax}
\BIBdecl

\bibitem{GJMS_2015_conf}
R.~Ghods, C.~Jeon, A.~Maleki, and C.~Studer, ``Optimal large-{MIMO} data
  detection with transmit impairments,'' in \emph{53rd Annual Allerton
  Conference on Communication, Control, and Computing}, Sept. 2015, pp.
  1211--1218.

\bibitem{ABCHLAZ2014}
J.~Andrews, S.~Buzzi, W.~Choi, S.~Hanly, A.~Lozano, A.~Soong, and J.~Zhang,
  ``What will {5G} be?'' \emph{{IEEE} J. Sel. Areas Commun.}, vol.~32, no.~6,
  pp. 1065--1082, Jun. 2014.

\bibitem{WBVSCJD2013}
M.~Wu, B.~Yin, A.~Vosoughi, C.~Studer, J.~Cavallaro, and C.~Dick, ``Approximate
  matrix inversion for high-throughput data detection in the large-scale {MIMO}
  uplink,'' in \emph{Proc. IEEE Int. Symp. Circuits and Syst. (ISCAS)}, May
  2013, pp. 2155--2158.

\bibitem{JGMS2015}
C.~Jeon, R.~Ghods, A.~Maleki, and C.~Studer, ``Optimal data detection in large
  mimo,'' \emph{arXiv preprint arXiv:1811.01917}, 2018.

\bibitem{JGMS2015conf}
------, ``Optimality of large {MIMO} detection via approximate message
  passing,'' in \emph{IEEE Int. Symp. Inf. Theory (ISIT)}, Jun. 2015, pp.
  1227--1231.

\bibitem{wu2014low}
S.~Wu, L.~Kuang, Z.~Ni, J.~Lu, D.~Huang, and Q.~Guo, ``Low-complexity iterative
  detection for large-scale multiuser mimo-ofdm systems using approximate
  message passing,'' \emph{IEEE Journal of Selected Topics in Signal
  Processing}, vol.~8, no.~5, pp. 902--915, 2014.

\bibitem{donoho2009message}
D.~Donoho, A.~Maleki, and A.~Montanari, ``Message-passing algorithms for
  compressed sensing,'' \emph{Proc. Natl. Academy of Sciences (PNAS)}, vol.
  106, no.~45, pp. 18\,914--18\,919, Sept. 2009.

\bibitem{Studer_Tx_OFDM}
C.~Studer, M.~Wenk, and A.~Burg, ``{MIMO} transmission with residual
  transmit-{RF} impairments,'' in \emph{Int. ITG Workshop on Smart Antennas
  (WSA)}, Feb. 2010, pp. 189--196.

\bibitem{Tx_Replica}
M.~Vehkaper{\"a}, T.~Riihonen, M.~A. Girnyk, E.~Bj{\"o}rnson, M.~Debbah, L.~K.
  Rasmussen, and R.~Wichman, ``Asymptotic analysis of {SU-MIMO} channels with
  transmitter noise and mismatched joint decoding,'' \emph{{IEEE} Trans.
  Commun.}, vol.~32, no.~6, pp. 1065--1082, Mar. 2015.

\bibitem{schenk2008rf}
T.~C. Schenk, \emph{{RF} imperfections in high-rate wireless systems: impact
  and digital compensation}.\hskip 1em plus 0.5em minus 0.4em\relax Springer
  Netherlands, 2008.

\bibitem{donoho2006}
D.~L. Donoho, ``Compressed sensing,'' \emph{{IEEE} Trans. Inf. Theory},
  vol.~52, no.~4, pp. 1289--1306, Apr. 2006.

\bibitem{candes2006c}
E.~J. {Cand\`es}, J.~Romberg, and T.~Tao, ``Robust uncertainty principles:
  Exact signal reconstruction from highly incomplete frequency information,''
  \emph{{IEEE} Trans. Inf. Theory}, vol.~52, no.~2, pp. 489--509, Feb. 2006.

\bibitem{davenport2012proscons}
M.~Davenport, J.~Laska, J.~Treichler, and R.~Baraniuk, ``The pros and cons of
  compressive sensing for wideband signal acquisition: noise folding versus
  dynamic range,'' \emph{{IEEE} Trans. Signal Process.}, vol.~60, no.~9, pp.
  4628--4642, Sep. 2012.

\bibitem{treichler2009application}
J.~Treichler, M.~Davenport, and R.~Baraniuk, ``Application of compressive
  sensing to the design of wideband signal acquisition receivers,''
  \emph{US/Australia Joint Work. Defense Apps. of Signal Processing (DASP),
  Lihue, Hawaii}, vol.~5, 2009.

\bibitem{arias2011noise}
E.~Arias-Castro and Y.~C. Eldar, ``Noise folding in compressed sensing,''
  \emph{IEEE Signal Processing Letters}, vol.~18, no.~8, pp. 478--481, 2011.

\bibitem{V1998}
S.~Verd\'u, \emph{Multiuser Detection}, 1st~ed.\hskip 1em plus 0.5em minus
  0.4em\relax Cambridge University Press, 1998.

\bibitem{cooper1990computational}
G.~F. Cooper, ``The computational complexity of probabilistic inference using
  bayesian belief networks,'' \emph{Artificial intelligence}, vol.~42, no. 2-3,
  pp. 393--405, 1990.

\bibitem{donoho2009}
D.~L. Donoho, A.~Maleki, and A.~Montanari, ``Message-passing algorithms for
  compressed sensing,'' \emph{Proc. Natl. Acad. Sci. USA}, vol. 106, no.~45,
  pp. 18\,914--18\,919, Nov. 2009.

\bibitem{andreaGMCS}
A.~Montanari, \emph{Graphical models concepts in compressed sensing, Compressed
  Sensing (Y.C. Eldar and G. Kutyniok, eds.)}.\hskip 1em plus 0.5em minus
  0.4em\relax Cambridge University Press, 2012.

\bibitem{Maleki2010phd}
A.~Maleki, ``Approximate message passing algorithms for compressed sensing,''
  Ph.D. dissertation, Stanford University, Jan. 2011.

\bibitem{DMM10a}
D.~Donoho, A.~Maleki, and A.~Montanari, ``Message passing algorithms for
  compressed sensing: {I}. {M}otivation and construction,'' in \emph{Proc. IEEE
  Inf. Theory Workshop (ITW)}, Jan. 2010, pp. 1--5.

\bibitem{DMM10b}
------, ``Message passing algorithms for compressed sensing: {II}. {A}nalysis
  and validation,'' in \emph{Proc. IEEE Inf. Theory Workshop (ITW)}, Jan. 2010,
  pp. 1--5.

\bibitem{studer2011system}
C.~Studer, M.~Wenk, and A.~Burg, ``System-level implications of residual
  transmit-{RF} impairments in {MIMO} systems,'' in \emph{Proc. European Conf.
  on Antennas and Propagation (EUCAP)}, Apr. 2011, pp. 2686--2689.

\bibitem{schenk2005performance}
T.~C. Schenk, P.~F. Smulders, and E.~R. Fledderus, ``Performance of {MIMO OFDM}
  systems in fading channels with additive {TX} and {RX} impairments,'' in
  \emph{Proc. IEEE BENELUX/DSP Valley Signal Process. Symp.}, Apr. 2005, pp.
  41--44.

\bibitem{goransson2008effect}
B.~Goransson, S.~Grant, E.~Larsson, and Z.~Feng, ``Effect of transmitter and
  receiver impairments on the performance of {MIMO} in {HSDPA},'' in
  \emph{Proc. IEEE Int. Workshop Signal Process. Advances Wireless Commun.
  (SPAWC)}, Jul. 2008, pp. 496--500.

\bibitem{suzuki2008transmitter}
H.~Suzuki, T.~V.~A. Tran, I.~B. Collings, G.~Daniels, and M.~Hedley,
  ``Transmitter noise effect on the performance of a {MIMO-OFDM} hardware
  implementation achieving improved coverage,'' \emph{{IEEE} J. Sel. Areas
  Commun.}, vol.~26, no.~6, pp. 867--876, Aug. 2008.

\bibitem{suzuki2009practical}
H.~Suzuki, I.~B. Collings, M.~Hedley, and G.~Daniels, ``Practical performance
  of {MIMO-OFDM-LDPC} with low complexity double iterative receiver,'' in
  \emph{Proc. IEEE Int. Symp. Personal, Indoor, Mobile Radio Commun. (PIMRC)},
  Sep. 2009, pp. 2469--2473.

\bibitem{gonzalez2011impact}
J.~P. Gonz{\'a}lez-Coma, P.~M. Castro, and L.~Castedo, ``Impact of transmit
  impairments on multiuser {MIMO} non-linear transceivers,'' in \emph{Int. ITG
  Workshop on Smart Antennas (WSA)}, Feb. 2011, pp. 1--8.

\bibitem{gonzalez2011transmit}
------, ``Transmit impairments influence on the performance of {MIMO} receivers
  and precoders,'' in \emph{Proc. European. Wireless Conf. -- Sustainable
  Wireless Technol. (European Wireless)}, Apr. 2011, pp. 1--8.

\bibitem{bjornson2013capacity}
E.~Bjornson, P.~Zetterberg, M.~Bengtsson, and B.~Ottersten, ``Capacity limits
  and multiplexing gains of {MIMO} channels with transceiver impairments,''
  \emph{{IEEE} Commun. Lett.}, vol.~17, no.~1, Jan. 2013.

\bibitem{zhang2014mimo}
X.~Zhang, M.~Matthaiou, E.~Bjornson, M.~Coldrey, and M.~Debbah, ``On the {MIMO}
  capacity with residual transceiver hardware impairments,'' in \emph{Proc.
  IEEE Int. Conf. Commun. (ICC)}, Jun. 2014, pp. 5299--5305.

\bibitem{NFalg2015}
S.~Peter, M.~Artina, and M.~Fornasier, ``Damping noise-folding and enhanced
  support recovery in compressed sensing,'' \emph{IEEE Transactions on Signal
  Processing}, vol.~63, no.~22, pp. 5990--6002, Nov 2015.

\bibitem{GAMP2011}
S.~Rangan, ``Generalized approximate message passing for estimation with random
  linear mixing,'' \emph{CoRR}, vol. abs/1010.5141, 2010.

\bibitem{FG1}
F.~R. Kschischang, B.~J. Frey, and H.~A. Loeliger, ``Factor graphs and the
  sum-product algorithm,'' \emph{IEEE Transactions on Information Theory},
  vol.~47, no.~2, pp. 498--519, Feb 2001.

\bibitem{donoho2011design}
D.~L. Donoho, A.~Maleki, and A.~Montanari, ``How to design message passing
  algorithms for compressed sensing,'' \emph{preprint}, 2011.

\bibitem{weiss2000correctness}
Y.~Weiss, ``Correctness of local probability propagation in graphical models
  with loops,'' \emph{Neural computation}, vol.~12, no.~1, pp. 1--41, 2000.

\bibitem{bayatimontanari}
M.~Bayati and A.~Montanari, ``The dynamics of message passing on dense graphs,
  with applications to compressed sensing,'' \emph{{IEEE} Trans. Inf. Theory},
  vol.~57, no.~2, pp. 764--785, Feb. 2011.

\bibitem{ZMWL2015}
L.~Zheng, A.~Maleki, H.~Weng, X.~Wang, and T.~Long, ``Does $\ell_
  p$-minimization outperform $\ell_1$-minimization?'' \emph{IEEE Transactions
  on Information Theory}, vol.~63, no.~11, pp. 6896--6935, 2017.

\bibitem{barbier2020mutual}
J.~Barbier, N.~Macris, M.~Dia, and F.~Krzakala, ``Mutual information and
  optimality of approximate message-passing in random linear estimation,''
  \emph{IEEE Transactions on Information Theory}, 2020.

\bibitem{GV2005}
D.~Guo and S.~Verd\'u, ``Randomly spread {CDMA}: {Asymptotics} via statistical
  physics,'' \emph{{IEEE} Trans. Inf. Theory}, vol.~51, no.~6, pp. 1983--2010,
  Jun. 2005.

\bibitem{candes2006stable}
E.~J. Candes, J.~K. Romberg, and T.~Tao, ``Stable signal recovery from
  incomplete and inaccurate measurements,'' \emph{Commun. Pure Appl. Math.},
  vol.~59, no.~8, pp. 1207--1223, 2006.

\bibitem{mousavi2015consistent}
A.~Mousavi, A.~Maleki, R.~G. Baraniuk \emph{et~al.}, ``Consistent parameter
  estimation for lasso and approximate message passing,'' \emph{The Annals of
  Statistics}, vol.~46, no.~1, pp. 119--148, 2018.

\bibitem{stein1981estimation}
C.~M. Stein, ``Estimation of the mean of a multivariate normal distribution,''
  \emph{The annals of Statistics}, pp. 1135--1151, 1981.

\bibitem{polak1969note}
E.~Polak and G.~Ribiere, ``Note sur la convergence de m{\'e}thodes de
  directions conjugu{\'e}es,'' \emph{ESAIM: Mathematical Modelling and
  Numerical Analysis-Mod{\'e}lisation Math{\'e}matique et Analyse
  Num{\'e}rique}, vol.~3, no.~R1, pp. 35--43, 1969.

\bibitem{frandsen1999unconstrained}
P.~E. Frandsen, K.~Jonasson, H.~B. Nielsen, and O.~Tingleff, ``Unconstrained
  optimization,'' 1999.

\bibitem{van2000asymptotic}
A.~W. Van~der Vaart, \emph{Asymptotic statistics}.\hskip 1em plus 0.5em minus
  0.4em\relax Cambridge university press, 2000, vol.~3.

\bibitem{GWSS2011}
D.~Guo, Y.~Wu, S.~Shamai, and S.~Verd\'u, ``Estimation in {Gaussian} noise:
  Properties of the minimum mean-square error,'' \emph{{IEEE} Trans. Inf.
  Theory}, vol.~57, no.~4, pp. 2371--2385, Apr. 2011.

\bibitem{WV2010}
Y.~Wu and S.~Verd\'u, ``{MMSE} dimension,'' \emph{Proc. IEEE Int. Symp. Inf.
  Theory (ISIT)}, pp. 1463--1467, June 2010.

\bibitem{gradshteyn2014table}
I.~S. Gradshteyn and I.~M. Ryzhik, \emph{Table of integrals, series, and
  products}.\hskip 1em plus 0.5em minus 0.4em\relax Academic press, 2014.

\bibitem{bagnoli2005log}
M.~Bagnoli and T.~Bergstrom, ``Log-concave probability and its applications,''
  \emph{Economic theory}, vol.~26, no.~2, pp. 445--469, 2005.

\bibitem{boyd2004convex}
S.~Boyd and L.~Vandenberghe, \emph{Convex optimization}.\hskip 1em plus 0.5em
  minus 0.4em\relax Cambridge university press, 2004.

\end{thebibliography}

\balance

\clearpage

\onecolumn

\begin{center}
\Large Supplementary Derivations for  \\[0.2cm] ``Optimal Data Detection and Signal Estimation  in Systems with Input Noise''
\end{center}

 \section{Derivation of $\mathsf{F}$ and $\mathsf{G}$ in Compressed Sensing}\label{app:FG_CS}
 
 We start by deriving the PDF for  $\bmx=\bms+\bme$, where $s_\ell\sim p(s_\ell)$ and $e_\ell\sim\setN(0,\NT)$ for all $\ell=1,\ldots,N$. For $p(s_\ell)=\frac{\lambda}{2}\exp\left(-\lambda\abs{s_\ell}\right)$, $e_\ell\sim\setN(0,\NT)$, and with the relation $x_\ell=s_\ell+e_\ell$, we have
 \begin{align}
 p(x_\ell)
 = &\, \int_\reals\frac{1}{\sqrt{2\pi\NT}}\exp\left(-\frac{1}{2\NT}(s_\ell-x_\ell)^2\right)\frac{\lambda}{2}\exp\left(-\lambda\abs{s_\ell}\right)\dd s_\ell\\
 = &\, \frac{\lambda}{2}\exp\left(\lambda x_\ell+\frac{\lambda^2\NT}{2}\right)\int_{-\infty}^0\frac{1}{\sqrt{2\pi\NT}}\exp\left(-\frac{1}{2\NT}(x_\ell-(s_\ell+\lambda\NT))^2\right)\dd s_\ell \notag \\
 &+ \frac{\lambda}{2}\exp\left(-\lambda x_\ell+\frac{\lambda^2\NT}{2}\right)\int_0^\infty\frac{1}{\sqrt{2\pi\NT}}\exp\left(-\frac{1}{2\NT}(x_\ell-(s_\ell-\lambda\NT))^2\right)\dd s_\ell.
 \end{align}
 With the Q-function $Q(x)=\int_x^\infty\frac{1}{\sqrt{2\pi}}\exp\left(-\frac{t^2}{2}\right)\dd t$, the above integral ``simplifies'' to
 \begin{align}\label{eq:prior}
 p(x_\ell)
 &= \frac{\lambda}{2}\exp\left(\frac{\lambda^2\NT}{2}\right)
 \left(
 \exp(\lambda x_\ell)Q\left(\frac{x_\ell+\lambda\NT}{\sqrt{\NT}}\right)
 +\exp(-\lambda x_\ell)\left(1-Q\left(\frac{x_\ell-\lambda\NT}{\sqrt{\NT}}\right)\right)
 \right).
 \end{align}
 Note that as $\NT\rightarrow0$ ($\bmx \to \bms$), $Q\left(\frac{x_\ell}{\sqrt{\NT}}\right)\rightarrow\frac{1-\sign(x_\ell)}{2}$, we have that 
 \begin{align}
 \lim_{\NT\rightarrow0}p(x_\ell) &\rightarrow \frac{\lambda}{2}
 \left(
 \exp(\lambda x_\ell)\left(\frac{1-\sign(x_\ell)}{2}\right)
 +\exp(-\lambda x_\ell)\left(\frac{\sign(x_\ell)+1}{2}\right)\right)\\
 &=\frac{\lambda}{2}\exp\left(-\lambda\abs{x_\ell}\right),
 \end{align}
 which is the distribution of $s_\ell$ as expected.
 
 With the PDF of the new prior $p(\bmx)=\prod_\ell^N p(x_\ell)$ where $p(x_\ell)$ is given in \fref{eq:prior}, we now proceed to computing  the functions relevant to the cB-AMP algorithm given in \fref{alg:AMP}.

 We start by computing the posterior $p(x_{\ell}\vert z^t_{\ell},\sigma_t^2)$ in \fref{eq:Ffunc}, where we have $p(z^t_{\ell}\vert x_{\ell},\sigma_t^2)\sim\setC\setN(x_{\ell},\sigma_t^2)$: 
 \begin{align}\label{eq:posterior}
 p(x_{\ell}\vert z^t_{\ell},\sigma_t^2)
 &=\frac{p(z^t_{\ell}\vert x_\ell,\sigma_t^2)p(x_\ell)}{\int_\reals p(z^t_{\ell}\vert x_\ell,\sigma_t^2)p(x_\ell)\dd x_\ell}\\
 &= \frac{1}{\sqrt{2\pi\sigma_t^2}}\frac{\exp\left(-\frac{1}{2\sigma_t^2}(z^t_{\ell}-x_\ell)^2\right)p(x_\ell)}{p(z^t_{\ell})},
 \end{align}
 where $p(z^t_{\ell})$ is derived as follows.
 %
%
 \begin{align} 
p(z^t_{\ell})=&\int_\reals p(z^t_{\ell}\vert x_\ell,\sigma_t^2)p(x_\ell)\dd x_\ell\\
\stackrel{(a)}{=} &\, \frac{\lambda}{2}\exp\left(\frac{\lambda^2\NT}{2}\right)\nonumber
 \int_\reals\exp(\lambda x_\ell)Q\left(\frac{x_\ell+\lambda\NT}{\sqrt{\NT}}\right)
 \frac{1}{\sqrt{2\pi\sigma_t^2}}\exp\left(-\frac{1}{2\sigma_t^2}(x_\ell-z^t_{\ell})^2\right)\dd x_\ell\\ 
 &+ \frac{\lambda}{2}\exp\left(\frac{\lambda^2\NT}{2}\right)
 \int_\reals\exp(-\lambda x_\ell)\left(1-Q\left(\frac{x_\ell-\lambda\NT}{\sqrt{\NT}}\right)\right)
 \frac{1}{\sqrt{2\pi\sigma_t^2}}\exp\left(-\frac{1}{2\sigma_{t}^2}(x_\ell-z^t_{\ell})^2\right)\dd x_\ell\\\nonumber
\stackrel{(b)}{=} &\, \frac{\lambda}{2}\exp\left(\lambda z^t_{\ell}+\frac{\lambda^2(\NT+\sigma_t^2)}{2}\right)Q\left(\frac{z^t_{\ell}+\lambda(\NT+\sigma_t^2)}{\sqrt{\NT+\sigma_t^2}}\right)\\
 &+\frac{\lambda}{2}\exp\left(-\lambda z^t_{\ell}+\frac{\lambda^2(\NT+\sigma_t^2)}{2}\right)\left(1-Q\left(\frac{z^t_{\ell}-\lambda(\NT+\sigma_t^2)}{\sqrt{\NT+\sigma_t^2}}\right)\right),
 \end{align}
 here, (a) follows from replacing $p(x_\ell)$ from \fref{eq:prior} and $p(z^t_{\ell}\vert x_{\ell},\sigma_t^2)\sim\setC\setN(x_{\ell},\sigma_t^2)$ from its definition in \fref{alg:AMP}. (b) follows from a precomputed integral relation in \cite[Eq. 8.259]{gradshteyn2014table} as follows:
 \begin{align}\label{eq:simplification1}
 \int_\reals \frac{1}{\sqrt{2\pi\tau}}
 Q\left(b+\frac{x}{\sigma}\right)
 \exp\left(-\frac{1}{2\tau}x^2\right)\dd x = Q\left(\frac{b\sigma}{\sqrt{\sigma^2+\tau}}\right),
 \end{align}

 We now compute the posterior mean and variance functions $\mathsf{F}(z^t_{\ell},\sigma_t^2)$ and $\mathsf{G}(z^t_\ell,\sigma_t^2)$ as defined in \fref{eq:Ffunc} and \fref{eq:Gfunc}.

\subsubsection{Posterior Mean}
The posterior mean can be derived as follows:
 \begin{align} 
 \mathsf{F}(z^t_{\ell},\sigma_t^2) =&\, \int_\reals x_\ell p(x_\ell \vert z^t_{\ell},\sigma_t^2)\dd x_\ell\\\label{eq:F1}
 = &\, \frac{\lambda\exp\left(\frac{\lambda^2\NT}{2}\right)}{2p(z^t_{\ell})}
 \int_\reals \frac{x_{\ell}\exp(\lambda x_{\ell})}{\sqrt{2\pi\sigma_t^2}}\exp\left(-\frac{1}{2\sigma_t^2}(x_{\ell}-z^t_{\ell})^2\right)Q\left(\frac{x_{\ell}+\lambda\NT}{\sqrt{\NT}}\right)\dd x_\ell \\
 &+ \frac{\lambda\exp\left(\frac{\lambda^2\NT}{2}\right)}{2p(z^t_{\ell})}
 \int_\reals \frac{x_{\ell}\exp(-\lambda x_{\ell})}{\sqrt{2\pi\sigma_t^2}}\exp\left(-\frac{1}{2\sigma_t^2}(x_{\ell}-z^t_{\ell})^2\right)\left(1-Q\left(\frac{x_{\ell}-\lambda\NT}{\sqrt{\NT}}\right)\right)\dd x_{\ell}. \label{eq:F2}
 \end{align}
 We first simplify \fref{eq:F1} with $K=\frac{\lambda\exp\left(\frac{\lambda^2\NT}{2}\right)}{2p(z^t_{\ell})}$ which yields
 \begin{align}
 \fref{eq:F1}= K\exp\left(\lambda z^t_{\ell}+\frac{\lambda^2\sigma_t^2}{2}\right)\int_\reals \frac{x_\ell}{\sqrt{2\pi\sigma_t^2}}\exp\left(-\frac{1}{2\sigma_t^2}(x_\ell-(z^t_{\ell}+\lambda\sigma_t^2))^2\right)Q\left(\frac{x_\ell+\lambda\NT}{\sqrt{\NT}}\right)\dd x_\ell.
 \end{align}
 
 We start with the following simplification, which is obtained by integration by parts:
 \begin{align}\notag
 &-\int_\reals \frac{x-u}{\sigma_t^2}\frac{1}{\sqrt{2\pi\sigma_t^2}}\exp\left(-\frac{1}{2\sigma_t^2}(x-u)^2\right)Q\left(\frac{x-v}{\sqrt{\NT}}\right)\dd x \label{eq:simplification2}
 \\&\qquad\qquad =\frac{\sqrt{\NT}}{\sqrt{2\pi(\sigma_t^2+\NT)}}\exp\left(-\frac{1}{2(\sigma_t^2+\NT)}(u-v)^2)\right).
 \end{align}
 Therefore, we have
 \begin{align}
 &\int_\reals \frac{x_{\ell}}{\sqrt{2\pi\sigma_t^2}}\exp\left(-\frac{1}{2\sigma_t^2}(x_{\ell}-(z^t_{\ell}+\lambda\sigma_t^2))^2\right)Q\left(\frac{x_{\ell}+\lambda\NT}{\sqrt{\NT}}\right)\dd x_{\ell}  \\
 =&\,  (z^t_{\ell}+\lambda\sigma_t^2)\int_\reals\frac{1}{\sqrt{2\pi\sigma_t^2}}\exp\left(-\frac{1}{2\sigma_t^2}(x_{\ell}-(z^t_{\ell}+\lambda\sigma_t^2))^2\right)Q\left(\frac{x_{\ell}+\lambda\NT}{\sqrt{\NT}}\right)\dd x_{\ell} \notag \\
 &+ \sigma_t^2 \int_\reals \frac{x_{\ell}-(z^t_{\ell}+\lambda\sigma_t^2)}{\sigma_t^2}\frac{1}{\sqrt{2\pi\sigma_t^2}}\exp\left(-\frac{1}{2\sigma_t^2}(x_{\ell}-(z^t_{\ell}+\lambda\sigma_t^2))^2\right)Q\left(\frac{x_{\ell}+\lambda\NT}{\sqrt{\NT}}\right)\dd x_{\ell}\\
\stackrel{(a)}{=} & \, (z^t_{\ell}+\lambda\sigma_t^2)Q\left(\frac{z^t_{\ell}+\lambda(\NT+\sigma_t^2)}{\sqrt{\NT+\sigma_t^2}}\right) \notag \\
 &\stackrel{(b)}{-}\frac{\sigma_t^2\sqrt{\NT}}{\sqrt{2\pi(\sigma_t^2+\NT)}}\exp\left(-\frac{1}{2(\sigma_t^2+\NT)}(z^t_{\ell}+\lambda(\NT+\sigma_t^2))^2)\right),
 \end{align}
 where (a) follows from \fref{eq:simplification1} and (b) follows from \fref{eq:simplification2}. Similarly, we have \fref{eq:F2} as
 \begin{align}
 \fref{eq:F2}=K\exp\left(-\lambda z^t_{\ell}+\frac{\lambda^2\sigma_t^2}{2}\right)\int_\reals \frac{x_{\ell}}{\sqrt{2\pi\sigma_t^2}}\exp\left(-\frac{1}{2\sigma_t^2}(x_{\ell}-(z^t_{\ell}-\lambda\sigma_t^2))^2\right)\left(1-Q\left(\frac{x_{\ell}-\lambda\NT}{\sqrt{\NT}}\right)\right)\dd x_{\ell}.
 \end{align}
 This expression can be simplified to
 \begin{align}
 &\int_\reals \frac{x_{\ell}}{\sqrt{2\pi\sigma_t^2}}\exp\left(-\frac{1}{2\sigma_t^2}(x_{\ell}-(z^t_{\ell}-\lambda\sigma_t^2))^2\right)\left(1-Q\left(\frac{x_{\ell}-\lambda\NT}{\sqrt{\NT}}\right)\right)\dd x_{\ell} \notag \\
 = &\, z^t_{\ell}-\lambda\sigma_t^2-  \int_\reals \frac{x_{\ell}}{\sqrt{2\pi\sigma_t^2}}\exp\left(-\frac{1}{2\sigma_t^2}(x_{\ell}-(z^t_{\ell}-\lambda\sigma_t^2))^2\right)Q\left(\frac{x_{\ell}-\lambda\NT}{\sqrt{\NT}}\right)\dd x_{\ell}\\
 = &\,z^t_{\ell}-\lambda\sigma_t^2 - \sigma_t^2 \int_\reals \frac{x_{\ell}-(z^t_{\ell}-\lambda\sigma_t^2)}{\sigma_t^2}\frac{1}{\sqrt{2\pi\sigma_t^2}}\exp\left(-\frac{1}{2\sigma_t^2}(x_{\ell}-(z^t_{\ell}-\lambda\sigma_t^2))^2\right)Q\left(\frac{x_{\ell}-\lambda\NT}{\sqrt{\NT}}\right)\dd x_{\ell}\notag\\
 &-(z^t_{\ell}-\lambda\sigma_t^2)\int_\reals\frac{1}{\sqrt{2\pi\sigma_t^2}}\exp\left(-\frac{1}{2\sigma_t^2}(x_{\ell}-(z^t_{\ell}-\lambda\sigma_t^2))^2\right)Q\left(\frac{x_{\ell}-\lambda\NT}{\sqrt{\NT}}\right)\dd x_{\ell}\\
 = &\, z^t_{\ell}-\lambda\sigma_t^2 + \frac{\sigma_t^2\sqrt{\NT}}{\sqrt{2\pi(\sigma_t^2+\NT)}}\exp\left(-\frac{1}{2(\sigma_t^2+\NT)}(z^t_{\ell}-\lambda(\NT+\sigma_t^2))^2)\right)\notag \\
 &-(z^t_{\ell}-\lambda\sigma_t^2)Q\left(\frac{z^t_{\ell}-\lambda(\NT+\sigma_t^2)}{\sqrt{\NT+\sigma_t^2}}\right).
 \end{align}
 
 Therefore, with $K_1 = \frac{\lambda}{2p(z^t_{\ell})}
 \exp\left(\frac{\lambda^2(\NT+\sigma_t^2)}{2}\right)\exp\left(\lambda z^t_{\ell}\right)$ and $K_2=\frac{\lambda}{2p(z^t_{\ell})}
 \exp\left(\frac{\lambda^2(\NT+\sigma_t^2)}{2}\right)\exp\left(-\lambda z^t_{\ell}\right)$ we have
 \begin{align}
 \mathsf{F}(z^t_{\ell},\sigma_t^2)&=K_1(z^t_{\ell}+\lambda\sigma_t^2)Q\left(\frac{z^t_{\ell}+\lambda(\NT+\sigma_t^2)}{\sqrt{\NT+\sigma_t^2}}\right) \notag \\
 &-K_1\frac{\sigma_t^2\sqrt{\NT}}{\sqrt{2\pi(\sigma_t^2+\NT)}}\exp\left(-\frac{(z^t_{\ell}+\lambda(\NT+\sigma_t^2))^2}{2(\sigma_t^2+\NT)}\right) \notag \\
 &+K_2(z^t_{\ell}-\lambda\sigma_t^2) \notag \\
 &+K_2\frac{\sigma_t^2\sqrt{\NT}}{\sqrt{2\pi(\sigma_t^2+\NT)}}\exp\left(-\frac{(z^t_{\ell}-\lambda(\NT+\sigma_t^2))^2}{2(\sigma_t^2+\NT)}\right) \notag \\
 &-K_2(z^t_{\ell}-\lambda\sigma_t^2)Q\left(\frac{z^t_{\ell}-\lambda(\NT+\sigma_t^2)}{\sqrt{\NT+\sigma_t^2}}\right),
 \end{align}
 which can be simplified to 
 \begin{align}
 \mathsf{F}(z^t_{\ell},\sigma_t^2)=K_1(z^t_{\ell}+\lambda\sigma_t^2)Q\left(\frac{z^t_{\ell}+\lambda(\NT+\sigma_t^2)}{\sqrt{\NT+\sigma_t^2}}\right)+K_2(z^t_{\ell}-\lambda\sigma_t^2)Q\left(\frac{-z^t_{\ell}+\lambda(\NT+\sigma_t^2)}{\sqrt{\NT+\sigma_t^2}}\right).
 \end{align}
Since this function contains a multiplication of the  Q-function with exponentials, it is numerically unstable to compute. In order to compute this quantity one can use the more stable function $\mathrm{erfcx}(x)=x^2 \mathrm{erfc}(x)$, where $\mathrm{erfc}(x)$ is the error function that satisfies
 \begin{align}
 \mathrm{erfc}(x)=2Q(\sqrt{2}x)
 \end{align}
 By replacing the values of $K_1$ and $K_2$ and Q-function, we obtain
 \begin{align}
 \mathsf{F}(z^t_{\ell},\sigma_t^2)=z^t_{\ell}+\lambda \sigma_t^2 \eta,
 \end{align}
 where 
 \begin{align}
 \eta&=\frac{\mathrm{erfcx}(\alpha)-\mathrm{erfcx}(\beta)}{\mathrm{erfcx}(\alpha)+\mathrm{erfcx}(\beta)}\\
 \alpha&=\frac{z^t_{\ell}+\lambda(\NT+\sigma_t^2)}{\sqrt{2(\NT+\sigma_t^2)}}\\
 \beta&=\frac{-z^t_{\ell}+\lambda(\NT+\sigma_t^2)}{\sqrt{2(\NT+\sigma_t^2)}}.
 \end{align}
 
 \subsubsection{Posterior Variance}
The posterior variance can be computed as follows:
 \begin{align}
 \mathsf{G}(z^t_{\ell},\sigma_t^2)&=\int_\reals x_{\ell}^2 p(x_{\ell}\vert z^t_{\ell})\dd x_{\ell} - \mathsf{F}^2(z^t_{\ell},\sigma_t^2).
 \end{align}
Here, we have
 \begin{align}
 \int_\reals x_{\ell}^2 p(x_{\ell}\vert z^t_{\ell})\dd x_{\ell} 
 = &\,
 \frac{\lambda\exp\left(\frac{\lambda^2\NT}{2}\right)}{2p(z^t_{\ell})}
 \int_\reals \frac{x_{\ell}^2\exp(\lambda x_{\ell})}{\sqrt{2\pi\sigma_t^2}}\exp\left(-\frac{1}{2\sigma_t^2}(x_{\ell}-z^t_{\ell})^2\right)Q\left(\frac{x_{\ell}+\lambda\NT}{\sqrt{\NT}}\right)\dd x_{\ell}\label{eq:G1}\\\label{eq:G2}
 &+ \frac{\lambda\exp\left(\frac{\lambda^2\NT}{2}\right)}{2p(z^t_{\ell})}
 \int_\reals \frac{x_{\ell}^2\exp(-\lambda x_{\ell})}{\sqrt{2\pi\sigma_t^2}}\exp\left(-\frac{1}{2\sigma_t^2}(x_{\ell}-z^t_{\ell})^2\right)Q\left(\frac{-x_{\ell}+\lambda\NT}{\sqrt{\NT}}\right)\dd x_{\ell}.
 \end{align}
 Let us first compute \fref{eq:G1}, which yields
 \begin{align}
 \fref{eq:G1}=&\, K_1 \int_\reals \frac{x_{\ell}^2}{\sqrt{2\pi\sigma_t^2}} \exp\left(-\frac{ (x_{\ell}-(z^t_{\ell}+\lambda \sigma_t^2))^2}{2\sigma_t^2}\right) Q\left( \frac{x_{\ell}+\lambda \NT}{\sqrt{\NT}}\right) \dd x_{\ell}\\
 =&\,K_1 \int_{\reals} \frac{(t+(z^t_{\ell}+\lambda \sigma_t^2))^2}{\sqrt{2\pi\sigma_t^2}} \exp \left(-\frac{t^2}{2\sigma_t^2}\right) Q\left( \frac{t+z^t_{\ell}+\lambda(\sigma_t^2+\NT)}{\sqrt{\NT}}\right) \dd t\\\label{eq:G11}
 = &\,K_1 \int_{\reals} \frac{t^2}{\sqrt{2\pi\sigma_t^2}} \exp \left(-\frac{t^2}{2\sigma_t^2}\right) Q\left( \frac{t+z^t_{\ell}+\lambda(\sigma_t^2+\NT)}{\sqrt{\NT}}\right) \dd t\\\label{eq:G12}
 &+K_1 (z^t_{\ell}+\lambda \sigma_t^2)^2 \int_{\reals} \frac{1}{\sqrt{2\pi\sigma_t^2}} \exp \left(-\frac{t^2}{2\sigma_t^2}\right) Q\left( \frac{t+z^t_{\ell}+\lambda(\sigma_t^2+\NT)}{\sqrt{\NT}}\right) \dd t \\\label{eq:G13}
 &+K_1 \times 2 (z^t_{\ell}+\lambda \sigma_t^2) \int_{\reals} \frac{t}{\sqrt{2\pi\sigma_t^2}} \exp \left(-\frac{t^2}{2\sigma_t^2}\right) Q\left( \frac{t+z^t_{\ell}+\lambda(\sigma_t^2+\NT)}{\sqrt{\NT}}\right) \dd t.
 \end{align}
 
 Using the relations \cite[Eq. 8.259]{gradshteyn2014table}, \fref{eq:simplification1} and \fref{eq:simplification2} respectively for \fref{eq:G11}, \fref{eq:G12}  and \fref{eq:G13} we obtain
 \begin{align}
 \fref{eq:G1}=&\,K_1\left( \sigma_t^2+(z^t_{\ell}+\lambda \sigma_t^2)^2\right) Q\left( \frac{z^t_{\ell}+\lambda(\NT+\sigma_t^2)}{\sqrt{(\NT+\sigma_t^2)}} \right) \notag \\
 &+K_1 \left( \frac{\sigma_t^4 (z^t_{\ell}+\lambda(\NT+\sigma_t^2))}{\sqrt{2 \pi}(\NT+\sigma_t^2)^{\frac{3}{2}}} - \frac{2\sigma_t^2 (z^t_{\ell}+\lambda \sigma_t^2)}{\sqrt{2 \pi (\NT+\sigma_t^2)}} \right) \exp \left( -\frac{(z^t_{\ell}+\lambda(\NT+\sigma_t^2))^2}{2(\sigma_t^2+\NT)} \right).
 \end{align}
 Similarly, we have for \fref{eq:G2}  the following expression:
 \begin{align}
 \fref{eq:G2}=&\, K_2\left( \sigma_t^2+(z^t_{\ell}-\lambda \sigma_t^2)^2\right) Q\left( \frac{-z^t_{\ell}+\lambda(\NT+\sigma_t^2)}{\sqrt{(\NT+\sigma_t^2)}} \right) \notag \\
 &-K_2 \left( \frac{\sigma_t^4 (z^t_{\ell}-\lambda(\NT+\sigma_t^2))}{\sqrt{2 \pi}(\NT+\sigma_t^2)^{\frac{3}{2}}} - \frac{2\sigma_t^2 (z^t_{\ell}-\lambda \sigma_t^2)}{\sqrt{2 \pi (\NT+\sigma_t^2)}} \right) \exp \left( -\frac{(z^t_{\ell}-\lambda(\NT+\sigma_t^2))^2}{2(\sigma_t^2+\NT)} \right).
 \end{align}
 Therefore, we have
 \begin{align}
 \mathsf{G}(z^t_{\ell},\sigma_t^2)= &\, K_1\left( \sigma_t^2+(z^t_{\ell}+\lambda \sigma_t^2)^2\right) Q\left( \frac{z^t_{\ell}+\lambda(\NT+\sigma_t^2)}{\sqrt{(\NT+\sigma_t^2)}} \right)\notag\\
 &+K_2\left( \sigma_t^2+(z^t_{\ell}-\lambda \sigma_t^2)^2\right) Q\left( \frac{-z^t_{\ell}+\lambda(\NT+\sigma_t^2)}{\sqrt{(\NT+\sigma_t^2)}} \right)\notag\\
 &+K_1 \left( \frac{\sigma_t^4 (z^t_{\ell}+\lambda(\NT+\sigma_t^2))}{\sqrt{2 \pi}(\NT+\sigma_t^2)^{\frac{3}{2}}} - \frac{2\sigma_t^2 (z^t_{\ell}+\lambda \sigma_t^2)}{\sqrt{2 \pi (\NT+\sigma_t^2)}} \right) \exp \left( -\frac{(z^t_{\ell}+\lambda(\NT+\sigma_t^2))^2}{2(\sigma_t^2+\NT)} \right)\notag\\
 &-K_2 \left( \frac{\sigma_t^4 (z^t_{\ell}-\lambda(\NT+\sigma_t^2))}{\sqrt{2 \pi}(\NT+\sigma_t^2)^{\frac{3}{2}}} - \frac{2\sigma_t^2 (z^t_{\ell}-\lambda \sigma_t^2)}{\sqrt{2 \pi (\NT+\sigma_t^2)}} \right) \exp \left( -\frac{(z^t_{\ell}-\lambda(\NT+\sigma_t^2))^2}{2(\sigma_t^2+\NT)} \right)\notag\\
 &- \mathsf{F}^2(z^t_{\ell},\sigma_t^2)
 \end{align}
By replacing $K_1$, $K_2$, $\mathsf{F}$ and using $\mathrm{erfcx}(x)$ instead of $Q$ for improved numerical stability, we obtain 
 \begin{align*}
 \mathsf{G}(z^t_{\ell},\sigma_t^2)= \sigma_t^2 + \lambda^2 \sigma_t^4 (1- \eta^2) - \frac{4}{\gamma} \frac{\lambda \sigma_t^4}{\sqrt{2\pi (\NT+\sigma_t^2)}}.
 \end{align*}
 where we define
 \begin{align}
 \gamma = \mathrm{erfcx}(\alpha)+\mathrm{erfcx}(\beta).
 \end{align}

 \section{Proof of \fref{lem:CS_convexity}}\label{app:CS_convexity}
 
 We know that $\bmx=\bms+\bme$, where $\bms$ and $\bme$ are two independent random variables and  both $p(\bme)$ and $p(\bms)$ are log-concave \cite{bagnoli2005log}.
 By using properties of log-concavity \cite[Sec. 3.5.2]{boyd2004convex}, the convolution of two log-concave functions, here $p(\bmx)=p(\bms)*p(\bme)$, is also log-concave. Thus, $-\log p(\bmx)$ is convex. Clearly, $\frac{1}{2\No}\|\bmy-\bH \bmx\|_2^2$ is also convex and their sum $q(\bmx)$ remains convex.
 
Now, we can compute the gradient of $q(\bmx)$, which is given by
 \begin{align}
 \nabla_\bmx q(\bmx)=-\nabla_\bmx \left[ \log p(\bmx) \right]+\frac{1}{\No}(\bH \bmx-\bmy)^\Tran \bH,
 \end{align}
where
 \begin{align}
-\nabla_\bmx \left[ \log p(\bmx) \right]=
\begin{bmatrix}
\frac{\partial \log p(\bmx)}{\partial x_1} \\ \vdots \\ \frac{\partial \log p(\bmx)}{\partial x_N} 
\end{bmatrix}.
 \end{align}
The missing piece is to compute $\frac{\partial \log p(\bmx)}{\partial x_i}$. From \fref{lem:FG_CS}, we have
\begin{align}
\log p(\bmx)=N \log \left( \frac{\lambda}{2} \right) +N \frac{\lambda^2 \NT}{2} + \sum_{i=1}^{N} \log \left[ e^{\lambda x_i} Q\left(\frac{x_i+\lambda \NT}{\sqrt{\NT}}\right) + e^{- \lambda x_i} Q\left(\frac{-x_i+\lambda \NT}{\sqrt{\NT}}\right) \right].
\end{align}
As a consequence, we have
\begin{align}
\frac{\partial \log p(\bmx)}{\partial x_i} &= \frac{\lambda e^{\lambda x_i} Q(\frac{x_i+\lambda \NT}{\sqrt{\NT}}) - \lambda e^{- \lambda x_i} Q(\frac{-x_i+\lambda \NT}{\sqrt{\NT}})}{e^{\lambda x_i} Q(\frac{x_i+\lambda \NT}{\sqrt{\NT}}) + e^{- \lambda x_i} Q(\frac{-x_i+\lambda \NT}{\sqrt{\NT}}) } \\
&= \lambda \frac{\mathrm{erfcx(\frac{x_i+\lambda \NT}{\sqrt{2\NT}})}-\mathrm{erfcx(\frac{-x_i+\lambda \NT}{\sqrt{2\NT}})}}{\mathrm{erfcx(\frac{x_i+\lambda \NT}{\sqrt{2\NT}})}+\mathrm{erfcx(\frac{-x_i+\lambda \NT}{\sqrt{2\NT}})}}\\
&= \lambda \eta(x_i,\tau=0),
\end{align}
where $\eta(z^t_{\ell},\tau)$ is defined in \fref{eq:eta}. 

%

\end{document}